\documentclass[12pt,english,reqno]{smfart}
\usepackage{mathptmx}
\usepackage{multirow}
\usepackage{tikz-cd}
\usetikzlibrary{arrows,chains,matrix,positioning,scopes,shapes.geometric,decorations.pathmorphing,
decorations,calc,trees,backgrounds,shapes,decorations.markings,fadings,decorations.pathreplacing,intersections}
\pdfpageattr{/Group <</S /Transparency /I true /CS /DeviceRGB>>}
\newcommand{\IQPV}{\rm IQPV}
\newcommand{\nm}{d_1}
\newcommand{\nmZ}{d_2}
\newcommand{\AI}{A{\rm I}}
\newcommand{\AII}{A{\rm I\!I}}
\newcommand{\AIII}{A{\rm I\!I\!I}}
\newcommand{\CI}{C{\rm I}}
\newcommand{\CII}{C{\rm I\!I}}
\newcommand{\BDI}{BD{\rm I}}
\newcommand{\DIII}{D{\rm I\!I\!I}}
\newcounter{seqnumb}[section]
\newtheorem{theorem}[seqnumb]{Theorem}
\newtheorem{lemma}[seqnumb]{Lemma}
\newtheorem{proposition}[seqnumb]{Proposition}
\newtheorem{corollary}[seqnumb]{Corollary}
\newtheorem{definition}[seqnumb]{Definition}
\theoremstyle{remark}
\newtheorem{remark}[seqnumb]{Remark}
\newtheorem{example}[seqnumb]{Example}
\renewcommand\thesection{\arabic{section}}

\renewcommand{\theequation}{\arabic{section}.\arabic{equation}}
\title[$\mathbb{Z}_2$ equivariant quasi-particle vacua]
{Bott periodicity for \\ $\mathbb{Z}_2$ symmetric ground states of \\ gapped free-fermion systems}
\author{R.\ Kennedy \& M.R. Zirnbauer}
\date{September 6, 2014}
\begin{document}
\maketitle

\begin{abstract}
Building on the symmetry classification of disordered fermions, we give a proof of the proposal by Kitaev, and others, for a ``Bott clock'' topological classification of free-fermion ground states of gapped systems with symmetries. Our approach differs from previous ones in that (i) we work in the standard framework of Hermitian quantum mechanics over the complex numbers, (ii) we directly formulate a mathematical model for ground states rather than spectrally flattened Hamiltonians, and (iii) we use homotopy-theoretic tools rather than $K$-theory. Key to our proof is a natural transformation that squares to the standard Bott map and relates the ground state of a $d$-dimensional system in symmetry class $s$ to the ground state of a $(d+1)$-dimensional system in symmetry class $s+1$. This relation gives a new vantage point on topological insulators and superconductors.
\end{abstract}

\section{Introduction}

In this article we address the following problem of mathematical physics. (We first formulate the mathematical problem as such, and then indicate its origin in physics.) Let there be a Hermitian vector space $W \equiv (\mathbb{C}^{2n}, \left\langle \cdot , \cdot \right\rangle)$ with $n$ a sufficiently large integer, and let $W$ have the additional structure of a non-degenerate symmetric complex bilinear form $\{ \cdot , \cdot \}$. Assume that $W$ carries an action by operators $J_1 , \ldots, J_s$ that satisfy the Clifford algebra relations
\begin{equation}\label{Eq:Cliff-s}
    J_l\, J_m + J_m\, J_l = - 2 \delta_{lm}\, \mathrm{Id}_W \quad (l,m = 1, \ldots, s)
\end{equation}
and preserve $\langle \, , \, \rangle$ as well as $\{ \, , \, \}$. Moreover, let $M$ be a $d$-dimensional manifold, namely momentum space or phase space, with an involution $\tau : \; M \to M$ whose physical meaning is momentum inversion. Our objects of interest then are rank-$n$ complex vector bundles $\pi : \; \mathcal{A} \to M$ whose fibers $A_k = \pi^{-1}(k) \subset W$ are constrained for all $k \in M$ by the conditions
\begin{equation}\label{Eq:VB-conds}
    0 = \{ A_k , A_{\tau(k)} \} = \left\langle A_k , J_1 A_k \right\rangle = \ldots = \left\langle A_k , J_s\, A_k \right\rangle .
\end{equation}
We refer to them as vector bundles of symmetry class $s$, or class $s$ for short. The goal is to give a homotopy classification for the ``classifying maps'' $k \mapsto A_k$ of such vector bundles. In the present paper we achieve this goal for the case of the $d$-sphere, $M = \mathrm{S}^d$, and for $n$ large enough relative to $d$. A companion paper \cite{KG14} deals with the case of $M = \mathrm{T}^d$.

\subsection{Motivation}

We now explain briefly how the posed problem arises from a situation of current interest in condensed matter physics. Our objects of study are systems of ``free fermions''; more precisely, systems of fermions described in the Hartree-Fock-Bogoliubov (HFB) mean-field approximation by any Hamiltonian which is quadratic in the operators creating or annihilating a single particle. The many-particle ground state of such a system is called an HFB mean-field ground state, or free-fermion ground state, or quasi-particle vacuum. The homotopy theory developed in the present paper addresses those cases where the Hamiltonian commutes with a group $\Gamma$ of translations in real space so that momentum is conserved. In such a situation, the HFB mean-field ground state is a product state factoring in the single-particle momentum $k \in \widehat{\Gamma} \equiv M$ and is determined uniquely by its collection $\{ A_k \}_{k \in M}$ of spaces of quasi-particle annihilation operators. (An element of the complex vector space $A_k$ is an operator that annihilates a fermion in a state with momentum $k$ or is an operator that creates a fermion in a state of momentum $-k \equiv \tau(k)$ or is any complex linear combination of these two types of operator.) The constant $n = \dim A_k$ equals the sum of the number of conduction and valence bands of the fermion system. We assume that the Hamiltonian has finite range in position space and is gapped, describing a band insulator or gapped superconductor. These assumptions ensure that the vector spaces $A_k$ depend continuously on the momentum $k$ and thus constitute a vector bundle, $\mathcal{A}$, over $M$.

The rank-$n$ complex vector bundles $\mathcal{A} \to M$ arising in this way come with some extra structure as formulated in (\ref{Eq:VB-conds}). Firstly, the condition $\{ A_k , A_{\tau(k)} \} = 0$ on $\tau$-opposite fibers expresses the fundamental property of Fermi statistics that all operators of a set of annihilation operators must have vanishing anti-commutators with one another. Secondly, the Hermitian orthogonality conditions in (\ref{Eq:VB-conds}) express the consequences of certain symmetries that constrain the Hamiltonian of the gapped system and hence translate into symmetries of the ground state. More precisely, for $s = 0$ there exist no symmetries (other than translations). For $s = 1$, the system has one anti-unitary symmetry, $T$, namely the operation of time reversal. If $\gamma$ denotes Hermitian conjugation, which is a complex anti-linear operation exchanging creation and annihilation operators, then $T$-invariance of the quasi-particle vacuum is concisely expressed by the condition $\left\langle A_k , J_1 A_k \right\rangle = 0$ with $J_1 = \gamma\, T$. We assume that our fermions have half-integer spin, so that $T^2 = - \mathrm{Id}$ and $J_1^2 = - \mathrm{Id}$. For $s = 2$, conservation of particle number enters as an additional, unitary symmetry. This is expressed by the operator $Q$ for charge or particle number. If the Hamiltonian commutes with $Q$, the quasi-particle vacuum is an eigenstate of $Q$. By conservation of momentum, this property can be expressed by the condition $\langle A_k , J_2\, A_k \rangle = 0$ with $J_2 = \mathrm{i} Q J_1$. For $s = 3$, we add another anti-unitary symmetry, namely particle-hole conjugation $C$, which leads to a third condition, $\langle A_k , J_3\, A_k \rangle = 0$. All operators $J_1, J_2, \ldots$ are unitary and preserve the bracket $\{ \, , \, \}$ encoding the canonical anti-commutation relations for Fock operators. Moreover, they obey the Clifford algebra relations (\ref{Eq:Cliff-s}). To go beyond $s = 3$, we observe that four pseudo-symmetries $J_1, J_2, J_3, J_4$ have the same effect on $A_k$ as the quaternion algebra of $\mathrm{SU}_2$ spin-rotation symmetries, by a result known as $(1,1)$ periodicity; see Section \ref{sect:setting} for the details. This should suffice for now to motivate the mathematical setting sketched in Eqs.\ (\ref{Eq:Cliff-s}) and (\ref{Eq:VB-conds}).

\subsection{Relation to previous work}

The goal of the present paper is to give a homotopy-theoretic classification of symmetry-protected topological phases of gapped free-fermion systems with symmetries as sketched above. The investigation of this classification problem was pioneered by Schnyder, Ryu, Furusaki, and Ludwig \cite{SRFL2008} who observed that there exist, in every space dimension, 5 symmetry classes (among the 10 classes of the ``Tenfold Way'' of disordered fermions \cite{AZ,HHZ}) that house topological insulators or superconductors robust to disorder. Building on this observation, Kitaev recognized the mathematical principle behind the emerging pattern, which he named the ``Periodic Table for topological insulators and superconductors'' \cite{kitaev}. He understood that the constraints due to physical symmetries can be formulated as an extension problem for the Clifford algebra (\ref{Eq:Cliff-s}) of what we propose to call ``pseudo-symmetries'', and he saw the close connection with a mathematical phenomenon known as Bott periodicity. He also advocated $K$-theoretic methods as a tool to compute the topological invariants characterizing the different symmetry-protected topological (SPT) phases. A pedagogical discussion of some points outlined by Kitaev was offered by Stone, Chiu, and Roy \cite{StoneEtAl}. Symmetry aspects were elaborated by Abramovici and Kalugin \cite{AK}.

A remarkable extension of the Bott-type periodicity phenomenon for free-fermion SPT phases was proposed by Teo and Kane \cite{TeoKane}, who introduced position-like dimensions (associated with defects) in addition to the momentum-like dimensions considered in earlier work. Freedman, Hastings, Nayak, Qi, Walker, and Wang \cite{FHNQWW} developed this idea further and pointed out that their results lead to a mathematical proof of Kitaev's Periodic Table if one assumes (with referral to unpublished notes by Kitaev) that gapped lattice Hamiltonians are stably equivalent to Dirac Hamiltonians with a spatially varying mass term. Fidkowski and Kitaev \cite{FidkowskiKitaev} gave a complete classification of one-dimensional systems including interactions. A recent treatise on the subject at large is by Freed and Moore \cite{FreedMoore}, who set up a comprehensive framework based on the Galilean group and review the relevant notions of twisted equivariant $K$-theory, an algebraic variant of which is treated in \cite{Thiang}.

Let us now highlight the main differences between our work and the current literature. Firstly, to the extent that only the static properties (as opposed to the dynamical response) of the physical system are under investigation, the classification problem at hand is a problem of classifying \emph{ground states} -- that, in any case, is how we view it. Thus we never make any direct reference to a Hamiltonian. (Aside from a locality condition to ensure the continuity of the vector bundle $\mathcal{A} \to M$, the only information we need about the Hamiltonian is its symmetry class, as this determines the symmetry class of the ground state.) In particular, there will be no need for any process of ``flattening'' the Hamiltonian in this paper.

Secondly, all our symmetries are \emph{true} symmetries in the sense that they \emph{commute} with the Hamiltonian and leave the ground state invariant. ``Symmetries'' that anti-commute with the Hamiltonian (such as chirality for the massless Dirac operator) do not appear in our work.

Thirdly, a crucial element of our approach is that we work over the complex number field throughout. As a matter of fact, the vector bundles singled out by the constraints (\ref{Eq:VB-conds}) are \emph{complex}, i.e.\ their fibers are complex vector spaces. While some of them can be regarded as Real vector bundles in the sense of Atiyah \cite{atiyah}, others cannot be. Moreover, although the operator of Hermitian conjugation does single out a real (or Majorana) subspace $\mathbb{R}^{2n} \subset \mathbb{C}^{2n} = W$, one of our discoveries is that one should keep this real structure flexible in order to attain the best overall perspective. (In fact, the formulation and proof of our results employs two different notions of taking the complex conjugate!)

Finally, and most importantly, our work differs from the work of other authors by the principle of topological classification used. Starting with Kitaev \cite{kitaev}, most of the past and present literature has relied on the algebraic tools of $K$-theory to compute the topological invariants given by stable isomorphism classes of vector bundles. An exception is the approach in \cite{deNittis-AI,deNittis-AII}, where ordinary (as opposed to stable) isomorphism classes of vector bundles are computed for two of the ten symmetry classes. (These are the classes $A$I and $A$II, which are special in that they permit a description of ground states by Real and Quaternionic vector bundles, respectively. In the present paper we will encounter vector bundles of a more general kind.) In contradistinction, the present work uses homotopy-theoretic methods to establish a \emph{homotopy} classification for the classifying maps of the vector bundles $\mathcal{A} \to M$. It has to be emphasized that the equivalence relation of homotopy is finer than that of (ordinary or stable) isomorphisms of vector bundles: ordinary isomorphism classes are recovered for a large number of valence bands, while stable isomorphism classes are recovered if both the valence and the conduction bands are large in number.

\subsection{Results}

The following is a summary of the progress made in the present paper.

Our first result is a demonstration from first principles as to why the eight ``real'' symmetry classes of the Tenfold Way are to be put in the particular sequence featured by Kitaev's Periodic Table. Without invoking the after-the-fact reason of Bott-type periodicity, we build up Kitaev's sequence recursively by imposing true physical symmetries in a distinguished order. In this recursive process, each physical symmetry translates to a so-called pseudo-symmetry $J_l\,$, which in turn gives rise to one of the $s$ conditions $\langle A_k , J_l\, A_k \rangle = 0$ of Eq.\ (\ref{Eq:VB-conds}).

Our second result is an invariant and universal description of a map put forward by Teo and Kane \cite{TeoKane}, taking a ground state of (symmetry) class $s$ in $d$ dimensions and turning it into a ground state of class $s+1$ in $d+1$ dimensions. To state this more precisely, let $C_s(n)$ denote the classifying space for vector bundles of class $s$ with $W = \mathbb{C}^{2n}$. For each $s$, the bilinear form $\{ \, , \, \}$ determines an involution $\tau_s$ that sends any $A \in C_s(n)$ to its annihilator $A^\perp \in C_s(n)$ given by $\{ A , A^\perp \} = 0$. The first equation in (\ref{Eq:VB-conds}) is then restated as the condition $\tau_s (A_k) = A_{ \tau(k)}$, which can be interpreted as saying that the classifying map $k \mapsto A_k$ is equivariant with respect to a group $\mathbb{Z}_2$ whose non-trivial element acts on $M$ and $C_s(n)$ by $\tau$ and $\tau_s$ respectively.

Now, following the original paper of Bott \cite{Bott1959} we assign to every point $A$ of $C_s(n) \simeq C_{s+2}(2n)$ a minimal geodesic $[0,1] \ni t \mapsto \beta_t(A)$ joining some distinguished point of $C_{s+1} (2n)$ to its antipode. The operation of forming the geodesic can be concatenated with the classifying map $k \mapsto A_k\,$, and it preserves $\mathbb{Z}_2 $-equivariance in the sense that $\tau_{s+1}(\beta_t(A)) = \beta_{1-t} (\tau_s(A))$. The final outcome of this invariant and universal construction is what we call the ``diagonal map'', taking a vector bundle $\mathcal{A} \to M$ of class $s$ and transforming it into another vector bundle $\tilde{\mathcal{A}} \to \tilde{S} M$ of class $s+1$, where $\tilde{S} M$ denotes the momentum-type suspension of $M$. The diagonal map is, in a certain sense, a ``square root'' of the original Bott map. Indeed, a key step of our treatment is to take the square of the diagonal map and show that the outcome, properly understood, is the Bott map.

The diagonal map induces a mapping in homotopy --- to be precise: a mapping between homotopy classes of base-point preserving and $\mathbb{Z}_2 $-equivariant classifying maps; or, in formulas: from $[M, C_s(n)]_\ast^{ \mathbb{Z}_2}$ to $[\tilde{S} M, C_{s+1}(2n) ]_\ast^{ \mathbb{Z}_2}\,$; or, physically speaking: between SPT phases of gapped free fermions in adjacent symmetry classes and dimensions.

Our third and main achievement is a homotopy-theoretic proof that this map is bijective under favorable conditions. The precise statements are laid down in Theorems \ref{thm:7.1} and \ref{thm:7.2}. To give a quick summary, let $M$ be a path-connected $\mathbb{Z}_2$-CW complex with base point $k_\ast$ fixed by the $\mathbb{Z}_2$-action, and let $A_\ast \in C_s(n) \subset C_{s-1}(n) \simeq C_{s+1}(2n)$ be a target-space base point also fixed by the $\mathbb{Z}_2$-action. Then if $\dim M \ll n$ there exist two bijections,
\begin{equation}\label{eq:main-result}
    [S M, C_{s-1}(n) ]_\ast^{\mathbb{Z}_2} \simeq [M, C_s(n)]_\ast^{ \mathbb{Z}_2} \simeq [\tilde{S} M, C_{s+1}(2n) ]_\ast^{ \mathbb{Z}_2} \,,
\end{equation}
between homotopy classes of base-point preserving and $\mathbb{Z}_2$-equivariant maps. The left one, where $S M$ denotes the usual suspension (which is position-type, i.e., the space direction added in going from $M$ to $S M$ is acted upon trivially by the extension of the involution $\tau$), follows rather directly from the Bott periodicity theorems, by employing the $G$-Whitehead Theorem for $G = \mathbb{Z}_2$ in order to transcribe these classical results to our setting.

The right bijection is more difficult to establish. While a direct approach using Morse theory might be feasible, our strategy of proof is to establish a connection with the classical Bott isomorphism relating the stable homotopy groups of symmetric spaces. For this we first consider the special set of symmetry indices $s \in 4\mathbb{N}-2$. In these cases, there exists a certain fibration that connects our diagonal map with the standard Bott map; thus, the right bijection in (\ref{eq:main-result}) follows from the left bijection by an isomorphism due to the projection map $p$ of the fibration. (Applying $p$ amounts to taking a square, which is the squaring operation alluded to earlier.) Sadly, for $s \notin 4\mathbb{N}-2$ the said fibration is not available for finite $n$, although it does exist in the $K$-theory limit of infinite $n$ \cite{Giffen}. Therefore, we need to employ an additional argument to complete the proof. Adapting an idea of Teo and Kane \cite{TeoKane}, we make use of generalized momentum spaces $(M, \tau)$ with position-like as well as momentum-like directions. We then use the left bijection in Eq.\ (\ref{eq:main-result}) to dial the symmetry index $s$ to a value where the isomorphism by $p$ applies. The desired result now follows from $S^r \tilde{S} M = \tilde{S} S^r M$.

If $M$ is a sphere carrying any $\mathbb{Z}_2$-action, the bijections in Eq.\ (\ref{eq:main-result}) reproduce the result of \cite{TeoKane}, which in turn generalizes the Periodic Table presented in \cite{kitaev}.

In the results stated above, we gave ourselves the luxury of making the simplifying assumption $d = \dim M \ll n$. Our final result are practically useful bounds on $d$ (as a function of $n$) delineating the region of validity of Kitaev's Periodic Table and its generalization by Teo and Kane. The method of derivation used is stable inclusion of symmetric spaces.

\subsection{Plan of paper}

This paper is organized as follows. In Section \ref{sect:setting} we set up the vector-bundle description of translation-invariant ground states of gapped free-fermion systems for all symmetry classes, starting with class $s = 0$ (no symmetries; also known as class $D$) and increasing the number of (pseudo-)symmetries up to $s = 7$. The passage from vector bundles to an equivalent description by classifying maps is made in Section
\ref{sect:ClassMaps}. There we also give a number of examples illustrating the difference between the topological classification by homotopy classes of classifying maps, isomorphism classes of vector bundles, and the stable equivalence of $K$-theory. In Section \ref{sect:diag-map} we formulate the diagonal map and illustrate it at the special example of making the steps from $(s=0,d=0)$ to $(1,1)$ and further to $(2,2)$.

The role of Section \ref{sect:5} is to collect the results of homotopy theory relevant to our problem. We state the $G$-Whitehead Theorem and recall the classical Bott periodicity theorem in the complex and real settings. We also exploit the property of $\mathbb{Z}_2$-equivariance to reformulate homotopy as relative homotopy. By using all this information, we prove in Section \ref{sect:6} that the diagonal map induces a bijection in homotopy for the symmetry classes $s = 2$ and $s = 6$. In Section \ref{sect:7} we extend and complete the argument so as to cover all classes $s$. The final Section \ref{sect:8} presents the precise bounds on stability.

\section{From symmetries to vector bundles}\label{sect:setting}
\setcounter{equation}{0}

We begin with some notation and language. A quasi-particle vacuum (or free-fermion ground state, or Hartree-Fock-Bogoliubov mean-field ground state) is a state in Fock space which is annihilated by a set of (quasi-)particle annihilation operators. Two well-known examples are Hartree-Fock ground states, which have a definite particle number, and the paired states of the BCS theory of superconductivity.

To give a precise description of such many-fermion ground states, we set out from the formalism of second quantization. We assume that our translation-invariant physical system (with momentum space $M$) is built from a unit cell of Hilbert space dimension $n$. Single-particle states are then characterized by their momentum $k \in M$ and a band index $j = 1, \ldots, n$. The single-particle creation and annihilation operators (denoted by $c_{k,j}^\dagger$ resp.\ $c_{k,j}^{\vphantom{\dagger}}$ and called Fock operators for short) obey the canonical anti-commutation relations
\begin{equation}\label{eq0:CAR}
    \begin{split}
    &c_{k,\,i}^{\vphantom{\dagger}} \, c_{k^\prime,\,j}^{ \vphantom{\dagger}} + c_{k^\prime,\,j}^{\vphantom{\dagger}}\, c_{k,\,i}^{\vphantom{\dagger}} = 0, \quad c_{k,\,i}^\dagger \, c_{k^\prime,\,j}^\dagger + c_{k^\prime,\,j}^\dagger \, c_{k, \,i}^\dagger = 0 , \cr &c_{k,\,i}^\dagger \, c_{k^\prime,\,j}^{ \vphantom{\dagger}} + c_{k^\prime,\,j}^{\vphantom{\dagger}} \, c_{k,\,i}^\dagger = \delta_{ij} \, \delta(k-k^\prime) .
    \end{split}
\end{equation}
Organizing Fock operators by the momentum quantum number, we define $W_k$ as the vector space spanned by the Fock operators that lower the momentum by $k$. Thus
\begin{equation}\label{eq:2.2}
    W_k = U_k \oplus V_{-k}
\end{equation}
where $U_k$ is the space of single-particle annihilation operators for momentum $k$, while $V_{-k}$ is the space of single-particle creation operators for momentum $-k$.

From now on, we are going to denote the operation of inverting the momentum by
\begin{equation}
    \tau : \; M \to M , \quad k \mapsto -k .
\end{equation}
This is done in order to prepare the ground for a later modification of the involution $\tau$. (In fact, we will eventually consider ``momentum'' spaces $M$ where some of the components of $k$ are \emph{position}-like instead of momentum-like.) In the present section, we always have $\tau(k) \equiv -k$, and we will take the liberty of frequently writing $-k$ instead of $\tau(k)$ for better clarity of the notation.

In terms of the basis $c_{k,j}\,, c_{k,j}^\dagger\,$, the elements $\psi \in W_k$ are expressed as
\begin{equation}\label{eq:2.3}
    \psi = \sum_{j=1}^n \big( u_j^{\vphantom{\dagger}}\, c_{k,\,j}^{ \vphantom{\dagger}} + v_j^{\vphantom{\dagger}}\, c_{-k,\,j}^\dagger \big) \in U_k \oplus V_{\tau(k)} \,
\end{equation}
with coefficients $u_j\,, v_j \in \mathbb{C}$. The vector spaces $W_k$ are complex, and they all have the same dimension $2n$ independent of $k$. In fact, they are canonically isomorphic (by unitary momentum-boost operators taken from the Heisenberg group), and we often omit the momentum quantum number and write simply $W_k \equiv W \equiv \mathbb{C}^{2n}$. The number $n$ is referred to as the (total) number of (valence and conduction) bands. One may think of the collection of vector spaces $\{ W_k \}_{k \in M}$ as a complex vector bundle, say $\mathcal{W}$, over the momentum space $M$. This bundle is trivial in our setting: $\mathcal{W} \simeq M \times W$. It could, however, be non-trivial in a low-energy effective theory where the bands far from the Fermi surface have been discarded. In any case, $\mathcal{W}$ has non-trivial subvector bundles, and these are the objects of our interest.

We now highlight some important structures on the vector spaces $W_k\,$. First of all, the canonical anti-commutation relations (CAR) for fermionic Fock operators induce for all $k \in M$ a pairing between $W_k$ and $W_{\tau(k)}\,$, i.e.\ a non-degenerate bilinear form
\begin{equation}\label{Eq:CAR}
    \{ \cdot , \cdot \} : \quad W_{\tau(k)} \otimes W_k \to \mathbb{C} ,
\end{equation}
by dropping the $\delta$-function $\delta(k-k^\prime)$ in Eq.\ (\ref{eq0:CAR}). This pairing has the property of being symmetric for $\tau$-invariant momenta $\tau(k) = k$. We refer to it as the CAR pairing. Expressing $\psi \in W_{\tau(k)}$ and $\psi^\prime \in W_k$ as in Eq.\ (\ref{eq:2.3}) we have
\begin{equation}
    \{ \psi , \psi^\prime \} = \sum_{j = 1}^n ( u_j^{\vphantom{\prime}} \, v_j^\prime + v_j^{\vphantom{\prime}} \, u_j^\prime ).
\end{equation}

Next, Fock space comes equipped with a Hermitian scalar product, which determines an operation of Hermitian conjugation. Since Hermitian conjugation in Fock space takes operators that remove momentum $k$ into operators that create momentum $k$, it induces a complex anti-linear involution
\begin{equation}
    \gamma : \quad W_k \to W_{\tau(k)} \quad (\gamma^{\,2} = \mathrm{Id})
\end{equation}
for all $k \in M$. By combining this $\gamma$-operation with the CAR pairing between $W_{\tau(k)}$ and $W_k$, we get a Hermitian scalar product on each vector space $W_k$:
\begin{equation}
    \left\langle \cdot , \cdot \right\rangle : \quad W_k \times W_k \to \mathbb{C} .
\end{equation}
Its expression in components is
\begin{equation}\label{eq-mz:2.9}
    \left\langle \psi , \psi^{\,\prime} \right\rangle := \{\gamma\psi, \psi^{\,\prime}\} = \sum\nolimits_j \big(\bar{u}_j^{\vphantom{\prime}} \, u^{\,\prime}_j + \bar{v}_j^{\vphantom{\prime}}\, v^{\,\prime}_j \big) .
\end{equation}

In summary, the set $\{ W_k \}_{k \in M}$ is a trivial bundle $\mathcal{W}$ of canonically isomorphic Hermitian vector spaces $W_k \equiv W \equiv \mathbb{C}^{2n}$. It comes with the extra structure given by the pairing (\ref{Eq:CAR}).

We are now in a position to formalize the type of free-fermion or Hartree-Fock-Bogoliubov mean-field ground state addressed in the present paper. In the following definition, the abbreviation $\IQPV$ stands for a \underline{q}uasi-\underline{p}article \underline{v}acuum with the property of being the trans\-lation-\underline{i}nvariant ground state of an \underline{i}nsulator (or gapped system).
\begin{definition}\label{Eq:Def-FundBund}
By an $\IQPV$ we mean a complex subvector bundle $\mathcal{A} \stackrel{\pi}{\to} M$ of fibers $\pi^{-1}(k) \equiv A_k \subset W_k = \mathbb{C}^{2n}$ of dimension $n$ such that all pairs $A_k \,, A_{\tau(k)}$ of $\tau$-opposite fibers annihilate one another with respect to the CAR pairing:
\begin{equation}\label{Eq:pair0}
    \forall k \in M : \quad \{ A_{\tau(k)} , A_k \} = 0 .
\end{equation}
\end{definition}
\begin{remark}
Physically speaking, the vector space $A_k \subset W_k$ is spanned by the quasi-particle operators of momentum $k$ which annihilate the quasi-particle vacuum. The Fock space description $\vert \IQPV \rangle$ of the quasi-particle vacuum is recovered \cite{BalianBrezin} by choosing a basis $\widetilde{c}_1(k), \ldots, \widetilde{c}_n(k)$ of $A_k$ for each $k$ and applying their product to a suitable reference state:
\begin{equation*}
    \vert \text{IQPV} \rangle := \prod\nolimits_k \widetilde{c}_1(k) \cdots \widetilde{c}_n(k) \, \vert \text{ref} \rangle .
\end{equation*}
The condition (\ref{Eq:pair0}) expresses the fact that all annihilation operators must have vanishing anti-commutators with each other. We refer to (\ref{Eq:pair0}) as the Fermi constraint.
\end{remark}
\begin{remark}\label{rem:2.2}
There exist two different notions of orthogonality on our universal vector space $W = \mathbb{C}^{2n}$. Firstly, given the CAR pairing (or bracket) $\{ \, , \, \}$, any complex linear subspace $L \subset W$ determines a complex linear subspace $L^\perp \subset W$ by
\begin{equation}
    L^\perp = \{ \psi \in W \mid \forall \psi^\prime \in L : \; \{ \psi , \psi^\prime \} = 0 \} .
\end{equation}
We will often use the $\perp$-operation to express the Fermi constraint (\ref{Eq:pair0}) as $A_k^\perp = A_{\tau(k)}$ (for all $k \in M$).
Secondly, given the Hermitian structure $\langle \; , \, \rangle$, the orthogonal complement $L^\mathrm{c}$ of $L$ is defined by
\begin{equation}
    L^\mathrm{c} = \{ \psi \in W \mid \forall \psi^\prime \in L : \; \langle \psi , \psi^\prime \rangle = 0 \} .
\end{equation}
For present use, we note that the two notions of orthogonality are connected by
\begin{equation}\label{eq:compatible}
    \gamma L^\perp = L^\mathrm{c} ,
\end{equation}
as a consequence of the relation $\langle \gamma \psi , \psi^\prime \rangle = \{ \psi , \psi^\prime \}$.
\end{remark}

In the remainder of this section we will impose various symmetries which
centralize the translation group: first time reversal $T$; then particle number $Q$; then particle-hole conjugation $C$; and so on. The optimal order in which to arrange these symmetries was first understood by Kitaev \cite{kitaev}; we therefore call it the Kitaev sequence.

All of the symmetries $T$, $Q$, $C$, etc., will have the status of \emph{true} symmetries (i.e., they \emph{commute} with the Hamiltonians of the appropriate symmetry class; never do they anti-commute). In particular, our operator $C$ of particle-hole conjugation commutes with a particle-hole symmetric Hamiltonian:
\begin{equation*}
    H = C H\, C^{-1} .
\end{equation*}
We emphasize this systematic and rigid feature, as it sets our approach apart from what is usually done in the current literature, with notable exceptions being \cite{AK,FHNQWW}.

The resulting free-fermion ground states with symmetries all turn out to fit neatly into the following mathematical framework. Please be advised that the process of implementing the framework will convert true physical symmetries into ``pseudo-symmetries''.
\begin{definition}\label{def:2.2}
By an $\IQPV$ {\rm of class} $s$ ($s = 0, 1, 2, \ldots$) we mean a rank-$n$ complex subvector bundle $\mathcal{A} \stackrel{\pi}{\to} M$ as described in Def.\ \ref{Eq:Def-FundBund} but with the fibers $\pi^{-1}(k) \equiv A_k \subset W \simeq \mathbb{C}^{2n}$ constrained by the pseudo-symmetry conditions
\begin{equation}\label{eq:2.7}
    \forall k \in M : \quad J_1 A_k = \ldots = J_s\, A_k = A_k^\mathrm{c}\,,
\end{equation}
where the complex linear operators $J_1, \ldots, J_s : \; W \to W$ satisfy the Clifford algebra relations (\ref{Eq:Cliff-s}) and each operator $J_l$ ($l = 1, \ldots, s$) is an orthogonal unitary transformation of $W$.
\end{definition}
\begin{remark}
We speak of pseudo-symmetries because the $J_1, \ldots, J_s$ send $A_k$ to its orthogonal complement $A_k^\mathrm{c}\,$, whereas true (unitary) symmetries would leave $A_k$ invariant. For $s = 0$ the conditions (\ref{eq:2.7}) are understood to be void.
\end{remark}
\begin{remark}\label{remark:orthogonalunitary}
An orthogonal unitary transformation $J$ of $W$ is a $\mathbb{C}$-linear operator with the properties
\begin{equation*}
    \langle J \psi , J \psi^\prime \rangle =
    \langle \psi , \psi^\prime \rangle \quad \text{and} \quad
    \{ J \psi , J \psi^\prime \} = \{ \psi , \psi^\prime \}
\end{equation*}
for all $\psi, \psi^\prime \in W$. The condition $J \, A_k = A_k^\mathrm{c}$ is equivalent to $\left\langle A_k \,, J \, A_k \right\rangle = 0$; cf.\ Eq.\ (\ref{Eq:VB-conds}). It is also equivalent to the condition $H(k)\, J + J\, H(k) = 0$ for
\begin{equation}
    H(k) = - \Pi_{A_k} + \Pi_{A_k^\mathrm{c}} \,.
\end{equation}
The operator $H(k)$ is commonly referred to as the \textit{flattened Hamiltonian}, as it may be viewed as a Hamiltonian with energies $\pm 1$ independent of $k$. It is a unitary transformation which is not orthogonal in general, but rather satisfies
\begin{equation}
    \{ H(k) \psi , H(-k) \psi^\prime \} = \{ \psi , \psi^\prime \}
\end{equation}
for all $\psi, \psi^\prime \in W$. The notion of flattened Hamiltonian is used in \cite{kitaev,StoneEtAl}, along with an orthonormal basis of $W$ consisting of $\gamma$-fixed vectors. In this ``Majorana'' basis, all orthogonal unitary transformations are expressed as real orthogonal matrices.
\end{remark}
\begin{remark}
Based on the Kitaev sequence, Definition \ref{def:2.2} arranges for the $\IQPV$s of class $s+1$ to be contained in those of class $s$. The existence of such an inclusion has invited attempts \cite{StoneEtAl} to transcribe the  classical result of real Bott periodicity \cite{Bott1959,milnor} so as to derive the desired homotopy classification. In the present paper we pick up on this attempt and show that it can be brought to fruition by invoking additional information.
\end{remark}

As a final remark, let us elaborate on a comment made in the introductory section. In our setting and language, a Real vector bundle in the sense of Atiyah \cite{atiyah}, or Quaternionic vector bundle in the sense of \cite{deNittis-AII}, would be a complex vector bundle $\mathcal{A} \to M$ with a $\mathbb{C}$-anti-linear projective involution ($T^2 = \pm 1$) mapping the fiber over $k$ to the fiber over $\tau(k) = - k$. Our vector bundles are \emph{not} of this kind in general. Indeed, for $s = 0$ we do have the $\perp$-operation determining the vector space $A_{\tau(k)} = A_k^\perp$ as the annihilator space of $A_k\,$, yet there exists no canonical map taking the individual vectors in $A_k$ to vectors in $A_{\tau(k)}$.

Table \ref{table:symmetries} gives a quick summary of the systematic structure developed in the remainder of this Section (\ref{sect:D}--\ref{sect:BDI}). Readers who are prepared to take the systematics for granted may want to take a look at Section \ref{sect:s-geq-3} and then proceed directly to Section \ref{sect:ClassMaps}.

\begin{table}
\begin{center}
\begin{tabular}{l|l|l|l}
class   &symmetries &$s$    &pseudo-syms.\\
\hline
$D$     &none		&$0$		&Fermi constraint\\
$\DIII$     &$T$ (time reversal)		&$1$		&$J_1 = \gamma\, T$\\
$\AII$     &$T, Q$ (charge)		&$2$	&$J_2 = \mathrm{i} \gamma\, T Q$ \\
$\CII$     &$T, Q$, $C$ (ph-conj.)		&$3$		&$J_3 = \mathrm{i} \gamma \, C Q$\\
\hline
$C$     &$S_1$, $S_2$, $S_3$ (spin rot.)		&$4$		&see text\\
$\CI$     &$S_1$, $S_2$, $S_3$, $T$		&$5$		&\\
$\AI$     &$S_1$, $S_2$, $S_3$, $T$, $Q$ 		&$6$		&\\
$\BDI$     &$S_1$, $S_2$, $S_3$, $T$, $Q$, $C$	&$7$		&
\end{tabular}
\end{center}
\vspace{10pt}
\caption{Overview of the symmetries and the corresponding pseudo-symmetries that are to be introduced sequentially in the present section.} \label{table:symmetries}
\end{table}

\subsection{Class $s = 0$ (alias $D$)}\label{sect:D}

The first symmetry class to consider is that of $s = 0$. This class is realized by gapped superconductors or superfluids with no symmetries (other than translations); it is commonly referred to as class $D$. By Definition \ref{Eq:Def-FundBund} an $\IQPV$ of class $s = 0$, or translation-invariant free-fermion ground state of a gapped system in symmetry class $D$, is a vector bundle $\mathcal{A} \to M$ with fibers $A_k \subset W \simeq \mathbb{C}^{2n}$ that are complex $n$-dimensional vector spaces subject to the Fermi constraint (\ref{Eq:pair0}) or, equivalently,
\begin{equation}\label{eq:2.8}
    \forall k \in M : \quad A_k^\perp = A_{\tau(k)} ,
\end{equation}
see Remark \ref{rem:2.2}. As will be explained in Section \ref{sect:ClassMaps}, there exists an alternative description of such a vector bundle by a so-called classifying map.

We seize this opportunity to make two comments. For one, the literature on the subject often construes the relation (\ref{eq:2.8}) (or rather, its consequences for the Hamiltonian) as a ``particle-hole symmetry'', although it is actually no more than a fundamental constraint dictated by Fermi statistics -- a point forcefully made in \cite{HHZ}. Note especially that no anti-unitary or complex anti-linear operations are involved in (\ref{eq:2.8}).

Our second comment concerns the language used. Borrowing Cartan's notation for symmetric spaces, the terminology for symmetry classes of disordered fermions was introduced in \cite{AZ}. This was done in the context of mesoscopic metals and superconductors where translation invariance is broken by the presence of disorder. A good fraction of the condensed matter community has adopted the same terminology for the related, but different purpose of classifying translation-invariant ground states (instead of disordered Hamiltonians). This is suboptimal but probably beyond rectification given the developed state of the research field. It is suboptimal because a dictionary is required for the non-expert to translate the terminology into the pertinent mathematics. For example, an $\IQPV$ of class $D$ is determined (see Section \ref{sect:ClassMaps} below) by a $\mathbb{Z}_2$-equivariant mapping $M \to \mathrm{Gr}_n( \mathbb{C}^{2n})$ that maps the $\tau$-fixed points of $M$ to $\mathrm{O}_{2n} / \mathrm{U}_n$ -- a
symmetric space not of type $D$ but of type $\DIII$.

\begin{example}
Consider a single band ($n = 1$) of spinless fermions in one dimension with ground state
\begin{equation*}
    \mathrm{e}^{\; \sum_k z(k)\, c_{-k}^\dagger c_k^\dagger} \, \vert 0 \rangle \propto \prod\nolimits_k \big( u(k) + v(k) c_{-k}^\dagger c_k^\dagger \big) \vert 0 \rangle , \quad z(k) = v(k) / u(k) ,
\end{equation*}
where $z(k) \in \mathbb{C} \cup \{ \infty \}$ and $\vert 0 \rangle$ is the Fock vacuum. This state is annihilated for any $k$ by the quasi-particle operator $u(k) c_k + v(k) c_{-k}^\dagger$. Thus we may regard it as a vector bundle $\mathcal{A} \to M$ with fibers
\begin{equation*}
    A_k = \mathbb{C}\cdot \big( u(k)c_k + v(k)c_{-k}^\dagger \big).
\end{equation*}
The Fermi constraint $A_k^\perp = A_{\tau(k)}$ translates to $u(k) v(-k) + v(k) u(-k) = 0$ or $z(-k) = - z(k)$. For $\tau$-invariant momenta $k_0 = \tau(k_0) = - k_0$ it follows that either $u(k_0)$ or $v(k_0)$ must vanish; hence $z(k_0 = - k_0)$ is either zero or infinite.
\end{example}

\subsection{Class $s = 1$ (alias $\DIII$)}\label{sect:DIII}

We now impose the first symmetry (beyond translations), by requiring that our quasi-particle vacua are invariant under the anti-unitary operator $T$ which reverses the time direction. More precisely, we assume time-reversal symmetry for fermions with half-integer spin, so that $T^2 =-\mathrm{Id}$. (Although $T$ is fundamentally defined on the single-particle Hilbert space and then on Fock space, $T$ here denotes the induced action on single-particle creation and annihilation operators.) The resulting symmetry class is commonly called $\DIII$; it is realized, for example, by superfluid ${}^3{\rm He}$ in the B-phase.

Because time reversal inverts the momentum, it gives us a mapping
\begin{equation}
    T : \; W_k \to W_{\tau(k)} \,,
\end{equation}
which is actually a pair of maps $T : \; U_k \to U_{\tau(k)}$ and $T : \; V_{\tau(k)} \to V_k\,$. This pair is compatible with the CAR pairing (\ref{Eq:CAR}) in the sense that
\begin{equation}
    \{ T \psi , T \psi^\prime \} = \overline{ \{ \psi , \psi^\prime \} }.
\end{equation}
Notice that $T^2 = - \mathrm{Id}$ requires $n$ to be even.

The quasi-particle vacuum encoded in a vector bundle $\mathcal{A} \to M$ is time-reversal invariant if the quasi-particle annihilation operators at momentum $k$ are transformed by $T$ into quasi-particle annihilation operators at momentum $-k = \tau(k)$, i.e.,
\begin{equation}\label{eq:TR-inv}
    T A_k = A_{\tau(k)} .
\end{equation}
Note that even though an anti-unitary operation $T$ is now involved,
the fibers $A_k \subset W$ of the vector bundle $\mathcal{A}$ are still complex -- and it is not useful to fix any real subspace $W_\mathbb{R} \subset W$, as $W_\mathbb{R}$ would have to be (re-)polarized to accommodate the complex vector spaces $A_k\,$.

To bring (\ref{eq:TR-inv}) in line with Eq.\ (\ref{eq:2.7}) of Definition \ref{def:2.2}, we observe that the anti-unitary operator $T$ commutes with the operation $\gamma$ of Hermitian conjugation of Fock operators. Thus by concatenating $T$ with the $\gamma$-operation, we get a complex linear operator
\begin{equation}
    J_1 : \; W_k \to W_k \,, \quad \psi \mapsto (T\circ \gamma) \psi = (\gamma \circ T) \psi ,
\end{equation}
which has square $J_1^2 = - \mathrm{Id}$ since $T^2 = - \mathrm{Id}$ and $\gamma^{\,2} = \mathrm{Id}$. It is easy to check that
\begin{equation*}
    \langle J_1 \psi , J_1 \psi^\prime \rangle = \langle \psi , \psi^\prime \rangle \quad \text{and} \quad \{ J_1 \psi , J_1 \psi^\prime \} = \{\psi ,\psi^\prime \}
\end{equation*}
for all $\psi , \psi^\prime \in W$. Thus $J_1$ is an orthogonal unitary transformation of $W$. Moreover, the true symmetry condition (\ref{eq:TR-inv}) is equivalent to the pseudo-symmetry condition
\begin{equation*}
    J_1 A_k = A_k^\mathrm{c} \,,
\end{equation*}
since we have $\gamma\, A_{\tau(k)} = A_k^\mathrm{c}$ from $A_{\tau(k)} = A_k^\perp$ and the relation (\ref{eq:compatible}). Hence the translation-invariant free-fermion ground state of a gapped superconductor or superfluid in symmetry class $\DIII$ is precisely modeled by an $\IQPV$ of class $s = 1$ in the sense of Definition \ref{def:2.2}. A one-dimensional example of such a ground state is given in Section \ref{sect:4.3}.

\subsection{Class $s = 2$ (alias $\AII$)}\label{sect:AII}

Imposing another symmetry (beyond translation and time-reversal invariance), we now require that our quasi-particle vacua be compatible with the global $\mathrm{U}(1)$ gauge symmetry underlying the law of charge conservation (which is the same as conservation of particle number if all particles carry the same quantum of charge). The resulting symmetry class is commonly called $\AII$. It is realized in band insulators and it hosts, in particular, the so-called quantum spin Hall insulator.

Recall from (\ref{eq:2.2}) the decomposition $W_k = U_k \oplus V_{\tau(k)}$ by particle annihilation and creation operators. The operator $Q$ for charge (or particle number) acts on $U_k \subset W_k$ as $-1$ and on $V_{\tau(k)} \subset W_k$ as $+1$. We say that a quasi-particle vacuum conserves charge (or has fixed particle number) if it is invariant under the action of the $\mathrm{U}(1)$ gauge group of operators $\mathrm{e}^{ \mathrm{i}\theta Q}$; in that case, we prefer to call the quasi-particle vacuum a Hartree-Fock (mean-field) ground state. Since the invariance of a vector space under a one-parameter group is equivalent to the invariance under its generator, we have
\begin{equation}
    \forall k \in M : \quad Q \, A_k = A_k \,.
\end{equation}

To bring this symmetry condition in line with Definition \ref{def:2.2}, we observe that the operator $\mathrm{i} Q$ is unitary and preserves the CAR pairing $\{ \; , \, \}$ since
\begin{equation*}
    \mathrm{i}Q : \quad c \mapsto - \mathrm{i} c \,, \quad c^\dagger \mapsto \mathrm{i} c^\dagger \, ,
\end{equation*}
is an automorphism of the canonical anti-commutation relations (\ref{eq0:CAR}). Moreover, $J_1 = \gamma\, T$ anti-commutes with $Q$ because $T$ preserves the decomposition $W = U \oplus V$ while $\gamma$ swaps the two summands. Therefore, the operator $J_2$ defined by
\begin{equation}
    J_2 := \mathrm{i} Q J_1 = - \mathrm{i} J_1 Q
\end{equation}
has the properties of anti-commuting with $J_1$ and squaring to $- \mathrm{Id}$. Because both $J_1$ and $\mathrm{i}Q$ are orthogonal unitary transformations, so is $J_2\,$. Altogether, we now have two orthogonal unitary generators $J_1, J_2$ satisfying the Clifford algebra relations (\ref{Eq:Cliff-s}) for $s = 2$.

Now recall $J_1 A_k = A_k^\mathrm{c}$ and use $Q\, A_k = A_k$ to do the following computation:
\begin{equation*}
    J_2 A_k = - \mathrm{i} J_1 Q \, A_k = - \mathrm{i} J_1 A_k = - \mathrm{i} A_k^\mathrm{c} = A_k^\mathrm{c} \,.
\end{equation*}
Thus the fibers $A_k$ of a translation-invariant Hartree-Fock ground state of a band insulator in symmetry class $\AII$ are constrained by the pseudo-symmetry conditions
\begin{equation*}
    J_1 A_k = J_2 A_k = A_k^\mathrm{c} \,,
\end{equation*}
reflecting the true symmetry conditions $T A_k = A_{\tau(k)}$ and $Q \, A_k = A_k\,$. This means that such a ground state is an $\IQPV$ of class $s = 2$ in the sense of Definition \ref{def:2.2}.

\subsubsection{Discussion, and class $A$}\label{sect:2.3.1}

Let us add here some discussion to reveal the physical meaning of the ground-state fibers $A_k\,$, as this meaning may be somewhat concealed by our comprehensive framework. The condition $Q \, A_k = A_k$ of conserved particle number forces $A_k$ for all $k$ to be of the form
\begin{equation}
    A_k = A_k^{\rm p} \oplus A_k^{\rm h} \,, \quad A_k^{\rm p} = A_k \cap U_k \,, \quad A_k^{\rm h} = A_k \cap V_{\tau(k)} \,.
\end{equation}
Phrased in physics language, an annihilation operator in the fiber $A_k \subset W_k$ of the Hartree-Fock ground state $\mathcal{A}$ is either an operator that annihilates a particle in an unoccupied state of momentum $k$, or is an operator that annihilates a hole (i.e., creates a particle) in an occupied state of momentum $\tau(k)$. For the physical situation at hand (namely, that of a band insulator) the dimension $n_{\rm p} \equiv \dim A_k^{\rm p}$ is independent of $k$ and is called the number of conduction bands. The dimension $n - n_{\rm p} \equiv n_{\rm h} =  \dim A_k^{\rm h}$ is called the number of valence bands.

Now recall that $J_1 = T \gamma$ and $J_1 A_k = A_k^\mathrm{c}$. Since $\gamma$ maps $U_k$ to $V_k$ and $T$ maps $V_k$ to $V_{\tau(k)}$ we have $J_1 A_k^{\rm p} \subset V_{\tau(k)}\,$. Similarly, $J_1 A_k^{\rm h} \subset U_k \,$. Thus the orthogonality relation $\left\langle A_k , J_1 A_k \right\rangle = 0$ splits into two parts:
\begin{equation*}
    \big\langle A_k^{\rm h} , J_1 A_k^{\rm p} \big\rangle = 0 = \big\langle A_k^{\rm p} , J_1 A_k^{\rm h} \big\rangle .
\end{equation*}
Since $J_1$ is unitary and $J_1^2 = -\mathrm{Id}$, these two equations are not independent but imply one another. Moreover, given one of the two spaces, say $A_k^{\rm h}$, they determine the other space $A_k^{\rm p}$ as the orthogonal complement of $J_1 A_k^{\rm h}$ in $U_k$ (and, turning it around, $A_k^{\rm h}$ as the orthogonal complement of $J_1 A_k^{\rm p}$ in $V_{\tau(k)}$). Thus $A_k = A_k^{\rm p} \oplus A_k^{\rm h}$ is already determined completely by specifying just one of the two components, say $A_k^{\rm h}$. Physically speaking, this means that the number-conserving Hartree-Fock ground states at hand are determined by specifying for each momentum $k$ the space of valence band states. Let us also remark that the vector bundle $\mathcal{A} \to M$ with (reduced) fibers $\pi^{-1}(k) = A_k^{\rm h}$ and anti-unitary symmetry $T : \; A_k^{\rm h} \to A_{\tau(k)}^{\rm h}$ constitutes a Quaternionic vector bundle in the sense of \cite{deNittis-AII}.

We take this opportunity to mention one important symmetry class which lies outside the series $s = 0, 1, \ldots, 7$ considered in this paper -- namely symmetry class $A$, where one imposes the symmetry of $Q$ but \emph{not} that of $T$. What happens in that case? The answer is that one gets a complex vector bundle \emph{without} any additional structure. In fact, the process of imposing the symmetry $Q A_k = A_k$ and reducing from $A_k$ to $A_k^{\rm h}$ simply deletes the Fermi constraint and leaves a rank-$n_{\rm h}$ complex vector bundle with fibers $A_k^{\rm h}$ subject to no symmetry conditions at all. Class $A$ plays an important role in the historical development of the subject, as it hosts the class of systems exhibiting the integer quantum Hall effect, where the role of topology was first discovered and understood.

\subsection{Class $s = 3$ (alias $\CII$)}\label{sect:CII}

Next, we augment time reversal and particle number by a third symmetry: twisted particle-hole symmetry, which takes us to class $\CII$. The operator, $C$, of twisted particle-hole conjugation is an anti-unitary transformation exchanging particle creation with particle annihilation operators (or particles with holes, for short); it is a non-relativistic analog of charge conjugation for Dirac fermions.

In explicit terms, the transformation $C : \; W_k \to W_{\tau(k)}$ consists of a pair of maps
\begin{align*}
    &C : \quad U_k \to V_k \,, \quad \sum_j u_j^{\vphantom{\dagger}} \, c_{k,j}^{\vphantom{\dagger}} \mapsto \sum_{i j} S_{i j}^{ \vphantom{\dagger}} \, \bar{u}_j^{ \vphantom{\dagger}}\, c_{k,\,i}^\dagger \,, \cr
    &C : \quad V_{\tau(k)} \to U_{\tau(k)} \,, \quad \sum_j v_j^{\vphantom{\dagger}} \,c_{-k,j}^\dagger \mapsto \sum_{i j} S_{j\,i}^{ \vphantom{\dagger}}\, \bar{v}_j^{\vphantom{\dagger}} \, c_{-k,\,i}^{\vphantom{\dagger}} \,.
\end{align*}
``Twisting'' refers to the presence of a linear operator $S = S^\dagger = S^{-1} :\; V_k \to V_k$ with transpose $S^T :\; U_k\to U_k\,$. (Recall that for any linear operator $L : \; X \to Y$ one has a canonically defined adjoint or transpose, $L^T : \; Y^\ast \to X^\ast$. Note also that $U_k$ can be regarded as the dual vector space $V_k^\ast$ by the CAR pairing.) In the typical examples offered by physics, $S$ exchanges the conduction and valence bands of a system at half filling. We require that $S$ commutes with $T$. Note the relations
\begin{equation}\label{eq:2.19}
    C^2 = \mathrm{Id}, \quad C T = T C , \quad C \gamma = \gamma \, C \,.
\end{equation}

Now a particle-hole symmetric ground state $\mathcal{A} \to M$ obeys the symmetry condition
\begin{equation}\label{Eq:Vect-CII}
    \forall k \in M : \quad C A_k = A_{\tau(k)} \,.
\end{equation}
To bring this in line with the general scheme, consider the linear operator
\begin{equation}
    J_3 = \mathrm{i} Q \gamma\, C = \mathrm{i} \gamma\, C Q \,,
\end{equation}
which squares to $-\mathrm{Id}$ and is a unitary transformation preserving the CAR pairing of $W$ (because both $\mathrm{i}Q$ and $\gamma\, C$ are). It anti-commutes with both $J_1$ and $J_2$ (because $Q$ does, while $\gamma\, C$ commutes), so we now have the Clifford algebra relations (\ref{Eq:Cliff-s}) for $s = 3$.

$J_3$ applied to $A_k$ gives
\begin{equation*}
    J_3 A_k = \mathrm{i}\gamma\, C Q\,A_k = \gamma\,C A_k\,.
\end{equation*}
By using $C A_k = A_{\tau(k)} = A_k^\perp$ we arrive at
\begin{equation*}
    J_3 \, A_k = \gamma \, A_{\tau(k)} = A_k^\mathrm{c} \,.
\end{equation*}
Thus a translation-invariant free-fermion ground state of a gapped system in symmetry class $\CII$ is an $\IQPV$ of class $s = 3$ in the sense of Definition \ref{def:2.2}.

\subsubsection{Class $\AIII$}\label{sect:AIII}

For use in the final Sections \ref{sect:7} and \ref{sect:8}, we mention here another ``complex'' symmetry class, namely $\AIII$, which is like class $A$ in that it lies outside the 8-fold scheme of the ``real'' symmetry classes ($s = 0, \ldots, 7$). Class $\AIII$ differs from $\CII$ by the absence of time-reversal symmetry $T$; i.e., one has only the Fermi constraint and the symmetries under particle number $Q$ and particle-hole conjugation $C$. As discussed in Section \ref{sect:2.3.1}, the Fermi constraint gets effectively canceled by $Q$. Nevertheless, in the presence of the true symmetry $C$ there is still the pseudo-symmetry $J_3 = \mathrm{i} \gamma\, C Q$. In other words, the situation is formally like that of class $\DIII$ ($s = 1$), but with the Fermi constraint out of force. The pseudo-symmetry $J_3$ is often understood as a so-called sublattice symmetry; the latter, however, is not a true symmetry in our sense, as it anti-commutes with the Hamiltonian.

\subsection{Going beyond $s = 3$}\label{sect:s-geq-3}

To continue the Kitaev sequence beyond $s = 3$, we need to expand the physical setting by bringing in true symmetries (namely, spin rotations) in a way different from before. We first describe the total algebraic framework that emerges for $s \geq 4$ and then explain the physics for each of the symmetry classes $s = 4, 5, 6, 7$ in sequence.

Thus, let us assume that on $W = \mathbb{C}^{2n}$ we are given two sets of orthogonal unitary operators, $\{ j_1, j_2 \}$ and $\{ j_5 , \ldots, j_s \}$. The former will be recognized as (two of the three) spin-rotation generators and the latter as pseudo-symmetries due to the possible presence of $T$, $Q$, and $C$. Here $s \geq 4$ and the second set is understood to be empty when $s = 4$. The motivation for leaving a gap in the index set will become clear shortly.

We demand that the following algebraic relations be satisfied for our operators:
\begin{align}\label{eq:red-CliffAlg}
    j_l j_m + j_m j_l &= - 2 \delta_{l m} \mathrm{Id}_W \quad (1 \leq l , m\leq 2) , \cr j_l j_m - j_m j_l &= 0 \quad (1 \leq l \leq 2; \;\; 5 \leq m \leq s),\\ j_l j_m + j_m j_l &= - 2 \delta_{l m} \mathrm{Id}_W \quad (5 \leq l , m\leq s) . \nonumber
\end{align}
Thus $\{ j_1, j_2 \}$ and $\{ j_5 , \ldots, j_s \}$ are two sets of Clifford algebra generators on $W$, and any two generators belonging to different sets commute with one another.

As before, the translation-invariant free-fermion ground state of a gapped system (now of symmetry class $s \geq 4$) will be described by a vector bundle over $M$ with $n$-dimensional fibers $a_k \subset W = \mathbb{C}^{2n}$ spanned by the quasi-particle annihilation operators at momentum $k$. (The change of notation from $A_k$ to $a_k$ is to clear the symbol $A_k$ for use with a closely related, but different object.) For reasons that will be explained in detail in the following subsections, the vector spaces $a_k$ are required to obey the set of conditions
\begin{equation}\label{eq:red-psym}
    \forall k \in M : \quad a_k^\perp = a_{\tau(k)} , \quad j_1 a_k = j_2 \, a_k = a_k \,, \quad j_5 \, a_k = \ldots = j_s \, a_k = a_k^\mathrm{c} \,.
\end{equation}
Notice that $j_1, j_2$ are true symmetries taking $a_k$ to itself, whereas $j_5, \ldots, j_s$ are pseudo-symmetries taking $a_k$ to its orthogonal complement $a_k^\mathrm{c}\,$. We will now demonstrate that such a multiplet of (pseudo-)symmetries is equivalent to a set of $s$ pseudo-symmetries $J_1, \ldots, J_s\,$.

The key step is to double the dimension of $W$ by taking the tensor product with $\mathbb{C}^2$, and to consider on $\mathbb{C}^2 \otimes W$ the set of operators
\begin{equation}\label{eq-mz:2.30}
    \begin{split}
    &J_l = \begin{pmatrix} 1 &0\cr 0 &-1 \end{pmatrix} \otimes j_l \quad (l = 1, 2) , \quad J_3 = \begin{pmatrix} 1 &0\cr 0 &-1 \end{pmatrix} \otimes j_2 j_1 \,, \cr &J_4 = \begin{pmatrix} 0 &1\cr -1 &0 \end{pmatrix} \otimes \mathrm{Id}_W , \quad J_m = \begin{pmatrix} 0 &1\cr 1 &0 \end{pmatrix} \otimes j_m \quad (m = 5, \ldots, s) .
    \end{split}
\end{equation}
By using the algebraic properties laid down in (\ref{eq:red-CliffAlg}) one readily verifies that the operators $J_1, \ldots, J_s$ so defined satisfy the Clifford algebra relations (\ref{Eq:Cliff-s}).

The strategy now is to transfer all relevant structure of $W$ to $\mathbb{C}^2 \otimes W$. In the case of the Hermitian scalar product $\langle \, ,\, \rangle_W$ we do this by viewing the doubled space as the orthogonal sum $W_+ \oplus W_- = \mathbb{C}^2 \otimes W$ of two identical copies $W_+ = W_- = W$ and setting
\begin{equation}
    \langle \, , \, \rangle_{\mathbb{C}^2 \otimes W} =
    \langle \, , \, \rangle_{W_+} + \langle \, , \, \rangle_{W_-}\,.
\end{equation}
The CAR bracket $\{ \, , \}_W$ is transferred to $\mathbb{C}^2 \otimes W$ by the same principle. The transferred structures define involutions $L \mapsto L^\mathrm{c}$ and $L \mapsto L^\perp$ as before. Note that with these conventions all operators $J_1, \ldots, J_s$ are orthogonal unitary transformations of $\mathbb{C}^2 \otimes W$.

Now let $\{ a_k \}_{k \in M}$ be a vector bundle with $n$-dimensional fibers $a_k \subset W$ that satisfy the conditions (\ref{eq:red-psym}). Then we construct a new vector bundle $\{ A_k \}_{k \in M}$ with $2n$-dimensional fibers $A_k = f(a_k) \subset \mathbb{C}^2 \otimes W$ by applying the transformation
\begin{equation}\label{eq-mz:2.31}
    f : \; a \mapsto A = \left\{ \begin{pmatrix} 1 \cr 1 \end{pmatrix} \otimes w + \begin{pmatrix} 1 \cr -1 \end{pmatrix} \otimes w^\prime \mid w \in a \,, \; w^\prime \in a^\mathrm{c} \right\} .
\end{equation}
A short computation shows that the relations (\ref{eq:red-psym}) translate into the relations
\begin{equation}\label{eq-mz:2.32}
    \forall k \in M : \quad A_k^\perp = A_{\tau(k)} , \quad J_1 A_k = \ldots = J_s \, A_k = A_k^\mathrm{c} \,.
\end{equation}
Thus we have assigned to a vector bundle $\{ a_k \}_{k \in M}$ constrained by the (pseudo-)symmetry conditions (\ref{eq:red-psym}) an $\IQPV$ $\mathcal{A} \to M$ of class $s$ in the sense of Definition \ref{def:2.2}. This correspondence turns out to be one-to-one.
\begin{proposition}\label{prop:2.1-new}
Fix a system $j_1, j_2, j_5, \ldots, j_s$ and a corresponding system $J_1, \ldots, J_s\,$. Then the solutions $a_k$ of Eqs.\ (\ref{eq:red-psym}) are in bijection with the solutions $A_k$ of Eqs.\ (\ref{eq-mz:2.32}).
\end{proposition}

While the proof does have some bearing on the rest of this paper, it is not essential here. We therefore relegate it to the Appendix (in combination with the material of Section \ref{sect:1,1-red} and the Proposition \ref{prop:2.1X} proved there) and proceed with the main message of this section.

\subsection{Class $s = 4$ (alias $C$)}\label{sect:C}

We are now ready to address class $C$, which is defined to be the symmetry class of fermions with spin $1/2$ and $\mathrm{SU}_2$ spin-rotation symmetry (plus the pervasive translation invariance of the present context). Note that class $C$ does not follow upon $\CII$ in the same way that class $\CII$ follows upon $\AII$ or class $\AII$ upon $\DIII$. In fact, the operators $T$, $Q$, and $C$ characteristic of the preceding classes cease to be symmetries here; they are superseded by the spin-rotation generators. Examples of quasi-particle vacua of symmetry class $C$ are found among superconductors with spin-singlet pairing.

Let the generators of $\mathrm{SU}_2$ spin rotations be denoted by $j_1, j_2$, and $j_3$. As operators on the spinor space $\mathbb{C}^2$ they are represented by $2 \times 2$ matrices, say
\begin{equation*}
    j_1 = \begin{pmatrix} 0 &\mathrm{i}\cr \mathrm{i} &0 \end{pmatrix} , \quad j_2 = \begin{pmatrix} 0 &1\cr -1 &0 \end{pmatrix} , \quad
    j_3 = \begin{pmatrix} \mathrm{i} &0\cr 0 &-\mathrm{i}\end{pmatrix} .
\end{equation*}
One may also think of these matrices $j_1$, $j_2$ and $j_3 = j_2 j_1$ as a basis (including the unit matrix) for the algebra $\mathbb{H}$ of quaternions. For the following, we assume that the quaternion algebra of $j_1 , j_2 , j_3$ acts reducibly on our vector spaces $W_k = U_k \oplus V_{\tau(k)} \simeq \mathbb{C}^{2n}$ for $k \in M$.

Here as always, the translation-invariant quasi-particle vacuum of a gapped system (now of class $C$) is described by a vector bundle over $M$ with $n$-dimensional fibers $a_k \subset W = \mathbb{C}^{2n}$ spanned by the quasi-particle annihilation operators at momentum $k$. These fibers are still subject to the Fermi constraint $a_k^\perp = a_{\tau(k)}\,$. The property of spin-rotation invariance of the quasi-particle vacuum is expressed by the true symmetry conditions $j_l\, a_k = a_k$ ($l = 1, 2, 3$). Altogether, we now have the set of equations
\begin{equation}\label{eq:SU2-inv}
    \forall k \in M : \quad a_k^\perp = a_{\tau(k)} , \quad j_1 a_k = j_2 \, a_k = j_3 \, a_k = a_k \,.
\end{equation}
Owing to the quaternion relation $j_3 = j_2 j_1$ we may drop the last condition ($j_3 \, a_k = a_k$) as this is already implied by $j_l \, a_k = a_k$ for $l = 1, 2$. We then see that the conditions (\ref{eq:SU2-inv}) coincide with the set of conditions (\ref{eq:red-psym}) for $s = 4$.

{}Following the blueprint of Section \ref{sect:s-geq-3}, we now double up the vector space $W$ to $\mathbb{C}^2 \otimes W$ and use the mapping $f$ of (\ref{eq-mz:2.31}) to transform the vector bundle with fibers $a_k$ to an equivalent vector bundle $\mathcal{A} \to M$ with fibers $A_k = f(a_k)$. By the assignments in (\ref{eq-mz:2.30}), the Clifford algebra $\mathbb{H} = \mathrm{Cl}(\mathbb{R}^2)$ generated by $j_1$ and $j_2$ becomes the Clifford algebra $\mathrm{Cl}(\mathbb{R}^4)$ generated by $J_1, \ldots, J_4\,$. According to Eq.\ (\ref{eq-mz:2.32}) the transformed fibers $A_k$ are subject to
\begin{equation}
    \forall k \in M : \quad A_k^\perp = A_{\tau(k)} , \quad J_1 A_k = \ldots = J_4 \, A_k = A_k^\mathrm{c} \,.
\end{equation}
Since the mapping $a_k \leftrightarrow A_k$ is one-to-one, we see that the translation-invariant free-fermion ground state of a gapped superconductor in symmetry class $C$ is precisely modeled by an $\IQPV$ of class $s = 4$ in the sense of Definition \ref{def:2.2}.

\subsection{Class $s = 5$ (alias $\CI$)}\label{sect:CI}

The genesis of the remaining 3 symmetry classes ($s = 5, 6, 7$) is parallel to that of the classes $s = 1, 2, 3$: they are obtained by first imposing time-reversal invariance, then charge conservation, and finally particle-hole conjugation symmetry. The difference from the earlier setting is that $\mathrm{SU}_2$ spin rotations now are symmetries throughout. In view of the detailed treatment given in Sections \ref{sect:DIII}--\ref{sect:CII}, we can be brief here.

The first additional symmetry to impose is time-reversal invariance. As before, we assume fermions with spin $1/2$, so that $T^2 = - \mathrm{Id}$. The new symmetry condition on the fibers is
\begin{equation}\label{eq:sym-CI}
    T a_k = a_{\tau(k)} .
\end{equation}
The resulting symmetry class is commonly called $\CI$.

By composing $T : \; W \to W$ with $\gamma : \; W \to W$ we get an orthogonal unitary operator
\begin{equation}
    j_5 = \gamma \,T : \;\; W \to W , \quad \text{with} \quad j_5^2 = - \mathrm{Id} \,.
\end{equation}
We will now argue on physical grounds that $j_5$ commutes with the spin-rotation generators $j_l$ for $l = 1, 2, 3$. For this, we first observe that the physical observable of spin, like any component of momentum or angular momentum, is inverted by the operation of time reversal. Since $T$ is complex anti-linear and our generators $j_l$ carry an extra factor of $\mathrm{i} = \sqrt{-1}$ as compared to the physical spin observables, we infer that $T j_l \, T^{-1} = + j_l$ (for $l = 1, 2, 3)$. Secondly, spin rotations $g = \mathrm{e}^{\,\sum x_l j_l}$ preserve the CAR pairing $\{ \, , \, \}$ and (for $x_l \in \mathbb{R}$) the Hermitian structure $\langle \, , \, \rangle$; thus they are orthogonal unitary transformations of $W$. This implies that spin rotations commute with $\gamma$ and so do their generators $j_l\,$. Altogether, we obtain
\begin{equation}
    j_l j_5 - j_5 j_l = 0 \quad (l = 1, 2, 3),
\end{equation}
as claimed. Thus we have all the relations (\ref{eq:red-CliffAlg}) for $s = 5$ in place.

Now for reasons explained in Section \ref{sect:DIII}, the condition (\ref{eq:sym-CI}) is equivalent to
\begin{equation*}
    j_5 \, a_k = a_k^\mathrm{c} .
\end{equation*}
We recall the Fermi constraint $a_k^\perp = a_{\tau(k)}$ and the symmetry conditions (\ref{eq:SU2-inv}). By the transcription $a_k \leftrightarrow A_k$ of Section \ref{sect:s-geq-3}, it follows that the translation-invariant free-fermion ground state of a gapped superconductor in symmetry class $\CI$ is exactly given by an $\IQPV$ of class $s = 5$ in the sense of Definition \ref{def:2.2}.

\subsection{Class $s = 6$ (alias $\AI$)}\label{sect:AI}

Next, by including the $\mathrm{U}(1)$ symmetry group underlying particle-number conservation, we are led to what is called symmetry class $\AI$. In addition to the previous conditions on fibers we now have
\begin{equation}
    \forall k \in M : \quad Q \, a_k = a_k \,.
\end{equation}
As before, $Q = + 1$ on creation operators and $Q = - 1$ on annihilation operators.

To transcribe this condition to the present framework, we introduce
\begin{equation}
    j_6 := \mathrm{i} Q j_5 \,.
\end{equation}
The two operators $j_5$ and $j_6$ share the algebraic properties of the pair $J_1, J_2\,$; for the detailed reasoning we refer to Section \ref{sect:AII}. Moreover, $j_6$ like $j_5$ commutes with the spin-rotation generators $j_1, j_2, j_3$. Thus we now have the algebraic relations (\ref{eq:red-CliffAlg}) for $s = 6$.

The true symmetry conditions $a_k = Q\, a_k = T a_{\tau(k)}$ are equivalent to the pseudo-symmetry conditions
\begin{equation*}
    j_5  \, a_k = j_6 \, a_k = a_k^\mathrm{c} .
\end{equation*}
In conjunction with the Fermi constraint $a_k^\perp = a_{\tau(k)}$ and the spin-rotation symmetries (\ref{eq:SU2-inv}), this means that translation-invariant Hartree-Fock ground states of insulators in symmetry class $\AI$ are given by $\IQPV$s of class $s = 6$.

\subsection{Class $s = 7$ (alias $\BDI$)}\label{sect:BDI}

Finally, to arrive at class $s = 7$ (also known as $\BDI$) we augment the symmetry operations of translations, spin rotations, time reversal and $\mathrm{U}(1)$ gauge transformations by (twisted) particle-hole conjugation $C$. Thus we require
\begin{equation}
    \forall k \in M : \quad C a_k = a_{\tau(k)} \,.
\end{equation}
The properties of the anti-unitary operator $C$ were listed in (\ref{eq:2.19}). In addition, we demand that the twisting operator $\gamma\, C$ commutes with the spin-rotation generators $j_1, j_2, j_3\,$.

For reasons that were explained in Section \ref{sect:CII}, the unitary operator
\begin{equation}
    j_7 = \mathrm{i} Q \gamma\, C = \mathrm{i} \gamma\, C Q
\end{equation}
preserves the CAR pairing of $W$. It squares to $-\mathrm{Id}$, anti-commutes with both $j_5$ and $j_6\,$, and commutes with $j_1$, $j_2$, and $j_3$. Thus we now have the relations (\ref{eq:red-CliffAlg}) for $s = 7$.

The symmetry condition $C a_k = a_{\tau(k)}$ is equivalent to the pseudo-symmetry condition
\begin{equation*}
    j_7 \, a_k = a_k^\mathrm{c}.
\end{equation*}
In view of this and all the other constraints obeyed by $a_k\,$, the translation-invariant free-fermion ground state of a gapped system in symmetry class $\BDI$ is an $\IQPV$ of class $s = 7$.

As a final remark, let us mention that there exist simpler ways of realizing class $\BDI$ in physics. (A similar remark applies to class $\AI$.) By the $(1,1)$ periodicity theorem of Section \ref{sect:1,1-red} and the 8-fold periodicity of real Clifford algebras \cite{ABS}, the effect of 7 ``real'' pseudo-symmetries $J_1, \ldots, J_7$ is the same (after reducing the number of bands by a factor of $2^4$) as that of a single ``imaginary'' pseudo-symmetry $K$. One may take $K = \mathrm{i} \gamma\, C$; thus class $\BDI$ is realized by superconductors with particle-hole conjugation symmetry. For another superconducting realization, one may take $K = \mathrm{i} \gamma T$ with a time-reversal operator $T$ that squares to $+\mathrm{Id}$.

\section{From vector bundles to classifying maps}\label{sect:ClassMaps}
\setcounter{equation}{0}

In this section we pass from the vector-bundle description to an equivalent description by what we call ``classifying maps'' for short. (Note that this usage is not in accordance with standard terminology.) Recall from Definition \ref{def:2.2} that an $\IQPV$ of class $s$ is a rank-$n$ complex subvector bundle $\mathcal{A} \stackrel{\pi}{\to} M$ with the property that its fibers $\pi^{-1}(k) = A_k \subset \mathbb{C}^{2n}$ obey the pseudo-symmetry conditions $J_1 A_k = \ldots = J_s \, A_k = A_k^\mathrm{c}$ and the Fermi constraint $A_k^\perp = A_{\tau(k)}$ for all momenta $k \in M$. The equivalent description by a classifying map is as follows.

Let $C_0(n) \equiv \cup_{r=0}^{2n} \mathrm{Gr}_r(\mathbb{C}^{2n})$ where $\mathrm{Gr}_r(\mathbb{C}^{2n})$ is the Grassmannian of complex $r$-planes $A$ in $W = \mathbb{C}^{2n}$. (Although the Fermi constraint $A^\perp = A$ singles out $r = n$, we allow $r \not= n$ here for later convenience.) Given $C_0(n)$, let $C_s(n) \subset C_0(n)$ be the subspace of complex hyperplanes that satisfy the constraints due to $s$ pseudo-symmetries $J_1, \ldots, J_s:$
\begin{equation}
    C_s(n) = \{A\in C_0(n) \mid J_1 A = \ldots = J_s \, A = A^\mathrm{c}\}.
\end{equation}
The classifying map $\Phi$ for a vector bundle $\mathcal{A} \to M$ of class $s$ then is simply the map
\begin{align}\label{eq:3.2}
    \Phi : \; M \to C_s(n) , \quad k \mapsto A_k \,,
\end{align}
assigning to the momentum $k \in M$ the complex hyperplane $A_k \in C_s(n)$.

This reformulation does not yet account for the Fermi constraint $A_k^\perp = A_{\tau(k)}$. To incorporate it, we denote by
\begin{equation}
    \tau_0 : \; C_0(n) \to C_0(n)
\end{equation}
the involution that sends a complex $r$-plane $L \subset W$ to the complex $(2n-r)$-plane $L^\perp \subset W$. We notice that $C_0(n) \supset C_1(n) \supset \ldots \supset C_s(n)$. Since the transformations $J_l : \; C_0(n) \to C_0(n)$ preserve the CAR pairing, $\{ J_l L \,, J_l L^\perp \} = \{ L \,, L^\perp \} = 0$, they commute with $\tau_0\,$. Therefore $\tau_0$ descends to an involution
\begin{equation}
    \tau_s : \; C_s(n) \to C_s(n)
\end{equation}
for all $s = 1, 2, \ldots$ by restriction. The condition $A_k^\perp = A_{\tau(k)}$ now becomes
\begin{equation}\label{eq:Phi-eqvt}
    \tau_s \circ \Phi = \Phi \circ \tau .
\end{equation}
Fixing a class $s$, we have that the group $\mathbb{Z}_2$ acts on two spaces, $M$ and $C_s(n)$, with the non-trivial element acting by $\tau$ on the former and $\tau_s$ on the latter. In view of this, the condition (\ref{eq:Phi-eqvt}) can be rephrased as saying that the mapping $\Phi : \; M \to C_s(n)$ is $\mathbb{Z}_2$-equivariant.

An important role is played by the special momenta that satisfy $k = \tau(k)$. At these points of $M$, the condition (\ref{eq:Phi-eqvt}) of $\mathbb{Z}_2$-equivariance constrains $\Phi$ to take values in the set of fixed points of $\tau_s\,$. We denote this subspace by
\begin{equation}
    R_s(n)\equiv \mathrm{Fix}(\tau_s) = \{A\in C_s(n)\mid A = A^\perp \}.
\end{equation}

The reformulation of the current subsection is summarized by the following statement.
\begin{proposition}\label{prop:3.2}
Let $W = \mathbb{C}^{2n}$. The set of rank-$n$ complex subvector bundles $\mathcal{A} \to M$ of symmetry class $s$ (also referred to as $\IQPV$s of class $s$; see Def.\ \ref{def:2.2}) is in one-to-one correspondence with the set of classifying maps $\Phi : \; M \to C_s(n) \subset \mathrm{Gr}_n (W)$ that are $\mathbb{Z}_2$-equivariant, $\Phi = \tau_s \circ \Phi \circ \tau^{-1}$, for the involution $\tau_s : \; C_s(n) \to C_s(n)$, $A \mapsto A^\perp$. At $\tau$-invariant momenta $k = \tau(k)$ the map $\Phi$ takes values in a subspace $R_s(n) = \mathrm{Fix}(\tau_s)$.
\end{proposition}
\begin{remark}
Having recast the Fermi constraint as a condition of $\mathbb{Z}_2$-equivariance, one may wonder why we could not regard our quasi-particle vacua as $\mathbb{Z}_2$-equivariant vector bundles. The answer is that although the $\perp$-operation gives rise to a well-defined involution $\tau_s$ on $C_s(n)$, it does not determine (not for general values of $s$) any kind of complex linear or anti-linear mapping from $A_k$ to $A_{\tau(k)}$.
\end{remark}
\begin{remark}
We refer to $C_s(n)$ and $R_s(n)$ as the ``complex'' and ``real'' classifying spaces for vector bundles of symmetry class $s$. Although the two descriptions by vector bundles and classifying spaces are in principle equivalent, they suggest different notions of topological equivalence. This point is elaborated in the next subsection.
\end{remark}

Proposition \ref{prop:3.2} gives a characterization of our vector bundles
which is concise and efficient for the purpose of systematic classification by topological equivalence. Yet, the precise nature of the spaces of $\mathbb{Z}_2$-equivariant classifying maps $\Phi$ may not reveal itself immediately to the novice, as the situation seems to get more and more involved and constrained for an increasing number of pseudo-symmetries $J_1, \ldots, J_s\,$. However, the identification and detailed discussion of the classifying spaces $C_s(n)$ and their subspaces $R_s(n)$ of $\tau_s$-fixed points for all classes $s = 0, 1, 2, ..., 7$ can be found in the published literature; see \cite{milnor,StoneEtAl}. (To see that our definition of the ``real'' spaces $R_s(n)$ agrees with that of the literature, one observes that by the relation (\ref{eq-mz:2.9}) the Hermitian structure $\langle \, , \, \rangle$ and the CAR bracket $\{ \, , \, \}$ reduce to the same Euclidean structure on the real subspace $\mathbb{R}^{2n} = W_\mathbb{R} \subset W$ of $\gamma$-fixed points, see Remark \ref{remark:orthogonalunitary}.) The well-known outcome of this exercise is displayed in Table \ref{table:CsRs}, where we substitute $n \equiv 8r$.
\begin{table}
\begin{center}
\begin{tabular}{c|c|c}
$s$	& $C_s(8r)$ &$R_s(8r)$ \\ \hline
$0$	
& $\cup_{p+q=16r}\, \mathrm{U}_{16r}/ (\mathrm{U}_{p} \times \mathrm{U}_{q})$
& $\mathrm{O}_{16r}/ \mathrm{U}_{8r}$	\\
$1$	
& $(\mathrm{U}_{8r} \times \mathrm{U}_{8r}) / \mathrm{U}_{8r}$
& $\mathrm{U}_{8r}/\mathrm{Sp}_{8r}$\\
$2$	
& $\cup_{p+q=8r}\, \mathrm{U}_{8r}/(\mathrm{U}_{p}\times \mathrm{U}_{q})$
& $\cup_{p+q=4r} \, \mathrm{Sp}_{8r}/(\mathrm{Sp}_{2p}\times \mathrm{Sp}_{2q})$\\
$3$	
& $(\mathrm{U}_{4r} \times \mathrm{U}_{4r}) / \mathrm{U}_{4r}$
& $(\mathrm{Sp}_{4r} \times \mathrm{Sp}_{4r}) / \mathrm{Sp}_{4r}$\\
$4$	
& $\cup_{p+q=4r}\, \mathrm{U}_{4r}/(\mathrm{U}_{p}\times \mathrm{U}_{q})$
& $\mathrm{Sp}_{4r}/\mathrm{U}_{2r}$\\
$5$	
& ($\mathrm{U}_{2r} \times \mathrm{U}_{2r}) / \mathrm{U}_{2r}$
& $\mathrm{U}_{2r}/\mathrm{O}_{2r}$\\
$6$	
& $\cup_{p+q=2r}\, \mathrm{U}_{2r}/ (\mathrm{U}_{p}\times \mathrm{U}_{q})$
& $\cup_{p+q=2r}\, \mathrm{O}_{2r}/ (\mathrm{O}_{p}\times \mathrm{O}_{q})$\\
$7$	
& ($\mathrm{U}_{r} \times \mathrm{U}_{r}) / \mathrm{U}_r$
& ($\mathrm{O}_{r} \times \mathrm{O}_{r}) / \mathrm{O}_r$
\end{tabular}
\end{center}
\vspace{10pt}
\caption{Realization of $C_s$ and $R_s = \mathrm{Fix}(\tau_s)$ as (unions of) homogeneous spaces.}\label{table:CsRs}
\end{table}
One observes that $C_{s+2}(2n) = C_s(n)$. This 2-fold periodicity reflects the fact that doubling the representation space and extending a complex Clifford algebra by $2$ generators is the same as tensoring it with the full algebra of complex $2 \times 2$ matrices. In the same vein, there is an 8-fold periodicity $R_{s+8}(16n) = R_s(n)$, reflecting a similar isomorphism \cite{ABS} over the real number field.

\subsection{Classification schemes}\label{sect:equiv-class}

To recapitulate: we have two descriptions of an $\IQPV$ of class $s$. On one hand, we may view it as a rank-$n$ complex subvector bundle $\mathcal{A} \stackrel{\pi}{\to} M$ with fibers $\pi^{-1}(k) = A_k \subset W = \mathbb{C}^{2n}$ subject to $A_k^\perp = A_{\tau(k)}$ and the pseudo-symmetry conditions (\ref{eq:2.7}). On the other hand, we may describe it by a classifying map $\Phi : \; M \to C_s(n)$ subject to the condition $\tau_s \circ \Phi = \Phi \circ \tau$ of $\mathbb{Z}_2$-equivariance. The two descriptions are equivalent.

Our goal is to establish a topological classification of translation-invariant free-fermion ground states of gapped systems with given symmetries (i.e.\ of $\IQPV$s in a given symmetry class). To do so, we need to settle on a notion of topological equivalence. In the present paper, we employ the equivalence relation which is given by \emph{homotopy}: we say that two $\IQPV$s belong to the same topological class if they are connected by a continuous deformation known as a homotopy. More precisely, a homotopy between two $\IQPV$s in class $s$ with classifying maps $\Phi_0$ and $\Phi_1$ is given by a continuous family $\Phi_t$ with $\tau_s \circ \Phi_t = \Phi_t \circ \tau$ for all $t\in[0,1]$. We emphasize that all vector bundles in our setting are subbundles of the trivial bundle $M\times W = M\times\mathbb{C}^{2n}$. Understood in this way, the equivalence relation of homotopy leads, in general, to more topological classes than does the equivalence relation given by the notion of isomorphy of vector bundles. This is illustrated by the following example.

\begin{example}
In the simple case of class $A$ (see Section \ref{sect:2.3.1}), $\IQPV$s with $q$ valence bands and $p = n - q$ conduction bands are rank-$q$ complex subvector bundles of $M\times\mathbb{C}^n$. Denoting the set of isomorphism classes of these bundles by $\mathrm{Vect}_q^\mathbb{C} (M)$, and writing $[M,Y]$ for the set of homotopy classes of maps $M \to Y$, one has a bijection \cite{husemoller}
\begin{equation*}
    \mathrm{Vect}_q^\mathbb{C} (M) \simeq [M , \mathrm{Gr}_q (\mathbb{C}^n)]
\end{equation*}
as long as $2 p \ge \dim M$. This bijection breaks down, however, when the inequality of dimensions is violated; it then becomes possible for two $\IQPV$s to be isomorphic without being homotopic. A concrete example is provided by the ``Hopf magnetic insulator'' \cite{MRW} for $M = \mathrm{S}^3$ with $p = q = 1$, where $2 p = 2 < 3 = \dim \mathrm{S}^3$. Indeed, while all complex line bundles over $\mathrm{S}^3$ are isomorphic to the trivial one ($\mathrm{Vect}_1^\mathbb{C}( \mathrm{S}^3) = 0$), such vector bundles, viewed as subbundles of $\mathrm{S}^3 \times \mathbb{C}^2$, organize into distinct homotopy classes since
\begin{equation*}
    [\mathrm{S}^3, \mathrm{Gr}_1(\mathbb{C}^2)] = \pi_3(\mathrm{S}^2) = \mathbb{Z} .
\end{equation*}
These homotopy classes are distinguished by what is called the Hopf invariant.
\end{example}

A standard approach used in the literature is to work with a further reduction of the topological information contained in isomorphism classes, by adopting the equivalence relation of \textit{stable equivalence} between vector bundles. We will use class $A$ once more in order to illustrate the construction. Two vector bundles $\mathcal{A}_0 \to M$ and $\mathcal{A}_1 \to M$ are stably equivalent if they are isomorphic after adding trivial bundles (meaning trivial valence bands in physics language), i.e.\ if there exist $m_1, m_2 \in \mathbb{N}_0$ such that
\begin{equation*}
    \mathcal{A}_0 \oplus (M \times \mathbb{C}^{m_1}) \simeq \mathcal{A}_1 \oplus (M \times \mathbb{C}^{m_2}) .
\end{equation*}
Under the direct-sum operation, the stable equivalence classes constitute a group called the (reduced) complex $K$-group of $M$, which is denoted as $\widetilde{K}_\mathbb{C} (M)$. (Inverses in this group are given by the fact that for compact $M$, all complex vector bundles $\mathcal{A}$ have a partner $\mathcal{A}^\prime$ such that $\mathcal{A} \oplus \mathcal{A}^\prime \simeq M \times \mathbb{C}^n$ for some $n \in \mathbb{N}$, where the right-hand side represents the neutral element.) In the limit of a large number of valence and conduction bands, namely the \textit{stable regime}, the elements of the reduced $K$-group are in bijection with the homotopy classes of maps into the classifying space \cite{husemoller}:
\begin{equation*}
    \widetilde{K}_\mathbb{C}(M) \simeq [M, \mathrm{Gr}_n (\mathbb{C}^{2n})] \qquad (\text{for}\; 2n \ge \dim M).
\end{equation*}
Outside the stable regime, stably equivalent vector bundles need not be isomorphic, much less homotopic.

\begin{example}
Consider the tangent bundle $T \mathrm{S}^2$ of the two-sphere. By regarding $\mathrm{S}^2$ as the unit sphere in $\mathbb{R}^3$, we also have the normal bundle $N \mathrm{S}^2 \simeq \mathrm{S}^2 \times \mathbb{R}$. The direct sum of $T\mathrm{S}^2$ and $N \mathrm{S}^2$ is $\mathrm{S}^2 \times \mathbb{R}^3$. Thus $T \mathrm{S}^2$ is stably equivalent to the trivial bundle. Yet the isomorphism class of $T \mathrm{S}^2$ differs from that of the trivial bundle. It is the non-trivial element $2 \in \mathbb{N}_0 = \mathrm{Vect}_2^\mathbb{R} (\mathrm{S}^2)$ in Table A.1 of \cite{deNittis-AI}.

In the present context, a physical realization of $T \mathrm{S}^2$ is the ground state of a system in symmetry class $\AI$ in two spatial dimensions ($M = \mathrm{S}^2$), albeit in the generalized sense that the operation of time reversal is replaced by the combination of time reversal and space inversion, which effectively restricts the fibers $A_k$ to be real vector spaces.
\end{example}
\begin{remark}
To compare our approach with that of $K$-theory, we picked the example of class $A$. It turns out that only two more of our symmetry classes are accommodated by the standard formulation of $K$-theory for vector bundles: these are class $\AI$ ($s = 6$), where vector bundles are equipped with a complex anti-linear involution (corresponding to the physical symmetry of time reversal $T$ with $T^2 = +1$), and class $\AII$ ($s = 2$), where the involution is replaced by a projective involution (time reversal $T$ with $T^2 = -1$). In the former case, taking stable equivalence classes leads to $KR$-groups \cite{atiyah,deNittis-AI}, while in the latter case it leads to $KQ$-groups \cite{dupont,deNittis-AII}. For the other symmetry classes, the corresponding $K$-theory groups can only be inferred indirectly by an algebraic construction using Clifford modules as in \cite{kitaev,FreedMoore}. In all cases, the $K$-theory groups of momentum space $M$ are in bijection with the homotopy classes of $\mathbb{Z}_2$-equivariant maps $M\to C_s(n)$ -- denoted by $[M , C_s(n)]^{ \mathbb{Z}_2}$ as a set -- in the limit of large $n$ (as well as large $p$ and $q$ where applicable, see Table \ref{table:CsRs}).
\end{remark}

To sum up, the natural equivalence relation for us to use is that of homotopy. It is a finer tool than stable equivalence (as considered in \cite{kitaev}) and even isomorphy of vector bundles (as considered in \cite{deNittis-AI,deNittis-AII} for $s = 6$ and $s = 2$), and is therefore adopted as our topological classification principle. Although we will ultimately work in the stable regime in order to utilize such results as the Bott periodicity theorem, the use of homotopy theory allows us to keep track of the precise conditions under which our equivalences hold. In other words, we are able to say how many bands are required in order for the physical system to be in the stable regime for a given space dimension.

\section{The diagonal map}\label{sect:diag-map}
\setcounter{equation}{0}

In this section we introduce the ``master diagonal map'' -- a universal mapping that takes a $d$-dimensional $\IQPV$ of class $s$ and transforms it into a $(d+1)$-dimensional $\IQPV$ of class $s + 1$. While there exist in principle many such maps -- for some previous efforts in this direction see \cite{StoneEtAl,TeoKane} -- the one described here stands out in that it can be proven to induce a one-to-one mapping between stable homotopy classes of base-point preserving and $\mathbb{Z}_2$-equivariant maps $M \to C_s(n)$ and $\tilde{S} M \to C_{s+1}(2n)$, where $\tilde{S} M$ denotes the momentum-type suspension of $M$ (see below). Our mapping also bears a close relation to the map underlying the phenomenon of real Bott periodicity.

{}From now on, we will use the model of an $\IQPV$ of symmetry class $s$ as a $\mathbb{Z}_2$-equivariant map $\phi$ from $M$ into the classifying space $C_s \equiv C_s(n)$ with $s$ pseudo-symmetries. The goal is to construct from $\phi$ a new mapping, $\Phi$, which maps $\tilde{S} M$ into a classifying space $C_{s+1}$ with one additional pseudo-symmetry. It is not difficult to see that such a map will not induce an injective map of homotopy classes in general unless the ambient vector space $W$ is enlarged. Therefore our story of constructing $\Phi$ begins with a modification of $W$: we double its dimension by replacing it by $\mathbb{C}^2 \otimes W$. The procedure is identical to that of Section \ref{sect:s-geq-3}, which we assume here to be understood. At the same time, we now extend the given Clifford algebra of pseudo-symmetries by two generators, in the process reviewing and exploiting a result known as $(1,1)$ periodicity.

Let us mention that the physical meaning of the step $W \to \mathbb{C}^2 \otimes W$ depends on the case. For example, for $s = 0$ the tensor factor $\mathbb{C}^2$ introduces a spin-1/2 degree of freedom. For $s = 1$ it replaces a single band by a pair of bands -- one valence and one conduction band.

\subsection{(1,1) periodicity}\label{sect:1,1-red}

To offer some perspective on the following, the statement we are driving at is closely related to two standard isomorphisms of complex and real Clifford algebras, namely $\mathrm{Cl} (\mathbb{C}^{s+2}) \simeq \mathrm{Cl} (\mathbb{C}^2) \otimes \mathrm{Cl}(\mathbb{C}^s)$ and $\mathrm{Cl} (\mathbb{R}^{s+1, 1}) \simeq \mathrm{Cl} (\mathbb{R}^{1,1}) \otimes \mathrm{Cl}(\mathbb{R}^s)$.

Let there be Clifford algebra generators $j_1 , \ldots, j_s$ that satisfy the relations (\ref{Eq:Cliff-s}) and are orthogonal unitary transformations of $W = \mathbb{C}^{2n}$, which means that they preserve $\langle \, , \, \rangle_W$ and $\{ \, , \, \}_W$. Then we take the tensor product of $W$ with $\mathbb{C}^2$ and pass to a Clifford algebra with $s+2$ generators $J_1, \ldots, J_{s+2}$ defined on $\mathbb{C}^2 \otimes W$ by
\begin{equation}\label{eq:conn-rels}
    \begin{split}
    J_l &= \begin{pmatrix} 0 &1\cr 1 &0 \end{pmatrix} \otimes j_l \quad (l=1, \ldots, s) , \\ J_{s+1} &= \begin{pmatrix} 0 &1\cr -1 &0 \end{pmatrix} \otimes \mathrm{Id}_W , \quad J_{s+2} = \begin{pmatrix} \mathrm{i} &0 \cr 0 &- \mathrm{i} \end{pmatrix} \otimes \mathrm{Id}_W . \end{split}
\end{equation}
The Hermitian scalar product and the CAR bracket of $W$ are transferred to the doubled space in the natural way explained in Section \ref{sect:s-geq-3}. Note that on $\mathbb{C}^2 \otimes W$ we have the two involutions $A \mapsto A^\mathrm{c}$ and $A \mapsto A^\perp$ as before.

We now observe that all Clifford algebra generators $J_1, \ldots, J_{s+2}$ are orthogonal unitary transformations of $\mathbb{C}^2 \otimes W$ but for the distinguished generator $K = J_{s+2}\,$, which is unitary but \emph{sign-reverses} the extended CAR bracket:
\begin{equation}
    \{ K w , K w^\prime \}_{\mathbb{C}^2 \otimes W} = - \{ w , w^\prime \}_{\mathbb{C}^2 \otimes W} .
\end{equation}
We call $K$ ``imaginary'' while using the adjective ``real'' for the generators $J_1, \ldots, J_{s+1}\,$.

Let us note the alternative option of working with the modified generator $\mathrm{i}K$ instead of $K$. The former would be a bona fide orthogonal transformation of $\mathbb{C}^2 \otimes W$, but it has square plus one, and one would call it ``positive'' (in contradistinction with the ``negative'' generators $J_1, \ldots, J_{s+1}$) as is done in \cite{kitaev,StoneEtAl}. We prefer the present convention of a negative but imaginary generator $K$, as it will render our later discussion of the diagonal map more concise.

We are now ready to get to the point. Let us recall from Section \ref{sect:ClassMaps} the spaces ($s \geq 0$)
\begin{align}
    &C_s(n) = \{ a \subset W \mid j_l \, a = a^\mathrm{c} ; \; l = 1, \ldots, s \} \label{eq:def-Cq}, \\ &R_s(n) = \{ a \in C_s(n) \mid a = a^\perp \} . \label{eq:def-Rq}
\end{align}
By $C_0(n)$ we simply mean the space of all complex $r$-planes ($0 \leq r \leq 2n$) in $W = \mathbb{C}^{2n}$. The subspace $R_0(n)$ consists of all complex planes $a \subset W$ with the ``Lagrangian'' property $a = a^\perp$; such planes are necessarily of dimension $n$. The spaces $C_s$ and $R_s$ will be called the ``complex'' and ``real'' classifying spaces for our $\IQPV$s (or vector bundles) of class $s$. As major players of our classification work they were tabulated in Table \ref{table:CsRs} of Section \ref{sect:ClassMaps}.

Next, we decree the corresponding definitions at the level of the doubled space $\mathbb{C}^2 \otimes W$:
\begin{align}
    &C_{s+2}(2n) = \{ A \subset \mathbb{C}^2 \otimes W \mid J_l \, A = A^\mathrm{c}; \; l = 1, \ldots, s+2\} \label{eq:def-Cq2}, \\ &R_{s+1,1}(2n) = \{ A \in C_{s+2}(2n) \mid A = A^\perp \} .
    \label{eq:def-Rq2}
\end{align}
Here the more elaborate notation $R_{s+1,1}$ reflects the fact that the generators $J_1, \ldots, J_{s+1}$ are real, whereas the last generator $J_{s+2} = K$ is imaginary.

Now consider the mapping $f : \; a \mapsto A$ defined by Eq.\ (\ref{eq-mz:2.31}) of Section \ref{sect:s-geq-3}. It is clear that $f$ is a map from $C_s(n)$ to $C_{s+2}(2n)$. Indeed, one easily checks that for $a \in C_s(n)$ the image plane $A = f(a)$ satisfies the relations $J_l \, A = A^\mathrm{c}$ ($l = 1, \ldots, s+2$) of $C_{s+2}(2n)$. Moreover, from $a = a^\perp$ one deduces that $A = A^\perp$. Thus $f$ restricts to a map $f^\prime :\; R_s(n) \to R_{s+1,1}(2n)$. The statement of $(1,1)$ periodicity is now as follows.
\begin{proposition}\label{prop:2.1X}
If $f : \; C_s(n) \to C_{s+2}(2n)$ is the mapping defined by Eq.\ (\ref{eq-mz:2.31}), then both this map and its restriction $f^\prime :\; R_s(n) \to R_{s+1,1}(2n)$ are bijective.
\end{proposition}

The proof of the proposition will consume the rest of this subsection. What remains to be shown is that, given the Clifford algebra generators $J_1, \ldots, J_{s+2}$ on the doubled space, $\widetilde{W}$, one can reconstruct the original framework built on the generators $j_1, \ldots, j_s$ on $W$ so as to invert the mapping $f : \; a \mapsto A$. For the inverse direction, we may not assume the decomposition $\widetilde{W} = \mathbb{C}^2 \otimes W$ and the connecting relations (\ref{eq:conn-rels}) but must construct them. This is done in the following. We begin with some preparation and state a useful lemma along the way.

Like the other generators, the distinguished operator $K = J_{s+2}$ is unitary and anti-Hermitian and has eigenvalues $\pm \mathrm{i}$. Let the corresponding eigenspaces be denoted by
\begin{equation}
    W_\pm = E_{\pm \mathrm{i}}(K) .
\end{equation}
Note that all operators $J_1, \ldots, J_{s+1}$ exchange these spaces: $J_l W_\pm = W_\mp$ ($l = 1, \ldots, s+1$) and that $\dim W_+ = \dim W_- \,$. The idea of the sequel is to carry out a reduction from $\widetilde{W}$ to \begin{equation}
    W_+ \equiv W .
\end{equation}

First of all, the non-degenerate symmetric bilinear form $\{ \, , \, \}$ (the CAR pairing) given on $\widetilde{W}$ descends by restriction to a non-degenerate symmetric bilinear form
\begin{equation}
    \{ \, , \, \}_W : \; W \times W \to \mathbb{C} .
\end{equation}
Indeed, if $w_+ \in W_+$ and $w_- \in W_-$, then
\begin{equation*}
    \{w_+ , w_-\} = \{ \mathrm{i} w_+ , - \mathrm{i} w_-\} = \{K w_+ , K w_-\} = - \{w_+ ,w_-\} = 0 ,
\end{equation*}
since $K$ sign-reverses the CAR pairing. By similar reasoning, the Hermitian scalar product $\langle \, , \, \rangle : \; \widetilde{W} \times \widetilde{W} \to \mathbb{C}$ descends to a Hermitian scalar product
\begin{equation}
    \langle \; , \, \rangle_W : \; W \times W \to \mathbb{C} .
\end{equation}
It follows that the complex anti-linear involution $\gamma : \; \widetilde{W} \to \widetilde{W}$ restricts to a similar involution
$\gamma : \; W \to W$ by the defining equation $\{ w , w^\prime \}_W = \langle \gamma \, w , w^\prime \rangle_W$.

Now let $J_{s+1} \equiv I$, and let $\mathrm{Gr}_{2n}(\widetilde{W})$ be the Grassmann manifold of complex $2n$-planes in $\widetilde{W} \simeq \mathbb{C}^{4n}$. Consider then any $2n$-plane $A \in \mathrm{Gr}_{2n} (\widetilde{W})$ that obeys the orthogonality relations
\begin{equation}
    I A = K A = A^\mathrm{c} .
\end{equation}
Writing $L \equiv \mathrm{i} I K$, observe that $L^2 = \mathrm{Id}_{ \widetilde{W}}$ and $L\, A = A$. It follows that $A$ has an orthogonal decomposition by $L$-eigenspaces:
\begin{equation}
    A = \big(A\cap E_{+1}(L)\big) \oplus \big( A \cap E_{-1}(L) \big).
\end{equation}
As we shall see, $A$ is already determined by one of the two summands, say $A \cap E_{+1} (L)$. To show that, consider the operator $\Pi = {\textstyle{\frac{1}{2}}} (\mathrm{Id} - \mathrm{i} K)$ of orthogonal projection from $\widetilde{W}$ to $W$, and let $A^{(\pm)} \subset W$ be the image of $A \cap E_{\pm 1}(L)$ under the projector $\Pi$.
\begin{lemma}
The linear maps $\Pi : \; A\cap E_{\pm 1}(L) \to A^{(\pm)}$ are bijective. The space $A^{(-)}$ is the orthogonal complement of $A^{(+)}$ in $W = E_{+\mathrm{i}} (K)$.
\end{lemma}
\begin{proof}
Every $v \in A \cap E_{+1}(L)$ is of the form $v = w + L w$ with $w \in A$. If $v = v_+ + v_-$ and $w = w_+ + w_-$ are the orthogonal decompositions of $v\,, w$ by $\widetilde{W} = W_+ \oplus W_-\,$, then
\begin{equation*}
    v_+ = w_+ + L w_- \quad \text{and} \quad v_- = w_- + L w_+ = L v_+ \,,
\end{equation*}
because $L$ anti-commutes with $K$ and hence exchanges $W_+$ with $W_-\,$. The map $\Pi : \; v \mapsto v_+$ is surjective by the definition of $A^{(+)}$. It is also injective since $v_+ = 0$ implies $v = v_+ + L v_+ = 0$ and therefore $w \in E_{-1}(L)$. Thus the map $\Pi : \; A \cap E_{+1}(L) \to A^{(+)}$ is an isomorphism of vector spaces. The argument for $\Pi : \; A \cap E_{-1}(L) \to A^{(-)}$ is similar.

To prove the second statement, let $w \in A^{(+)}$ and $w^\prime \in A^{(-)}$. Then $w + L w \in A$ and
\begin{equation*}
    w^\prime + L w^\prime = - \mathrm{i} K ( w^\prime - L w^\prime ) \in K A = A^\mathrm{c} ,
\end{equation*}
and from $\left\langle A , A^\mathrm{c} \right\rangle = 0$ we infer that
$0 = \left\langle w + L w\,, w^\prime + L w^\prime \right\rangle = 2 \left\langle w , w^\prime \right\rangle$. Thus $A^{(+)}$ and $A^{(-)}$ are orthogonal to each other. Because of
\begin{equation*}
    \dim A^{(+)} + \dim A^{(-)} = \dim A \cap E_{+1}(L) + \dim A \cap E_{-1}(L) = \dim A = \dim W ,
\end{equation*}
$A^{(+)} \subset W$ and $A^{(-)} \subset W$ are in fact orthogonal complements of each other.
\end{proof}

As an immediate consequence, we have:
\begin{corollary}\label{cor:2.1}
The vector space $\widetilde{W}$ has an orthogonal decomposition by the following four subspaces:
\begin{align*}
    A \cap E_{+1}(L) &= \{ w + L w \mid w \in A^{(+)} \} , \quad A^\mathrm{c} \cap E_{+1}(L) = \{w + L w \mid w \in A^{(-)} \} , \cr
    A \cap E_{-1}(L) &= \{ w - L w \mid w \in A^{(-)} \} , \quad A^\mathrm{c} \cap E_{-1}(L) = \{ w - L w \mid w \in A^{(+)} \} .
\end{align*}
\end{corollary}
\begin{remark}
The vector spaces $A^{(+)}$ and $A^{(-)}$ need not have the same dimension; in particular, either one of them may be the zero vector space.
\end{remark}

We now carry out a reduction based on the isomorphism $\Pi : \; A\cap E_{+1}(L) \to A^{(+)}$. For this we observe that the relations $J_l \, A = A^\mathrm{c}$ have the following refinement:
\begin{equation}\label{eq-mz:2.29}
    L\, J_l \big(A \cap E_{\pm 1}(L)\big) = J_l \big(A \cap E_{\pm 1}(L) \big) = A^\mathrm{c} \cap E_{\pm 1}(L) \quad (l=1, \ldots, s) ,
\end{equation}
due to the fact that $J_l$ commutes with $L = \mathrm{i} I K$. We recall that the operators $J_1, \ldots, J_s$ exchange the subspaces $W_\pm = E_{ \pm\mathrm{i}}(K)$. The operators $L\, J_1, \ldots, L\, J_s$ then preserve these subspaces and hence commute with the projector $\Pi$. By applying $\Pi$ to Eq.\ (\ref{eq-mz:2.29}) and using Corollary \ref{cor:2.1} it follows that
\begin{equation}
    L\, J_l \, A^{(\pm)} = A^{(\mp)} \quad (l = 1, \ldots, s) .
\end{equation}

In view of the above, we introduce the restricted operators
\begin{equation}\label{eq:4.15-mz}
    j_l := L\, J_l \big\vert_W \quad (l = 1, \ldots, s) .
\end{equation}
Note that the $j_l$ inherit from the $J_l$ the algebraic relations
\begin{equation}
    j_l j_m + j_m j_l = - 2 \delta_{l m} \mathrm{Id}_W \quad (l,m = 1, \ldots, s).
\end{equation}

\medskip\noindent\textbf{Proof of Proposition \ref{prop:2.1X}.} --- Starting with $A \in C_{s+2}(2n)$ we form $A \cap E_{+1}(L)$ and apply the projector $\Pi$ to obtain $a \equiv A^{(+)} \subset W$. By this process, the relations $J_l\, A = A^\mathrm{c}$ turn into the relations $j_l\, a = a^\mathrm{c}$ for $l = 1, \ldots, s$. Thus $a$ is a point of $C_s(n)$.

Now if $v, v^\prime \in a$, then $w = v + Lv$ and $w^\prime = v^\prime + L v^\prime$ lie in $A$, and we have
\begin{equation*}
    2 \{ v , v^\prime \}_W = \{ v , v^\prime \}_{W_+}  +
    \{ L v , L v^\prime \}_{W_-} = \{ w , w^\prime \}_{\widetilde{W}} ,
\end{equation*}
because $L = \mathrm{i} I K$ preserves the CAR bracket. Therefore, $A = A^\perp$ implies $\{a,a\}_W = 0$. By the same reasoning, we have $\{ a^\mathrm{c} , a^\mathrm{c} \}_W = 0$, since $A = A^\perp$ entails $A^\mathrm{c} = (A^\mathrm{c})^\perp$. Now the combination of $\{ a , a \}_W = 0$ with $\{ a^\mathrm{c} , a^\mathrm{c} \}_W = 0$ implies that $a$ is half-dimensional: $\dim a = \dim a^\mathrm{c} = \frac{1}{2} \dim W$.
Hence $a = a^\perp$. Thus the mapping $C_{s+2}(2n) \to C_s(n)$ by $A \mapsto a$ restricts to a mapping from $R_{s+1,1}(2n)$ to $R_s(n)$. This inverts the maps $f :\; C_s(n) \to C_{s+2}(2n)$ and $f^\prime :\; R_s(n) \to R_{s+1,1}(2n)$ and completes the proof of the proposition.

\subsection{$\mathbb{Z}_2$-equivariant Bott map}\label{sect:Z2-Bott}

We now turn to the construction of the $\mathbb{Z}_2$-equivariant Bott map, or diagonal map for short. Fixing any symmetry index $s \geq 0$, we are given a pair of classifying spaces $C_s(n)$ and $R_s(n)$. We then apply to them the $(1,1)$ periodicity theorem in the expanding direction. That is, starting from $s$ real generators $j_1, \ldots, j_s$ on $W$, we follow Section \ref{sect:1,1-red} to pass to an extended system of $s + 2$ generators $J_1, \ldots, J_s , I , K$ on $\mathbb{C}^2 \otimes W$.

The ensuing construction begins with the space
\begin{align*}
    C_s(2n) = \{ A \subset \mathbb{C}^2 \otimes W \mid J_1 A = \ldots = J_s A = A^\mathrm{c} \} .
\end{align*}
By imposing on it the two additional pseudo-symmetries, first $I$ and subsequently $K$, we get a sequence of inclusions $C_{s+2}(2n) \subset C_{s+1}(2n) \subset C_s(2n)$ where
\begin{equation}\label{eq:def-Cs+2}
    \begin{split}
    C_{s+1}(2n) &= \{ A \in C_s(2n) \mid I A = A^\mathrm{c} \} , \cr
    C_{s+2}(2n) &= \{ A \in C_{s+1}(2n) \mid K A = A^\mathrm{c} \} .
    \end{split}
\end{equation}
{}From Proposition \ref{prop:2.1X} we recall the bijection $C_s(n) \simeq C_{s+2}(2n)$.

As before, we denote the subspaces of fixed points of the Fermi involution $A \mapsto A^\perp$ by $R_j(2n) \subset C_j(2n)$ (for $j = s, \, s+1$). Now the last one of the $s+2$ Clifford algebra generators, namely $K$, is imaginary, i.e., it sign-reverses the CAR pairing. While this is of no relevance for the spaces above, it does matter for the subspace of fixed points of the Fermi involution in $C_{s+2}(2n)$. Recall that this subspace is denoted by
\begin{equation}\label{eq:def-Rs+1,1}
    R_{s+1,1}(2n) = \{ A \in C_{s+2}(2n) \mid A = A^\perp \} .
\end{equation}
By construction, we have a bijective correspondence $R_{s+1,1}(2n) \simeq R_s(n)$; cf.\ Prop.\ \ref{prop:2.1X}. In the sequel, to avoid cluttering our notation we do not introduce a special symbol for the bijections $C_s(n) \simeq C_{s+2}(2n)$ and $R_s(n) \simeq R_{s+1,1}(2n)$ but assume them to be understood, for the most part; we spell them out occasionally to minimize any risk of confusion.

Next, for any $\mathbb{C}$-linear operator $X$ on $\mathbb{C}^2 \otimes W$, let $X \mapsto X^T$ denote the operation of taking the transpose w.r.t.\ the CAR pairing, i.e.\ $\{ X^T w , w^\prime \} = \{ w , X w^\prime \}$ for all $w, w^\prime \in \mathbb{C}^2 \otimes W$.
\begin{lemma}\label{lem:tau-car}
If an automorphism $\tau_\mathrm{car}$ of $\mathrm{GL}(\mathbb{C}^2 \otimes W)$ is defined by $\tau_\mathrm{car}(g) = (g^{-1})^T$, then for any subvector space $A \subset \mathbb{C}^2 \otimes W$ one has the relation
\begin{equation}
    (g \cdot A)^\perp = \tau_\mathrm{car}(g) \cdot A^\perp .
\end{equation}
\end{lemma}
\begin{proof}
By definition, the vectors of $(g \cdot A)^\perp$ have zero CAR pairing with those of $g \cdot A$. It follows that $g^T \cdot (g \cdot A)^\perp = A^\perp$ and hence $(g \cdot A)^\perp = (g^T)^{-1} \cdot A^\perp = \tau_\mathrm{car}(g) \cdot A^\perp$.
\end{proof}

To prepare the next formula, let each complex hyperplane $A$ in $\mathbb{C}^2 \otimes W$ be associated with an anti-Hermitian operator \begin{equation}
    J(A) = \mathrm{i}(\Pi_A - \Pi_{A^\mathrm{c}}) ,
\end{equation}
where $\Pi_A$ and $\Pi_{A^\mathrm{c}}$ project on $A$ and its orthogonal complement $A^\mathrm{c}$, respectively. (Note that $J(A)^2 = - \mathrm{Id}$, and $\mathrm{i} J(A)$ corresponds to the flattened Hamiltonian of Remark \ref{remark:orthogonalunitary}.) Then, for $A \in C_{s+2}(2n)$ and $t \in [0,1]$ consider the one-parameter family of complex $2n$-planes
\begin{equation}\label{eq:Z2-Bottmap}
    \beta_t(A) = \mathrm{e}^{(t\,\pi/2) K J(A)} \cdot E_{+\mathrm{i}}(K) \,.
\end{equation}
Recall that $E_{+\mathrm{i}}(K) \subset \mathbb{C}^2 \otimes W$ denotes the eigenspace of $K$ with eigenvalue $+ \mathrm{i}$. Since the Clifford generators $J_1, \ldots, J_s$, and $I$ anti-commute with $K$, they swap the two eigenspaces $E_{+\mathrm{i}}(K)$ and $E_{-\mathrm{i}}(K) = E_{+\mathrm{i}}(K)^\mathrm{c}$. This means that the two $2n$-planes $E_{\pm\mathrm{i}}(K)$ lie in $C_{s+1}(2n)$.
\begin{lemma}\label{lem:4.1}
The assignment $[0,1] \ni t \mapsto \beta_t(A)$ for $A \in C_{s+2}(2n)$ is a curve in $C_{s+1}(2n)$ with initial point $\beta_0(A) = E_{+\mathrm{i}} (K)$, final point $\beta_1(A) = E_{-\mathrm{i}}(K)$, and midpoint $\beta_{1/2}(A) = A$. It is $\mathbb{Z}_2$-equivariant in the sense that $\beta_t(A)^\perp = \beta_{1-t}(A^\perp)$.
\end{lemma}
\begin{proof}
First of all, note that $K A = A^\mathrm{c}$ implies $K J(A) = - J(A) K$. Now because the unitary operator $\mathrm{e}^{(t\,\pi/2) K J(A)}$ commutes with each of the generators $J_1, \ldots, J_s, I$, it is immediate from the definition (\ref{eq:def-Cs+2}) of $C_{s+2}(2n)$ that $\beta_t(A)$ satisfies the pseudo-symmetry relations
\begin{equation*}
    J_1 \beta_t(A) = \ldots = J_s \beta_t(A) = I \beta_t(A) = \beta_t(A)^\mathrm{c} \,.
\end{equation*}
Thus $\beta_t(A)$ lies in $C_{s+1}(2n)$. To see that the curve ends at $E_{-\mathrm{i}}(K)$, we recall that both $K$ and $J(A)$ square to minus the identity, and they anti-commute. This gives $(K J(A))^2 = - \mathrm{Id}$ and
\begin{equation*}
    \beta_1(A) = \mathrm{e}^{(\pi/2) K J(A)} \cdot E_{+\mathrm{i}}(K) = \sin(\pi/2) K J(A) \cdot E_{+\mathrm{i}}(K) = J(A) \cdot E_{+\mathrm{i}}(K) = E_{-\mathrm{i}}(K) ,
\end{equation*}
since $J(A)$ swaps the eigenspaces of $K$.

To verify the midpoint property of $\beta_{1/2}(A) = A$, we compute
\begin{equation*}
    \mathrm{e}^{(\pi/4) K J(A)} = \cos(\pi/4)\, \mathrm{Id}_W + \sin(\pi/4)\, K J(A) = (\mathrm{Id}_W + K J(A)) / \sqrt{2} \,.
\end{equation*}
Applying this (with the factor of $1/\sqrt{2}$ omitted) to any $w \in E_{+\mathrm{i}}(K)$ we get
\begin{equation*}
    (\mathrm{Id}_W + K J(A)) w = w - \mathrm{i} J(A) w  = - \mathrm{i} J(A) (w - \mathrm{i} J(A) w) \in E_{+\mathrm{i}}(J(A)) = A .
\end{equation*}
The linear transformation $\mathrm{e}^{(\pi/4) K J(A)} : \; E_{+\mathrm{i}}(K) \to A$, $w \mapsto w - \mathrm{i} J(A) w$, is an isomorphism because $J(A) \cdot E_{+ \mathrm{i}}(K) = E_{-\mathrm{i}}(K)$. Hence
\begin{equation*}
    \beta_{1/2}(A) = \mathrm{e}^{(\pi/4) K J(A)} \cdot E_{+\mathrm{i}}(K) = A .
\end{equation*}

Turning to the last stated property, we note that the automorphism $\tau_\mathrm{car}$ of Lemma \ref{lem:tau-car} sends $J(A)$ to $J(A^\perp)$. Since $K$ is imaginary, we have $\tau_\mathrm{car}(K) = - K$ and $E_{+\mathrm{i}}(K)^\perp = E_{-\mathrm{i}}(K)$. Therefore,
\begin{equation*}
    \beta_t(A)^\perp = \tau_\mathrm{car} \big( \mathrm{e}^{(t\,\pi/2) K J(A)} \big) \cdot E_{+\mathrm{i}}(K)^\perp = \mathrm{e}^{(- t\,\pi/2) K J(A^\perp)} \cdot E_{-\mathrm{i}}(K) = \beta_{1-t}(A^\perp) .
\end{equation*}
Thus $t \mapsto \beta_t(A)$ is $\mathbb{Z}_2$-equivariant in the stated sense.
\end{proof}

To summarize, our map $t \mapsto \beta_t(A) \in C_{s+1}(2n)$ is a $\mathbb{Z}_2$-equivariant curve (actually, a minimal geodesic in the natural Riemannian geometry of $C_{s+1}(2n)$) which joins the invariable point $E_{+\mathrm{i}}(K)$ with its antipode $E_{-\mathrm{i}}(K)$ by passing through the variable point $A$ at $t = 1/2$.

Let the space of all paths in $C_{s+1}(2n)$ from $E_{+\mathrm{i}}(K)$ to $E_{-\mathrm{i}}(K)$ be denoted by $\Omega_K C_{s+1}(2n)$. Then as an immediate consequence of Lemma \ref{lem:4.1} we have:
\begin{corollary}\label{cor:4.1}
Equation (\ref{eq:Z2-Bottmap}) defines a mapping $\beta$,
\begin{equation}\label{eq:def-beta}
     \beta :\; C_s(n) \simeq C_{s+2}(2n) \to \Omega_K C_{s+1}(2n), \quad f(a) = A \mapsto \{ t \mapsto \beta_t(A) \} ,
\end{equation}
from the classifying space $C_s(n)$ to the path space $\Omega_K C_{s+1}(2n)$. By its property of $\mathbb{Z}_2 $-equi\-variance, $\beta$ induces a mapping between the sets of $\mathbb{Z}_2$-fixed points:
\begin{equation}\label{eq:def-betap}
    \beta^\prime :\; C_s(n)^{\mathbb{Z}_2} \equiv R_s(n) \simeq R_{s+1,1}(2n) \to \left( \Omega_K C_{s+1}(2n) \right)^{\mathbb{Z}_2}\,.
\end{equation}
\end{corollary}
\begin{remark}
The non-trivial element of $\mathbb{Z}_2$ acts on $C_{s+1}(2n)$ by $A \mapsto A^\perp$ and on the interval $[0,1]$ by $t \mapsto 1-t$. There is an induced action of $\mathbb{Z}_2$ on the space of paths $\Omega_K C_{s+1}(2n)$. The symbol $\left( \Omega_K C_{s+1}(2n) \right)^{ \mathbb{Z}_2}$ denotes the subspace of paths that are fixed by this $\mathbb{Z}_2$-action.
\end{remark}
\begin{remark}
A mapping of this kind appeared already in the work of Teo \& Kane \cite{TeoKane}.
\end{remark}

Before continuing with the general development, let us give two examples illustrating $\beta$.

\subsection{Example 1: from (\emph{d},\,\emph{s}) = (0,0) to (1,1)}
\label{sect:4.3}

We start from data of class $D$ with zero-dimensional momentum space $M=\mathrm{S}^0$ consisting of two points, both of which are fixed by $\tau$. We now apply $\beta$ to manufacture a superconducting ground state with time-reversal invariance (class $\DIII$) in one dimension. Taking the simple case of $W = \mathbb{C}^2$ (or $n = 1$), we have a real classifying space consisting of just two points,
\begin{equation*}
    R_0(1) = \{ \mathbb{C} \cdot c \, , \mathbb{C} \cdot c^\dagger \} ,
\end{equation*}
which correspond to the empty and the fully occupied state, $|0\rangle$ and $| 1 \rangle$, respectively.

The procedure of doubling by $(1,1)$ periodicity here amounts to forming the tensor product with the two-dimensional spinor space, $(\mathbb{C}^2 )_\mathrm{spin}\,$. As input $A \in R_{1,1}(2) \simeq R_0(1)$ we take the complex line of the state with both spin states occupied:
\begin{equation*}
    A = \mathrm{span}_\mathbb{C} \{ c_\uparrow^\dagger \, , c_\downarrow^\dagger \} .
\end{equation*}
The operator $I$ is to be identified with the first pseudo-symmetry $J_1$ of the Kitaev sequence,
\begin{align*}
    I \equiv J_1 = \gamma\, T = (\sigma_1)_\mathrm{BdG} \otimes (\mathrm{i}\sigma_2)_\mathrm{spin} \,,
\end{align*}
where the left tensor factor (denoted by ``BdG'' for Bogoliubov-deGennes) acts in the two-dimensional quasi-spin space with basis $c$ and $c^\dagger$. The simplest choice of imaginary generator $K$ is
\begin{align*}
    K = \mathrm{i} (\sigma_1)_\mathrm{BdG} \otimes (\sigma_1)_\mathrm{spin} \,.
\end{align*}
We then apply the one-parameter group of $\beta$ to produce an $\IQPV$ of class $s = 1$. By using $\beta_{1/2}(A) = A$ and switching from the path parameter $t \in [0,1]$ to the momentum parameter $k = \pi (t - 1/2)$, we write the fibers $A_{k(t)} = \mathrm{e}^{(t\, \pi/2) K J(A)} \cdot E_{+\mathrm{i}}(K)$ as
\begin{equation*}
    A_k = \mathrm{e}^{(k/2) K J(A)} \cdot A = \mathrm{span}_\mathbb{C} \left\{ c_{-k,\sigma}^\dagger \cos(k/2) - c_{k,-\sigma} \sin(k/2) \right\}_{\sigma = \uparrow, \, \downarrow} \,.
\end{equation*}
(For a more informed perspective on this construction, please consult Remark \ref{rem:5.1} below.) To the physics reader this may look more familiar when written as a BCS-type ground state:
\begin{equation*}
    \vert \text{g.s.} \rangle = \mathrm{e}^{\; \sum_k \cot(k/2) P_k} \vert \mathrm{vac} \rangle \,, \quad P_k = c_{k,\uparrow}^\dagger c_{-k,\downarrow}^\dagger \,.
\end{equation*}
If the imaginary generator is chosen as $K = K(\alpha) = \mathrm{i} (\sigma_1)_\mathrm{BdG} \otimes \left( \sigma_1 \cos\alpha + \sigma_3 \sin\alpha \right)_\mathrm{spin}\,$, the Cooper pair operator $P_k$ takes the more general form
\begin{equation*}
    P_k = c_{k,\uparrow}^\dagger c_{-k,\downarrow}^\dagger \, \cos\alpha + \big( c_{k,\uparrow}^\dagger c_{-k,\uparrow}^\dagger -
    c_{k,\downarrow}^\dagger c_{-k,\downarrow}^\dagger \big) \, \sin\alpha \,,
\end{equation*}
which clearly displays the spin-triplet pairing of the superconductor at hand. The physical system is in a symmetry-protected topological phase, since the winding in its ground state cannot be undone without breaking the time-reversal invariance.

\subsection{Example 2: from (\emph{d},\,\emph{s}) = (1,1) to (2,2)}

To give a second example, we start from the outcome of the previous one and progress to a two-dimensional band insulator with conserved charge in class $\AII$. This time, the effect of $(1,1)$ doubling for the already spinful system is to introduce two bands, which we label by ${\rm p}$ and ${\rm h}$. To implement charge conservation directly and avoid working through all the details of $(1,1)$ doubling, we first perform a change of basis (by a particle-hole transformation) on our class-$\DIII$ superconductor to turn it into a particle-number conserving reference $\IQPV$ of class $s = 1$:
\begin{align*}
    A_{k_1} &= \mathrm{span}_\mathbb{C} \left\{ a_{k_1,\uparrow,+} \, , \, a_{k_1,\,\downarrow\,,-} \, , \, b_{-k_1, \,\downarrow\,,-}^\dagger \, , \, b_{- k_1,\uparrow,+}^\dagger \right\} , \cr
    a_{k_1,\,\sigma,\,\varepsilon} &=  c_{k_1, \,\sigma, \, {\rm p}} \, \cos(k_1 / 2) + \mathrm{i} \varepsilon \, c_{k_1, -\sigma ,\, {\rm h}} \, \sin(k_1 / 2) , \cr b_{k_1,\,\sigma,\,\varepsilon} &=  c_{k_1, \,\sigma, \,{\rm h}} \, \cos(k_1 / 2) - \mathrm{i} \varepsilon \, c_{k_1, -\sigma, \,{\rm p} } \, \sin(k_1 / 2) .
\end{align*}
We see that the change of basis has turned each quasi-particle operator into either an annihilation operator $a_{k_1, \bullet}$ or a creation operator $b_{-k_1, \bullet}^\dagger$. We stress that $\{ A_{k_1} \}$ is still the ground state in disguise of our one-dimensional class-$\DIII$ superconductor -- it has only been jacked up by $(1,1)$ periodicity. The pseudo-symmetry operators now are
\begin{align*}
    J_1 &= \gamma\, T = (\sigma_1)_{\rm BdG} \otimes (\mathrm{i} \sigma_2)_\mathrm{spin} \otimes \mathrm{Id}_{\rm ph}\,, \cr
    I &= J_2 = \mathrm{i} Q J_1 = (\sigma_2)_{\rm BdG} \otimes (\mathrm{i}\sigma_2)_\mathrm{spin} \otimes \mathrm{Id}_{\rm ph}\,, \cr K &= \mathrm{i} \, \mathrm{Id}_{\rm BdG} \otimes (\sigma_1)_\mathrm{spin} \otimes (\sigma_1)_{\rm ph}\,.
\end{align*}
In this representation, the operator $J(A_k)$ is expressed by
\begin{equation*}
    J(A_{k_1}) = \mathrm{i} (\sigma_3)_{\rm BdG} \otimes \left( \mathrm{Id}_\mathrm{spin} \otimes (\sigma_3)_\mathrm{ph} \cos k_1 + (\sigma_2)_\mathrm{spin} \otimes (\sigma_1)_\mathrm{ph} \sin k_1 \right) .
\end{equation*}

Now we again apply the one-parameter group of $\beta$, still with parameter $k_0 = \pi (t - 1/2)$ for $t \in [0,1]$. In this way we get a two-dimensional $\IQPV$ of class $s = 2$:
\begin{align*}
    A_k &= \mathrm{e}^{(k_0 / 2) K J(A_{k_1})} \cdot A_{k_1} =
    \mathrm{span}_\mathbb{C} \left\{ \widetilde{a}_{k,\uparrow,+} \, , \, \widetilde{a}_{k,\,\downarrow\,,-} \, , \, \widetilde{b}_{-k, \,\downarrow\,,-}^\dagger \, , \, \widetilde{b}_{- k,\uparrow ,+}^\dagger \right\} , \cr
    \widetilde{a}_{k,\,\sigma,\,\varepsilon} &= \left( c_{k, \,\sigma, \, {\rm p}} \, \cos(k_1 / 2) + \mathrm{i} \varepsilon \, c_{k,-\sigma,\,{\rm h}} \, \sin(k_1 / 2) \right) \cos(k_0/2) \cr &- \left( c_{k,-\sigma,\,{\rm h}} \, \cos(k_1 /2) + \mathrm{i} \varepsilon \, c_{k,\,\sigma,\,{\rm p}} \, \sin(k_1 / 2) \right) \sin(k_0/2) , \cr
    \widetilde{b}_{k,\sigma,\,\varepsilon} &=
    \left( c_{k, \,\sigma ,\, {\rm h}} \, \cos(k_1 / 2) - \mathrm{i} \varepsilon \, c_{k, -\sigma ,\, {\rm p}} \, \sin(k_1 / 2) \right) \cos(k_0/2) \cr &- \left( c_{k,-\sigma ,\, {\rm p}} \, \cos(k_1 / 2) - \mathrm{i} \varepsilon \, c_{k,\sigma ,\, {\rm h}} \, \sin(k_1 / 2) \right) \sin(k_0/2) ,
\end{align*}
where $k = (k_0 , k_1)$. Notice that there is a redundancy here: the four operators spanning $A_k$ are not independent; rather, the subspace of conduction bands ($\widetilde{a}$) is already determined by the subspace of valence bands ($\widetilde{b}$) and vice versa; cf.\ the discussion in Section \ref{sect:2.3.1}. This is the price to be paid for our comprehensive formalism handling all classes at once.

At $k_0 = \pm \pi/2$ -- the two poles of a two-sphere with polar coordinate $k_0 + \pi/2$ --, the $k_1$-dependence goes away by construction. These two points are easily seen to be the only points where the Kane-Mel\'e Pfaffian vanishes, which implies that our $\IQPV$ of class $s = 2$ has non-trivial Kane-Mel\'e invariant \cite{KaneMele} and lies in the quantum spin Hall phase.

Further examples illustrating the diagonal map $\beta$ can be found in \cite{RK-MZ-Nobel}.

\section{Homotopy theory for the diagonal map}\label{sect:5}
\setcounter{equation}{0}

In this section, we collect and develop a number of homotopy-theoretic results related to the diagonal map $\beta$. This is done en route to our goal of proving that $\beta$ induces the desired bijection in homotopy: a one-to-one mapping, $\beta_\ast^{\mathbb{Z}_2}$, between stable homotopy classes of base-point preserving and $\mathbb{Z}_2$-equivariant maps $f : \; M \to C_s(n)$ and $F : \; \tilde{S} M \to C_{s+1}(2n)$.

We have introduced $\beta$ somewhat informally, but now state precisely how it is to be used. Let $f : \; M \to C_s(n) \simeq C_{s+2}(2n)$, $k \mapsto A_k\,$, be the classifying map of an $\IQPV$ of class $s$. By composing $f$ with $\beta$ we get a new map
\begin{equation*}
    \beta \circ f : \quad M \to \Omega_K C_{s+1}(2n) , \quad k \mapsto \{ t \mapsto \beta_t(A_k) \} .
\end{equation*}
Recall that $\Omega_K C_{s+1}(2n)$ denotes the space of paths in $C_{s+1}(2n)$ from $E_{+\mathrm{i}}(K)$ to $E_{-\mathrm{i}}(K)$.

Next, we choose to view the path coordinate $t$ as a coordinate for the second factor in the direct product $M \times [0,1]$. We then re-interpret $\beta \circ f$ as a mapping
\begin{equation}\label{eq:5.12}
    F : \; M \times [0,1] \to C_{s+1}(2n) , \quad (k,t) \mapsto \beta_t(A_k) .
\end{equation}
Since $\beta_t$ degenerates at the two points $t = 0$ and $t = 1$, the mapping $F$ descends to a map (still denoted by $F$) from the suspension $\tilde{S} M = M \times [0,1] / (M \times \{0\} \cup M \times \{1\})$ into $C_{s+1}(2n)$:
\begin{equation}
    F : \; \tilde{S} M \to C_{s+1}(2n) .
\end{equation}
We let the non-trivial element of $\mathbb{Z}_2$ act on $\tilde{S} M$ by
\begin{equation}
    (k,t) \mapsto (\tau(k),1-t) \equiv \tau_{\tilde{S}M} (k,t) .
\end{equation}
Thus, $\tilde{S} M$ is the ``momentum-type'' suspension of $M$. (The symbol $S M$ is reserved for the position-type suspension invoked later on.) Then, since $f :\; M \to C_s(n) \simeq C_{s+2}(2n)$ is $\mathbb{Z}_2$-equivariant, so is the new map: $\tau_{s+1} \circ F = F \circ \tau_{\tilde{S}M}\,$. Indeed,
\begin{equation*}
    \tau_{s+1}(F(k,t)) = \beta_t(A_k)^\perp = \beta_{1-t} (A_k^\perp) = \beta_{1-t}(A_{\tau(k)}) = (F \circ \tau_{\tilde{S}M})(k,t).
\end{equation*}
Thus, starting from an $\IQPV$ of class $s$ over the $d$-dimensional space $M$, the composition with $\beta$ produces an $\IQPV$ of class $s+1$ over the $(d+1)$-dimensional space $\tilde{S} M$.

In the sequel, we restrict all discussion to the case of topological spaces with base points, say $(X,x_\ast)$ and $(Y,y_\ast)$, and to base-point preserving maps $f : \; X \to Y$, $f(x_\ast) = y_\ast$. Borrowing the language from physics, this means that there is (at least) one distinguished momentum $k_\ast \in M$ whose fiber $A_{k_\ast}$ is not free to vary but is kept fixed: $A_{k_\ast} \equiv A_\ast\,$. This condition is physically well motivated in many cases. For example, for a superconductor one takes $k_\ast$ to be a momentum far outside the Fermi surface of the underlying metal, and $A_\ast = U$ is then the ``vacuum'' space spanned by the bare annihilation operators.

We adopt the convention of denoting by $[X,Y]_\ast$ the set of homotopy classes of base-point preserving maps $f$ from a topological space $(X,x_\ast)$ to another topological space $(Y,y_\ast)$. If $X$ and $Y$ are $G$-spaces (with base points that are fixed by $G$), the symbol $[X,Y]_\ast^G$ denotes the set of homotopy classes of $G$-equivariant and base-point preserving maps $f : \; X \to Y$.

In our concrete setting, we choose a base point $k_\ast = \tau(k_\ast)$ for $M$ and the corresponding $\tau_{\tilde{S}M}$-fixed point $(k_\ast, 1/2)$ as the base point of $\tilde{S} M$. Our classifying maps $f : \, M \to C_s(n) \simeq C_{s+2}(2n)$ are required to assign to $k_\ast$ a fixed fiber $f(k_\ast) = A_\ast = A_\ast^\perp \in R_{s+1,1}(2n) \simeq R_s(n)$.
\begin{proposition}\label{prop:5.1}
The mapping $\beta$ that sends $f : \; M \to C_s(n)$ to $F : \; \tilde{S} M \to C_{s+1}(2n)$ by Eq.\ (\ref{eq:5.12}) induces a mapping $\beta_\ast : \; [M , C_s (n)]_\ast \to [\tilde{S} M , C_{s+1}(2n)]_\ast$ between homotopy classes. The latter descends to a map
\begin{equation}\label{eq:prop5.1}
    \beta_\ast^{\mathbb{Z}_2} : \; [M , C_s (n)]_\ast^{\mathbb{Z}_2} \to [\tilde{S} M , C_{s+1}(2n)]_\ast^{\mathbb{Z}_2}
\end{equation}
between homotopy classes of $\mathbb{Z}_2$-equivariant maps.
\end{proposition}
\begin{proof}
If $f_u$ for $u \in [0,1]$ is a homotopy connecting $f_0$ with $f_1$, then by composing it with our continuous map $\beta$ we get a homotopy $F_u = \beta \circ f_u$ connecting $F_0$ with $F_1$. Thus $\beta$ induces a well-defined map $\beta_\ast$ between homotopy classes.

If $f(k_\ast) = A_\ast \in R_{s+1,1}(2n)$ then $F(k_\ast,1/2) = A_\ast \in R_{s+1}(2n)$ since $\beta_{t=1/2}$ is the identity map. Thus $F$ maps the base point $(k_\ast,1/2)$ of $\tilde{S} M$ to the base point $A_\ast$ of $R_{s+1}(2n) \subset C_{s+1}(2n)$. Moreover, if $f_u$ for $u \in [0,1]$ is a homotopy of $\mathbb{Z}_2$-equivariant maps from $M$ to $C_s(n)$, then $F_u$ given by $F_u(k,t) = (\beta_t \circ f_u)(k)$ is a homotopy of $\mathbb{Z}_2$-equivariant maps from $\tilde{S} M$ to $C_{s+1}(2n)$. Thus $\beta_\ast$ descends to $\beta_\ast^{\mathbb{Z}_2}$ as claimed.
\end{proof}
\begin{remark}\label{rem:5.1}
It is perhaps instructive to highlight the workings of $\beta$ and $\beta_\ast^{\mathbb{Z}_2}$ for a zero-dimensional momentum space consisting of two $\tau$-fixed points, $M = \{ k \in \mathbb{R} \mid k^2 = 1 \} \equiv \mathrm{S}^0$ (see the example of Section \ref{sect:4.3}). In this case the suspension $\tilde{S}(\mathrm{S}^0)$ can be regarded as the circle $\mathrm{S}^1 \subset \mathbb{C}$ of unitary numbers with involution $\tau_{\tilde{S}M}$ given by complex conjugation. This viewpoint is realized by the map
\begin{equation*}
    M\times [0,1]\ni (k,t) \mapsto \mathrm{i}\,\mathrm{e}^{-\mathrm{i} k t\pi} \in \mathrm{S}^1 \subset \mathbb{C} .
\end{equation*}
Now, starting from $f : \; \mathrm{S}^0 \to R_s(n) \simeq R_{s+1,1}(2n)$ with $f(k_\ast) = A_\ast$ and $f(-k_\ast) = A$ (for some choice of base point $k_\ast = \pm 1$), we apply $\beta$ to obtain $F: \; \mathrm{S}^1 \to C_{s+1}(2n)$ as
\begin{equation*}
    F(\mathrm{i}\,\mathrm{e}^{-\mathrm{i} k t\pi}) = \left\{ \begin{array}{ll} \beta_t(A_\ast), &\quad k = k_\ast \,, \cr \beta_t(A), &\quad k = -k_\ast \,. \end{array} \right.
\end{equation*}
It is easy to verify that $F$ is continuous and satisfies $F(\mathrm{e}^{ \mathrm{i}\theta})^\perp = F(\mathrm{e}^{- \mathrm{i} \theta})$. Thus $F$ is a $\mathbb{Z}_2$-equivariant loop $F :\; \mathrm{S}^1 \to C_{s+1} (2n)$. Half of the loop is determined by the choice of base point $A_\ast\,$; the other half is variable and parameterized by $A \in R_s(n)$. By the reasoning given above, this construction induces a mapping of homotopy classes,
\begin{equation*}
    \beta_\ast^{\mathbb{Z}_2} :\; \pi_0(R_s(n)) \equiv [\mathrm{S}^0 , C_s(n) ]_\ast^{ \mathbb{Z}_2} \to [\mathrm{S}^1 , C_{s+1} (2n) ]_\ast^{ \mathbb{Z}_2} .
\end{equation*}
\end{remark}

\subsection{Connection with complex Bott periodicity}

In the previous subsection we introduced a mapping in homotopy, $\beta_\ast^{\mathbb{Z}_2}$, which makes sense for any momentum space $M$ with an involution $\tau$. Our goal now is to show that, under favorable conditions, this map is bijective.

Let us recapitulate the situation at hand: we have a $\mathbb{Z}_2$-equivariant mapping $\beta : \; C_s(n) \to \Omega_K C_{s+1}(2n)$, cf.\ Eq.\ (\ref{eq:def-beta}), doubling the dimension of $W$ and increasing the symmetry index and the momentum-space dimension by one. The first step of the following analysis is to investigate $\beta$ as an unconstrained map; which is to say that we forget the $\mathbb{Z}_2$-actions on $C_s(n)$ and $\Omega_K C_{s+1}(2n)$ for the moment. Note that $\pi_d(Y) \equiv \pi_d(Y,y_\ast)$ denotes the homotopy group of base-point preserving maps from $\mathrm{S}^d$ into the topological space $(Y,y_\ast)$.
\begin{proposition}\label{prop:beta}
The induced map between homotopy groups,
\begin{equation}
    \beta_* : \; \pi_d \left( C_s(n)\right) \to \pi_d \left( \Omega_K C_{s+1}(2n) \right) ,
\end{equation}
is an isomorphism for $1 \leq d \ll \dim C_s(n)_0$ where $C_s(n)_0$ denotes the connected component of $C_s(n) \simeq C_{s+2}(2n)$ containing the base point $A_\ast\,$.
\end{proposition}
\begin{proof}
Our diagonal map (\ref{eq:def-beta}) sends a point $A \in C_{s+2}(2n) \simeq C_s(n)$ to a curve $t \mapsto \beta_t(A)$ connecting $E_{+\mathrm{i}}(K)$ with its antipode $E_{-\mathrm{i}}(K)$ in $C_{s+1}(2n)$. By the algebraic properties of $K$ and $J(A)$, the assignment $t \mapsto \mathrm{e}^{(t \pi / 2) K J(A)}$ is a one-parameter group of isometries of $C_{s+1}(2n)$ in its natural Riemannian geometry.
It follows that $t \mapsto \beta_t(A)$ is a geodesic in the path space $\Omega_K C_{s+1}(2n)$. This geodesic has minimal length in $\Omega_K C_{s+1}(2n)$ since the straight line $t \mapsto (t \pi / 2) K J(A)$ (in the Lie algebra of the group of isometries) for $t \in [0,1]$ lies in an injectivity domain for the exponential map. This means that, when the $\mathbb{Z}_2$-actions on $C_{s+2}(2n) \simeq C_s(n)$ and $\Omega_K C_{s+1}(2n)$ are ignored, our diagonal map $\beta$ is none other than the standard Bott map \cite{Bott1959} underlying complex Bott periodicity.
It is well known that this map induces an isomorphism $\beta_\ast$ of homotopy groups for $1 \leq d \ll \dim C_s(n)_0\,$.
\end{proof}
\begin{remark}
We reiterate that all our maps are understood to be base-point preserving. The base point $A_\ast$ of $C_{s+2}(2n) \simeq C_s(n)$ lies in $R_{s+1,1} (2n) \simeq R_s(n) \subset C_s(n)$. The base point of $\Omega_K C_{s+1} (2n)$ is the minimal geodesic from $E_{+\mathrm{i}}(K)$ to $E_{-\mathrm{i}}(K)$ through $A_\ast$.
\end{remark}
\begin{remark}
{}From the original paper by Bott \cite{Bott1959} one has quantitative bounds on $d$ in order for the Bott map $\beta_\ast$ to be an isomorphism. For odd $s$, in which case $C_s(n)$ is a unitary group, these are $2 \leq d + 1 \leq 2^{(3-s)/2} n$. For even $s$, the optimal bounds depend on the choice of connected component of $C_s(n)$. A detailed derivation is given in Section \ref{sect:8}.
\end{remark}
\begin{remark}
The case of dimension $d = 0$ needs separate treatment. For example, $C_0(n)$ has $2n+1$ connected components $\mathrm{Gr}_r(\mathbb{C}^{2n})$ ($0 \leq r \leq 2n$), whereas $\pi_1(C_1(2n)) = \pi_1(\mathrm{U}_{2n}) = \mathbb{Z}$. In the literature, this discrepancy is often finessed by approximating $C_0(n)$ by $\mathbb{Z} \times \mathrm{BU}$.
\end{remark}

\subsection{$G$-Whitehead Theorem}

The mapping $\beta$ under consideration is $\mathbb{Z}_2$-equivariant, and the question to be addressed now is whether it is a homotopy equivalence between topological spaces carrying $\mathbb{Z}_2$-actions. The main tool to simplify (if not answer) this question is the so-called $G$-Whitehead Theorem, a standard homotopy-theoretic result that we now quote for the reader's convenience. Although we will be concerned only with the case of $G = \mathbb{Z}_2$, we will state the theorem for any group $G$. To do so in a concise way, we need to introduce some terminology first.
\begin{definition}
Let $X$ and $Y$ be topological spaces with base points. If a base-point preserving mapping $f : \; X \to Y$ induces isomorphisms $f_\ast : \; \pi_d(X) \to \pi_d(Y)$, $[g] \mapsto [f \circ g]$, for $d < m$ and a surjection $f_\ast : \; \pi_m(X) \to \pi_m(Y)$, then one says that $f$ is $m$-connected.
\end{definition}
\begin{example}
For $s$ odd, our mapping $\beta : \; C_s(n) \to \Omega_K C_{s+1}(2n)$ is $m$-connected with $m = 2^{(3-s)/2} n - 1$.
\end{example}

The statement of the $G$-Whitehead Theorem makes use of the notion of a $G$-CW complex, which we assume to be understood; see \cite{Hatcher-AT} for an introduction. (This reference deals with the case of the trivial group $G = \{ e \}$. For the case of a general group $G$, see \cite{Greenlees-May,TomDieck}.) A fact of importance for us is that all products of spheres with factor-wise $\mathbb{Z}_2$-action are $\mathbb{Z}_2$-CW complexes, as this covers all cases considered later.

Suppose, then, that we are given a $G$-equivariant mapping $f : \; Y \to Z$ between $G$-spaces. If $X$ is another $G$-space, consider the mapping induced by $f$,
\begin{equation*}
    f_* : \; [X,Y]_\ast^G \to [X,Z]_\ast^G ,
\end{equation*}
between homotopy classes of $G$-equivariant maps. For any subgroup $H$ of $G$, let $Y^H$ be the set of fixed points of $H$ in $Y$. Because $f$ is $G$-equivariant and hence $H$-equivariant, $f$ maps $Y^H$ to the set $Z^H$ of $H$-fixed points in $Z$. We denote the restricted map by $f^H : \; Y^H \to Z^H$.
\begin{definition}\label{def:connected}
If $G$ is a group, let $m$ denote an integer-valued function $H \mapsto m(H)$ defined on all subgroups $H$ of $G$. Then a $G$-equivariant map $f: Y \to Z$ is called $m$-{\rm connected} if for any subgroup $H \subset G$ the restriction $f^H : \; Y^H \to Z^H$ is $m(H)$-connected.
\end{definition}
We are now in a position to write down the desired statement; for a reference, see \cite{Greenlees-May}.
\begin{theorem}[$G$-Whitehead Theorem]
If $X$ is a $G$-$CW$ complex and the base-point preserving and $G$-equivariant map $f : \; Y\to Z$ is $m$-connected, then the induced map
\begin{align*}
    f_* : \; [X,Y]_\ast^G & \to [X,Z]_\ast^G \,, \quad [g] \mapsto [f\circ g] ,
\end{align*}
is bijective if $dim(X^H) < m(H)$ for all subgroups $H$ of $G$. It is surjective if $dim(X^H) \le m(H)$ for all subgroups $H$ of $G$. \end{theorem}

\subsection{Reformulation by relative homotopy}

We return to our task of investigating the mapping $\beta_\ast^{\mathbb{Z}_2}$ of Proposition \ref{prop:5.1}. The link with the material above is made by the identifications $Y = C_s(n)$, $Z = \Omega_K C_{s+1}(2n)$, and $G = \mathbb{Z}_2\,$. To apply the $G$-Whitehead theorem, we need to look at our map $\beta : \; Y \to Z$ and determine how connected (in the sense of Def.\ \ref{def:connected}) are its restrictions $Y^H \to Z^H$ to the fixed-point sets of all subgroups $H \subset \mathbb{Z}_2\,$. There are only two subgroups to consider: $H = \{ e \}$ (trivial group), and $H = G = \mathbb{Z}_2\,$. In the former case, the required result has been laid down in Proposition \ref{prop:beta}. What remains to be dealt with is the latter case, namely $\beta^H : \; Y^H \to Z^H$ for $H = \mathbb{Z}_2\,$.

Thus our focus now shifts to the restricted map
\begin{equation*}
    \beta^{\mathbb{Z}_2} \equiv \beta^\prime : \; C_s(n)^{\mathbb{Z}_2} = R_s(n) \to \left( \Omega_K C_{s+1}(2n) \right)^{\mathbb{Z}_2} \,;
\end{equation*}
cf.\ Eq.\ (\ref{eq:def-betap}). Recall that $(\Omega_K C_{s+1}(2n) )^{\mathbb{Z}_2}$ stands for the space of $\mathbb{Z}_2$-equivariant paths joining $E_{+\mathrm{i}}(K)$ with $E_{-\mathrm{i}}(K)$ in $C_{s+1}(2n)$. By the $G$-Whitehead Theorem, we are led to ask whether the induced maps in homotopy,
\begin{equation}\label{eq:apply-GWH}
    \beta_\ast^\prime : \; \pi_d( R_s(n) ) \to
    \pi_d \big( (\Omega_K C_{s+1}(2n))^{\mathbb{Z}_2} \big) ,
\end{equation}
are isomorphisms. To answer this question, we need yet another concept: relative homotopy.

\begin{definition}
Let $\mathrm{D}^d$ be the $d$-dimensional disk with boundary $\partial \mathrm{D}^d = \mathrm{S}^{d-1} \subset \mathrm{D}^d$ and base point $x_\ast \in \mathrm{S}^{d-1}$. If $C$ is a topological space with subspace $R \subset C$ and base point $A_\ast \in R$, then $\pi_d ( C, R , A_\ast)$ is defined as the set of homotopy equivalence classes of continuous maps taking the triple $(\mathrm{D}^d , \mathrm{S}^{d-1} , x_\ast)$ into the triple $(C, R , A_\ast)$. For $d \geq 2$ one calls $\pi_d ( C, R , A_\ast)$ a relative homotopy group.
\end{definition}
\begin{remark}
The group structure for $d \geq 2$ is defined by concatenating maps as usual. $\pi_1 ( C, R , A_\ast)$ is just a set (not a group).
\end{remark}
\begin{lemma}\label{lem:rel-hom}
The target space in (\ref{eq:apply-GWH}) may be viewed as a relative homotopy group:
\begin{equation*}
    \pi_d \big( (\Omega_K C_{s+1}(2n))^{\mathbb{Z}_2} \big) \simeq
    \pi_{d+1} ( C_{s+1}(2n), R_{s+1}(2n) , A_\ast ) .
\end{equation*}
\end{lemma}
\begin{proof}
A homotopy class in $\pi_d \big( (\Omega_K C_{s+1}(2n))^{\mathbb{Z}_2} \big)$ is represented by a base-point preserving mapping $f : \; \mathrm{S}^d \to (\Omega_K C_{s+1}(2n))^{\mathbb{Z}_2}$ or, equivalently, a map $F : \; \tilde{S}(\mathrm{S}^d) \to C_{s+1}(2n)$ with the properties $F(x,t)^\perp = F(x,1-t)$ and $F(x_\ast,1/2) = A_\ast\,$, where $x \in \mathrm{S}^d$ and $0 \leq t \leq 1$ is a polar coordinate for the suspension $\tilde{S}(\mathrm{S}^d)$.

By the first property, such a map $F$ is already determined by its values on one of the two hemispheres of $\tilde{S}(\mathrm{S}^d) = \mathrm{S}^{d +1}$. Such a hemisphere is a disk $\mathrm{D}^{d+1}$ parameterized by $t$ for, say $0 \leq t \leq 1/2$, with boundary $\mathrm{S}^d$ at the equator $t = 1/2$. The values of $F$ at the equator are constrained by $F(x,1/2) = F(x,1/2)^\perp \in R_{s+1}(2n)$. Thus the restriction of $F$ to $0 \leq t \leq 1/2$ is a mapping that takes $\mathrm{D}^{d+1}$ to $C_{s+1}(2n)$, the boundary $\mathrm{S}^{d}$ to $R_{s+1}(2n)$, and the base point $x_\ast$ to $A_\ast\,$. It is clear that this correspondence is bijective. Indeed, from the restricted data for $0 \leq t \leq 1/2$ the full function $F$ is reconstructed by the relation $F(x,1-t) = F(x,t)^\perp$.

It is also clear that this bijection of maps descends to a bijection of homotopy classes.
\end{proof}

Now, using the identification offered by Lemma \ref{lem:rel-hom} we reformulate the maps of (\ref{eq:apply-GWH}) as
\begin{equation}\label{eq:rel-hom1}
    \beta_\ast^\prime : \; \pi_d(R_s(n) , A_\ast) \to
    \pi_{d+1} \big( C_{s+1}(2n), R_{s+1}(2n) , A_\ast ) .
\end{equation}
The $G$-Whitehead Theorem then prompts us to ask under which conditions these maps are isomorphisms. A partial answer is given in the next section.

\section{Bijection in homotopy for $s \in \{ 2\, , 6\} $}\label{sect:6}
\setcounter{equation}{0}

In this section we are going to show that for two symmetry classes, namely for $s = 2$ and $s = 6$, the issue in question can be settled rather directly. What distinguishes these two cases is the existence of a fiber bundle projection that allows us to reduce the task at hand to the standard scenario of real Bott periodicity. (The other cases, $s \notin \{ 2 \,, 6 \}$, will have to be handled by a less direct argument.)

To anticipate the strategy in somewhat more detail, the main idea is as follows. When $s = 2$ or $s = 6$, we are able to construct a fibration (actually, a fiber bundle)
\begin{equation}
    R_{s+1}(2n) \hookrightarrow C_{s+1}(2n)
    \stackrel{p}{\longrightarrow} \widetilde{R}_{s,1}(2n) ,
\end{equation}
for a certain base space $\widetilde{R}_{s,1}(2n) \simeq R_{s,1}(2n)$. The projection $p$ sends the base point $A_\ast \in R_{s+1}(2n)$ to the base point $E_{-\mathrm{i}}(K) \in \widetilde{R}_{s,1}(2n)$ and induces an isomorphism
\begin{equation}
    p_\ast : \; \pi_{d+1} (C_{s+1}(2n), R_{s+1}(2n) , A_\ast )
    \to \pi_{d+1} (\widetilde{R}_{s,1}(2n) , E_{-\mathrm{i}}(K) )
\end{equation}
by basic principles. This isomorphism $p_\ast$ composes with $\beta_\ast^\prime$ to give a map
\begin{equation}
    p_\ast \circ \beta_\ast^\prime : \; \pi_d (R_{s+1,1}(2n) , A_\ast) \to \pi_{d+1} (\widetilde{R}_{s,1}(2n) , E_{-\mathrm{i}}(K)) .
\end{equation}
On using a suitable form $R_{s+1,1}(2n) \simeq \widetilde{R}_s(n)$ and $\widetilde{R}_{s,1}(2n) \simeq \widetilde{R}_{s-1}(n)$ of $(1,1)$ periodicity, this turns into the isomorphism underlying real Bott periodicity,
\begin{equation}
    \pi_d( \widetilde{R}_s(n)) \to \pi_{d+1}( \widetilde{R}_{s-1}(n) ).
\end{equation}
Thus the desired statement will be reduced to a known result in topology.

Let us make the historical remark that, in order to discover the space $\widetilde{R}_{s, 1}(2n)$ which is central to our argument, it was necessary for us to abandon the usual (Majorana) convention of realizing the involution $\tau_\mathrm{car}$ by complex conjugation. In fact, we find it optimal to work with \emph{two} such involutions at once. In the next subsection, which is preparatory, we introduce the second involution, $\widetilde{\tau}_\mathrm{car}\,$.

\subsection{$\widetilde{\mathrm{CAR}}$ involution}

Recall that the CAR pairing of $\mathbb{C}^2 \otimes W$ is determined by a bracket $\{ \, , \, \}$ due to the canonical anti-commutation relations of fermionic Fock operators. Introducing the unitary operator $u_0 = (\mathrm{Id} - I K ) / \sqrt{2}$ we define a modified bracket by
\begin{equation}
    \widetilde{\{ w , w^\prime \}} = \{ u_0 w , u_0 w^\prime \} .
\end{equation}
By using the fact that $I$ is real and $K$ imaginary, which is to say that $I$ preserves the bracket $\{ \, , \, \}$ while $K$ reverses its sign, one computes
\begin{equation*}
    \widetilde{\{ w , w^\prime \}} = {\textstyle{\frac{1}{2}}} \{ w - I K w, w^\prime - I K w^\prime \}
    = \{ -I K w , w^\prime \} = \{ (I K)^{-1} w , w^\prime \}.
\end{equation*}
Thus if $A^\perp$ is the annihilator space of $A$, i.e., if $w \in A^\perp$ annihilates all $w^\prime \in A$ with respect to $\{ \, , \, \}$, then so does $IK w$ with respect to $\widetilde{\{ \, , \, \}}$. Hence, by adopting the modified CAR bracket $\widetilde{ \{ \, , \, \} }$ we get a modified CAR involution $\widetilde{\tau}_{s+1} : \; C_{s+1}(2n) \to C_{s+1}(2n)$,
\begin{equation}\label{eq:def-taut}
    \widetilde{\tau}_{s+1} (A) = I K \tau_{s+1}(A) = I K A^\perp .
\end{equation}
The replacement of $\tau_{s+1}$ by $\widetilde{\tau}_{s+1}$ also changes the CAR involution on the operators $I, K$:
\begin{align}
    \widetilde{\tau}_\mathrm{car} (K) &= IK \tau_\mathrm{car}(K) (IK)^{-1} = I (-K) I^{-1} = + K , \\ \widetilde{\tau}_\mathrm{car} (I) &= IK \tau_\mathrm{car}(I) (IK)^{-1} = K I K^{-1} = - I .
\end{align}
Thus the roles of $I$ and $K$ get exchanged: while $K$ was imaginary with respect to $\tau_\mathrm{car}$ it is real with respect to $\widetilde{ \tau}_\mathrm{car}\,$, and vice versa for $I$. The remaining generators $J_l = \tau_\mathrm{car}(J_l) = \widetilde{\tau}_\mathrm{car} (J_l)$ for $l = 1, \ldots, s$ are real with respect to both structures, $\mathrm{CAR}$ and $\widetilde{ \mathrm{CAR}}$.

Guided by the above, we employ $\widetilde{\tau}_{s+1}$ to define a space $\widetilde{R}_{s,1}(2n)$ by
\begin{equation}\label{eq:def-tildeR}
    \widetilde{R}_{s,1}(2n) = \{ A \in C_s(2n) \mid I A^\mathrm{c} = A = \widetilde{\tau}_{s+1}(A) \} .
\end{equation}
This is to be compared with $R_{s,1}(2n) = \{ A \in C_s(2n) \mid K A^\mathrm{c} = A = \tau_{s+1}(A) \}$. Note that $R_{s,1}(2n)$ is mapped to $\widetilde{R}_{s,1}(2n)$ by the transformation $A \mapsto u_0 \,A$. Thus $\widetilde{R}_{s,1} (2n) \simeq R_{s,1}(2n)$.

\subsection{Connection with real Bott periodicity}

{}From Eq.\ (\ref{eq:Z2-Bottmap}) we recall the definition of the mapping $\beta$ behind our diagonal map:
\begin{equation*}
    \beta_t(A) = \mathrm{e}^{(t\,\pi/2)K J(A)} \cdot E_{+\mathrm{i}}(K) \,.
\end{equation*}
While this is a curve in $C_{s+1}(2n)$ when $A \in C_{s+2}(2n)$ is in general position (and our true goal is to characterize the mapping $\beta^\prime$ to $\mathbb{Z}_2$-equivariant curves; see (\ref{eq:def-betap})), we now observe that $t \mapsto \beta_t(A)$ for $A = A^\perp \in R_{s+1,1}(2n)$ has the following alternative interpretation.
\begin{lemma}
For $A \in R_{s+1,1}(2n)$ the curve $t \mapsto \beta_t(A)$ is a curve in $\widetilde{R}_{s,1}(2n)$.
\end{lemma}
\begin{proof}
By inspecting the definitions (\ref{eq:def-tildeR}) and (\ref{eq:def-Rs+1,1}) one sees that
\begin{equation}\label{eq:intersect}
     R_{s+1,1}(2n) = \widetilde{R}_{s,1}(2n) \cap R_{s+1}(2n) .
\end{equation}
Indeed, the two spaces on the right-hand side have the same pseudo-symmetries including $I A = A^\mathrm{c}$, but the points of the second space are fixed with respect to $\tau_{s+1}$ while the first space is the fixed-point set of $\widetilde{\tau}_{s+1}$. In view of Eq.\ (\ref{eq:def-taut}) this implies that $A \in \widetilde{R}_{s,1}(2n) \cap R_{s+1}(2n)$ is invariant under multiplication by $IK$. Since $I$ is a pseudo-symmetry, it follows that so is $K$, i.e., $K A = A^\mathrm{c}$. Therefore the intersection on the right-hand side of Eq.\ (\ref{eq:intersect}) does give the space on the left-hand side.

Owing to (\ref{eq:intersect}) all points $A$ of $R_{s+1,1}(2n)$ lie in both $R_{s+1}(2n)$ and $\widetilde{R}_{s,1}(2n)$. Also, the product $K J(A)$ commutes with all generators $I, J_1, \ldots, J_s\,$. It follows that the one-parameter group of unitary operators $\mathrm{e}^{(t\,\pi/2) K J(A)}$ preserves the pseudo-symmetry relations of $\widetilde{R}_{ s,1}(2n)$. Moreover, $\mathrm{e}^{(t\,\pi/2) K J(A)}$ is real with respect to the $\widetilde{\rm CAR}$ structure since $\widetilde{\tau}_\mathrm{car} (K) = + K$ and
\begin{equation*}
    \widetilde{\tau}_\mathrm{car} (J(A)) =
    J(\widetilde{\tau}_{s+1}(A)) = J(A) .
\end{equation*}
Hence $\beta_t(A) \in \widetilde{R}_{s,1}(2n)$ as claimed.
\end{proof}
\begin{remark}
In particular, $\beta_0(A) = E_{+\mathrm{i}}(K)$ and $\beta_1(A) = E_{-\mathrm{i}}(K)$ are points of $\widetilde{R}_{s,1}(2n)$.
\end{remark}

To prepare the next statement, note that the continuous map \begin{equation*}
    \widetilde{\beta}: \; R_s(n) \simeq R_{s+1,1}(2n) \to \Omega_K \widetilde{R}_{s,1}(2n) ,
\end{equation*}
which assigns to $A \in R_{s+1,1}(2n)$ the geodesic $[0,1] \ni t \mapsto \beta_t(A)$ joining $E_{+\mathrm{i}}(K)$ with its antipode $E_{-\mathrm{i}}(K)$, induces a mapping $\widetilde{\beta}_\ast$ in homotopy; more precisely, by concatenating $f : \; \mathrm{S}^d \to R_s(n)$ with $\widetilde{\beta}$ we get $F = \widetilde{\beta} \circ f : \; S(\mathrm{S}^d) = \mathrm{S}^{d+1} \to \widetilde{R}_{s,1}(2n)$, and this construction induces a mapping between homotopy groups as usual; cf.\ the beginning of Section \ref{sect:5}. Let it be stressed that $\mathrm{S}^d$ and $S(\mathrm{S}^d) = \mathrm{S}^{d+1}$ are plain spheres here, with no $\mathbb{Z}_2$-group acting.
\begin{proposition}\label{prop:real-Bott}
If $R_{s+1,1}(2n)_0$ denotes the connected component of $R_{s+1,1}(2n)$ containing the base point $A_\ast$, the induced map
\begin{equation*}
    \widetilde{\beta}_\ast :\; \pi_d(R_{s+1,1}(2n), A_\ast) \to \pi_{d+1} (\widetilde{R}_{s,1} (2n), A_\ast)
\end{equation*}
is an isomorphism for $1 \leq d \ll \dim R_{s+1,1}(2n)_0$ and $s \geq 1$.
\end{proposition}
\begin{proof}
Fundamental to the celebrated result of real Bott periodicity, an isomorphism $\pi_d( R_s(n)) \to \pi_{d+1}( R_{s-1}(n))$ exists \cite{Bott1959} for $d \geq 1$ and $d$ in the stable range. Our map $\widetilde{\beta}_\ast$ turns into this isomorphism on making some natural identifications via the $(1,1)$ periodicity theorem of Section \ref{sect:1,1-red}, Proposition \ref{prop:2.1X}. Indeed, by an elaboration of the reasoning in the proof of Proposition \ref{prop:beta} we may argue as follows.

Adopting the $\widetilde{\rm CAR}$ structure, we have an identification $R_{s+1,1}(2n) \simeq \widetilde{R}_s(n)$ where $\widetilde{R}_s(n) \subset \mathrm{Fix}(\widetilde{\tau}_\mathrm{car})$ is determined by $s$ pseudo-symmetries that descend from $J_1, \ldots, J_{s-1}, K$ via the $(1,1)$ periodicity theorem with $\widetilde{\rm CAR}$-real generator $J_s$ and $\widetilde{\rm CAR}$-imaginary generator $I$. In the same way, $\widetilde{R}_{s,1}(2n) \simeq \widetilde{R}_{s-1}(n)$ where $\widetilde{R}_{s-1}(n)$ is determined by the $s-1$ pseudo-symmetries that descend from $J_1, \ldots, J_{s-1}$. Our diagonal map $A \mapsto \{ t \mapsto \beta_t(A)\}$ then descends by $(1,1)$ periodicity to the standard Bott map (underlying real Bott periodicity), which assigns to any point of $\widetilde{R}_s(n)$ a minimal geodesic joining pode and antipode in $\widetilde{R}_{s-1}(n)$. The statement of the Proposition thus reduces to a standard result (see the book by Milnor \cite{milnor} for a pedagogical exposition).
\end{proof}

\subsection{Squaring by the CAR involution}

In the sequel we assume that the dimension of $M$ is at least one (thereby excluding the example $M = \mathrm{S}^0$ of Remark \ref{rem:5.1}), and we take $M$ to be path-connected. Then, while the target space $C_s(n)$ may consist of many connected components (see Table \ref{table:CsRs} of Section \ref{sect:ClassMaps}), the set $[M,C_s(n)]_\ast^{\mathbb{Z}_2}$ of homotopy classes probes only the connected component containing the base point $A_\ast$, because our maps $\phi : \; M \to C_s(n)$ are continuous and base-point preserving. The same goes for the set of homotopy classes $[\tilde{S} M,C_{s+1}(2n)]_\ast^{\mathbb{Z}_2}$. Therefore, we now restrict all relevant target spaces $Y$ to their connected components $Y_0$ containing $A_\ast$. We also use this opportunity to simplify our notation, by setting
\begin{equation}
    C \equiv C_{s+1}(2n)_0 \,, \quad R_{s+1} \equiv R_{s+1}(2n) \cap C , \quad \widetilde{R}_{s,1} \equiv \widetilde{R}_{s,1}(2n)_0 \,.
\end{equation}
The equality (\ref{eq:intersect}) changes to
\begin{equation}
    R_{s+1,1} \equiv R_{s+1,1}(2n)_0 \subset R_{s+1} \cap \widetilde{R}_{s,1} \,,
\end{equation}
since $R_{s+1}$ and $\widetilde{R}_{s,1}$ may intersect in more than one component. We recall that $C$, $R_{s+1}$, $\widetilde{R}_{s,1}$ share the same pseudo-symmetries $\{ I , J_1, \ldots, J_s \}$, and that our setting is equipped with two involutions $\tau_{s+1} , \widetilde{\tau}_{s+1} :\; C \to C$, fixing the points of $R_{s+1} \subset C$ and $\widetilde{R}_{s, 1} \subset C$, respectively.

We then recall that the eigenspaces $E_{\pm\mathrm{i}}(K)$ of the generator $K$ are exchanged by each of the linear operators $I, J_1, \ldots, J_s\,$. Thus $I, J_1, \ldots, J_s$ are pseudo-symmetries for $E_{\pm\mathrm{i}} (K)$ and we have $E_{\pm\mathrm{i}}(K) \in C$. This allows us to regard the connected space $C$ as the orbit of, say $E_{+\mathrm{i}}(K)$, under the action of its connected symmetry group, $U$:
\begin{equation}
    C = U \cdot E_{+\mathrm{i}}(K) , \quad U = \{ u \in \mathrm{U} (\mathbb{C}^2 \otimes W) \mid u = I u I^{-1} = J_1 u J_1^{-1} = \ldots = J_s u J_s^{-1} \} .
\end{equation}
From this perspective, we may also think of $C$ as a coset space $U / U_K$ where $U_K$ is the isotropy group of $E_{+\mathrm{i}}(K)$:
\begin{equation}
    U_K = \{u \in U\mid u\cdot E_{+\mathrm{i}}(K) = E_{+\mathrm{i}}(K)\}.
\end{equation}
This subgroup $U_K$ can be viewed as the group of fixed points of a Cartan involution $\theta$:
\begin{equation}
    U_K = \mathrm{Fix}_U(\theta) \equiv \{ u \in U \mid \theta(u) = u \} , \quad \theta(u) = IK u (IK)^{-1} .
\end{equation}
(One may compute $\theta$ more simply by $\theta(u) = K u K^{-1}$ as $I$ commutes with all $u \in U$.) On basic grounds, the fact that the elements of $U_K$ are fixed by a Cartan involution implies that $U_K$ is a symmetric subgroup and $C \simeq U / U_K$ is a symmetric space. (In fact, $C$ in all cases is either a unitary group or a complex Grassmannian; see Table \ref{table:CsRs} in Section \ref{sect:ClassMaps}.) Note the relation
\begin{equation}
    \widetilde{\tau}_\mathrm{car} = \theta \circ \tau_\mathrm{car} \,.
\end{equation}

Beyond $U_K \subset U$, two more groups of relevance for the following discussion are the subgroups $G$ and $L$ of elements fixed by the $\mathrm{CAR}$ and $\widetilde{ \mathrm{CAR}}$ involutions respectively:
\begin{equation}
    G = \{ g \in U \mid \tau_\mathrm{car}(g) = g \}, \quad
    L = \{ l \in U \mid \widetilde{\tau}_\mathrm{car}(l) = l \} .
\end{equation}
(If $L$ is not connected, we replace it by its connected component containing the identity.) As subgroups of $U$, both $G$ and $L$ act on $C \simeq U / U_K\,$. These Lie group actions of $G$ and $L$ have nice properties due to the fact that both involutions, $\tau_\mathrm{car}$ and $\widetilde{\tau}_\mathrm{car} \,$, commute with $\theta$, as is immediate from $\tau_\mathrm{car}(IK) = - IK = \widetilde{\tau}_\mathrm{car}(IK)$.

To get ready for Lemma \ref{lem:p-props} below, we need to accumulate a few more facts. First of all, the space $\widetilde{R}_{s,1}$ can be seen as the $L$-orbit in $C$ through $E_{+\mathrm{i}}(K)$. Alternatively, we may think of $\widetilde{R}_{s,1} = L \cdot E_{+\mathrm{i}}(K)$ as $\widetilde{R}_{s,1} \simeq L / H$ for $H = L \cap U_K\,$. Here we note that $\theta : \; U \to U$ restricts to an involution $\theta : \; L \to L$ and that $H = \mathrm{Fix}_L(\theta)$ is a symmetric subgroup. It is sometimes useful to identify the symmetric space $L / H$ with its Cartan embedding into $L \subset U$. This is defined to be the space
\begin{equation}
    \mathrm{U}(L/H) = \{ l \in L \mid \theta(l) = l^{-1} \} ,
\end{equation}
and the embedding goes by
\begin{equation}
    L/H \stackrel{1:1}{\longrightarrow} \mathrm{U}(L/H) , \quad l \mapsto l \theta(l^{-1}) .
\end{equation}
A similar discussion can be given for the $G$-orbit in $C$ through $A_\ast\,$, but the only fact we need in this case is the identification $G \cdot A_\ast = R_{s+1}\,$.
\begin{lemma}\label{lem:p-props}
Suppose that the principal bundle $U \to U/U_K = C$ admits a global section, i.e.\ a map $\sigma : \; C \to U$ with $\sigma(A) \cdot E_{+\mathrm{i}}(K) = A$ for all $A \in C$. Suppose further that
\begin{enumerate}
\item[(i)] for all $(A,A^\prime) \in C \times C$, the group elements $\sigma(A)$ and $\tau_\mathrm{car}(\sigma(A))$ commute with the group elements $\theta(\sigma(A^\prime))$ and $\widetilde{\tau}_\mathrm{car} (\sigma(A^\prime))$, and
\item[(ii)] for all $A \in \widetilde{R}_{s,1}$ the relation $\tau_\mathrm{car} (\sigma(A)) = \sigma(A)^{-1}$ holds.
\end{enumerate}
Then the following statements about the mapping $p : \; C \to C$ defined by
\begin{equation}
    p(A) = \tau_\mathrm{car}(\sigma(A))^{-1} \cdot A
\end{equation}
hold true:
\begin{enumerate}
\item $p$ is onto $\widetilde{R}_{s,1}\,$.
\item $p(\beta_t(A)) = \beta_{2t}(A)$ for all $A \in R_{s+1,1}$.
\item $p(R_{s+1}) = E_{-\mathrm{i}}(K)$.
\item $p$ has the homotopy lifting property: given any homotopy $f : \; M \times [0,1] \to \widetilde{R}_{s,1}$ and $F_0 :\; M \to C$ with $p \circ F_0 = f(\cdot,0)$, there exists a homotopy $F :\; M \times [0,1] \to C$ lifting $f$, i.e.\ $p \circ F = f$ and $F(\cdot,0) = F_0\,$.
\end{enumerate}
\end{lemma}
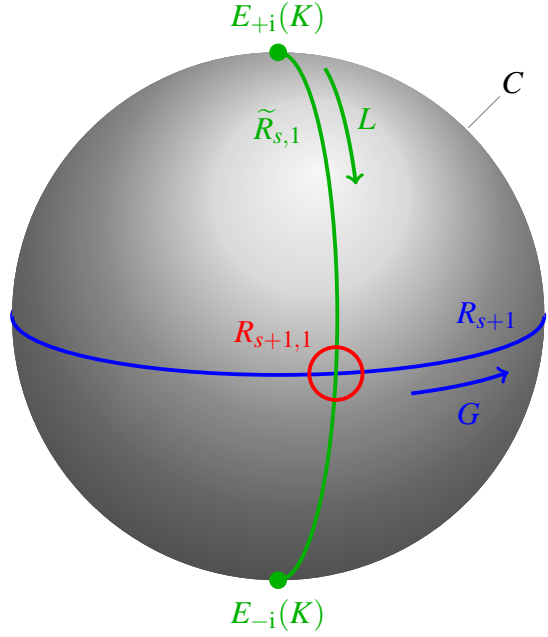
\begin{figure}
\begin{center}
\begin{tikzpicture}
\def\rad{3.5cm }
\def\radsmall{\rad/10}
\def\lw{0.05cm}
\def\arrowwidth{0.05cm}
\def\slope{\rad/4.5 }
\def\green{black!30!green}

\path [name path=bh] (-\rad,0) arc (180:0:\rad and \slope);
\path [name path=bv] (0,\rad) arc (90:270:\slope and \rad);
\path [name path=fh] (-\rad,0) arc (180:360:\rad and \slope);
\path [name path=fv] (0,\rad) arc (90:-90:\slope and \rad);
\path [name intersections={of=bh and bv, by=x1}];
\path [name intersections={of=fh and fv, by=x2}];

\draw[dashed,red,line width=\lw] (x1) circle (\rad/10);
\draw[dashed,blue,line width=\lw] (-\rad,0) arc (180:0:\rad and \slope);
\draw[dashed,\green,line width=\lw] (0,\rad) arc (90:270:\slope and \rad); 

\shade[transform canvas={rotate=-60},ball color=black!20!white,opacity=0.90] (0,0) circle (\rad);

\draw[blue,line width=\lw] (-\rad,0) arc (180:360:\rad and \slope);
\draw[\green,line width=\lw] (0,\rad) arc (90:-90:\slope and \rad);
\draw[label=test,red,line width=\lw] (x2) circle (\rad/10);

\node[anchor=south east] at ($(x2)+(-\radsmall/2.5,\radsmall/2.5)$) {${\color{red} R_{s+1,1}}$};
\node[anchor=east] at (\rad-0.2cm,0) {${\color{blue} R_{s+1}}$};
\node[anchor=north] at (0,\rad-0.6cm) {${\color{\green} \widetilde{R}_{s,1}}$};

\fill[\green] (0,\rad) circle (\rad/30);
\fill[\green] (0,-\rad) circle (\rad/30);
\node[anchor=south] at (0,\rad+0.1cm) {${\color{\green} E_{+\mathrm{i}}(K)}$};
\node[anchor=north] at (0,-\rad-0.1cm) {${\color{\green} E_{-\mathrm{i}}(K)}$};

\draw [->,\green,line width=\arrowwidth,domain=70:30] plot ({\slope*cos(\x)+\radsmall}, {\rad*sin(\x)});
\draw [->,blue,line width=\arrowwidth,domain=300:330] plot ({\rad*cos(\x)}, {\slope*sin(\x)-\radsmall});

\node[anchor=south] at (\rad/3,\rad/1.5) {${\color{\green} L}$};
\node[anchor=north] at (\rad/1.4,-\rad/3.5) {${\color{blue} G}$};

\draw[gray] (\rad/1.4,\rad/1.4) -- (\rad/1.2,\rad/1.2) node[anchor=south west,inner sep=1]{${\color{black} C}$};
\end{tikzpicture}
\end{center}
\caption{Schematic visualization of the setting in Lemma \ref{lem:p-props}: The orbits of $E_{+\mathrm{i}}(K)$ and $A_\ast$ under the groups $L$ and $G$ are the spaces $\widetilde{R}_{s,1}$ and $R_{s+1}$ (green and blue) respectively. Their intersection (red) contains the space $R_{s+1,1}$. The projection $p :\; C \to \widetilde{R}_{s,1}$ squares (``doubles'') in the green direction (property 1 and 2) and thereby sends the blue part to the south pole $E_{-\mathrm{i}}(K)$ (property 3).}
\label{fig:rgb}
\end{figure}

\begin{proof}
We first show that $p$ is into $\widetilde{R}_{s,1}\,$. For this we write $p(A) = \Sigma(A) \cdot E_{+\mathrm{i}}(K)$ with $\Sigma(A) = \tau_\mathrm{car}(\sigma(A))^{-1} \sigma(A)$ and send $p(A)$ to its image under the Cartan embedding:
\begin{equation*}
    p(A) \mapsto \Sigma(A) \theta(\Sigma(A))^{-1} \equiv \ell \,.
\end{equation*}
Let $\Sigma(A) \equiv \Sigma$ for short, and notice that $\tau_\mathrm{car} (\Sigma) = \Sigma^{-1}$. Applying $\widetilde{\tau}_\mathrm{car} = \tau_\mathrm{car} \circ \theta$ to $\ell$ one gets
\begin{align*}
    \widetilde{\tau}_\mathrm{car} (\ell) = \widetilde{\tau}_\mathrm{car} \left(\Sigma \theta(\Sigma)^{-1} \right) = (\theta \circ {\tau}_\mathrm{car})(\Sigma) \tau_\mathrm{car} (\Sigma)^{-1} = \theta(\Sigma)^{-1} \Sigma \,.
\end{align*}
Now a short calculation using the assumption (i) shows that $\Sigma$ commutes with $\theta(\Sigma)^{-1}$. We therefore have $\widetilde{\tau}_\mathrm{car} (\ell) = \ell \in L\,$. This means that $\ell = \tau_\mathrm{car} (\ell)^{-1} = \theta(\ell)^{-1}$ lies in the Cartan embedding $\mathrm{U} (L/H)$, which in turn implies that $p(A) \in L \cdot E_{+\mathrm{i}}(K)$. Thus $p$ is into $\widetilde{R}_{s,1}\,$.

To see that $p : \; C \to \widetilde{R}_{s,1}$ is surjective, let $A =
\sigma(A) \cdot E_{+\mathrm{i}}(K) \in \widetilde{R}_{s,1}$. By our second assumption (ii), the expression for $p(A)$ in this case takes the form
\begin{equation*}
    p(A) = \tau_\mathrm{car}(\sigma(A))^{-1} \cdot A
    = \sigma(A)^2 \cdot E_{+\mathrm{i}}(K) .
\end{equation*}
Thus $p : \; \widetilde{R}_{s,1} \to \widetilde{R}_{s,1}$ is the operation of squaring (or doubling the geodesic distance) from the point $E_{+\mathrm{i}}(K)$: in normal coordinates by the exponential mapping w.r.t.\ $E_{+\mathrm{i}}(K)$ it is the map
\begin{equation}\label{eq:6.37}
    p(A) = p (\exp(X) \cdot E_{+\mathrm{i}}(K))
    = \exp(2X) \cdot E_{+\mathrm{i}}(K) .
\end{equation}
Since the squaring map is surjective, it follows that $p : \; C \to \widetilde{R}_{s,1}$ is onto.

Now recall $R_{s+1,1} \subset \widetilde{R}_{s,1}$ and $\beta_t(A) = \mathrm{e}^{(t\,\pi/2) K J(A)} \cdot E_{+\mathrm{i}} (K)$. The second stated property is then an immediate consequence of the relation (\ref{eq:6.37}):
\begin{equation*}
    p(\beta_t(A)) = \big( \mathrm{e}^{(t\,\pi/2) K J(A)} \big)^2 \cdot E_{+\mathrm{i}}(K) = \beta_{2t}(A) .
\end{equation*}

Turning to the third property, we observe that $\sigma$ as a section of $U \to U/U_K$ satisfies
\begin{equation*}
    \sigma(u \cdot A) = u\, \sigma(A) \, h(u,A)  \qquad (u \in U)
\end{equation*}
for some $h(u,A)$ taking values in the isotropy group $U_K$ of $E_{+ \mathrm{i}}(K)$. By specializing this to $A = g \cdot A_\ast \in R_{s+1}$ for $u = g \in G$ and using $g = \tau_\mathrm{car}(g)$ we obtain
\begin{align*}
    p(A)&= \tau_\mathrm{car}(\sigma(A))^{-1}\sigma(A)\cdot E_{+\mathrm{i}} (K)\cr &= \tau_\mathrm{car}(h)^{-1} \tau_\mathrm{car}(\sigma (A_\ast) )^{-1} \sigma(A_\ast) \cdot E_{+\mathrm{i}}(K) = \tau_\mathrm{car} (h)^{-1} p(A_\ast) \,.
\end{align*}
{}From the second property of $p$ we know that $p(A_\ast) = p(\beta_{1/2} (A_\ast)) = \beta_1(A_\ast) = E_{-\mathrm{i}}(K)$. Now the subgroup $U_K$ of $\theta$-fixed points is stable under $\tau_\mathrm{car}\,$, since $\theta$ and $\tau_\mathrm{car}$ commute. Therefore, $\tau_\mathrm{car} (h)^{-1} \in U_K$ and we conclude that $p(A) = \tau_\mathrm{car}(h)^{-1} E_{-\mathrm{i}}(K) = E_{-\mathrm{i}}(K)$.

To establish the homotopy lifting property, it suffices to show that inside the tangent bundle $T C$ there exists a smooth distribution of ``horizontal'' spaces (in short: a horizontal sub-bundle of $T C$) assigning to every pair of a point $A \in C$ and a tangent vector $v \in T_{p(A)} \widetilde{R}_{s,1}$ a lifted tangent vector, i.e.\ a vector $\tilde{v} \in T_A C$ such that $p_\ast(\tilde{v}) = v$. Indeed, such a distribution associates with every curve $\gamma_k : \; [0,1] \to \widetilde{R}_{s,1}$, $t \mapsto f(k,t)$, a dynamical system in the form of a first-order ordinary differential equation that may be solved (using the initial condition given by $F_0(k) \in C$) to obtain a lifted curve $\tilde\gamma_k$ in $C$. The collection $\{ \tilde\gamma_k \}_{k \in M}$ of these curves yields the lifted map $F : \; M \times [0,1] \to C$.

Thus, in view of the stated goal of lifting tangent vectors, let
\begin{equation*}
    v_0 = \frac{d}{dt}\bigg\vert_{t = 0} \gamma_t \,, \quad \gamma_t = \mathrm{e}^{t Y} \cdot E_{+\mathrm{i}}(K) , \quad Y = - \theta_\ast(Y) = - (\tau_\mathrm{car})_\ast (Y) \in \mathrm{Lie}(L) ,
\end{equation*}
be a vector tangent to $\widetilde{R}_{s,1}$ at $E_{+\mathrm{i}}(K)$. This translates to a vector $v$ tangent to $\widetilde{R}_{s,1}$ at any point $p(A) = \tau_\mathrm{car}( \sigma(A) )^{-1} \cdot A$ by
\begin{equation*}
    v = \frac{d}{dt}\bigg\vert_{t = 0} \ell \cdot \gamma_t \,, \quad
    \ell = \tau_\mathrm{car}( \sigma(A) )^{-1} \theta( \sigma(A) )^{-1} .
\end{equation*}
Indeed, by our assumption (i) the product $\ell = \widetilde{ \tau}_\mathrm{car}(\ell)$ is an element of the subgroup $L$, and we have $\theta( \sigma(A) )^{-1} \cdot E_{+\mathrm{i}}(K) = A$ since $\sigma(A)^{-1} \theta( \sigma(A) )^{-1} \in U_K$ by the same assumption.

Consider now the vector
\begin{equation*}
    \tilde{v} = \frac{d}{dt}\bigg\vert_{t=0} \sigma(\gamma_{t/2})\cdot A
\end{equation*}
tangent to $C$ at $A$. We claim that the differential $p_\ast : \; T_A C \to T_{p(A)} \widetilde{R}_{s,1}$ takes $\tilde{v}$ to $v$. To verify this claim, we may assume that $\sigma(\sigma(\gamma_{t/2}) \cdot A) = \sigma( \gamma_{ t/2}) \sigma(A)$, since the bundle $U \to C$ is trivial. We then do the following computation:
\begin{align*}
    p \big( \sigma(\gamma_{t/2}) \cdot A \big) &= \tau_\mathrm{car}( \sigma(\gamma_{t/2}) \sigma(A) )^{-1} \sigma( \gamma_{t/2}) \sigma(A) \cdot E_{+\mathrm{i}}(K) \cr
    &= \tau_\mathrm{car}(\sigma(A))^{-1} \tau_\mathrm{car}( \sigma(\gamma_{t/2}) )^{-1} \sigma( \gamma_{t/2}) \theta(\sigma(A))^{-1} \cdot E_{+\mathrm{i}}(K) \cr
    &= \tau_\mathrm{car}(\sigma(A))^{-1} \theta(\sigma(A))^{-1} \tau_\mathrm{car}( \sigma(\gamma_{t/2}) )^{-1} \cdot \gamma_{t/2} = \ell \cdot \gamma_t \,.
\end{align*}
In the last step we used that $\tau_\mathrm{car}( \sigma(\gamma_{t/2}) )^{-1} \cdot \gamma_{t/2} = p(\gamma_{t/2}) = \gamma_t\,$. By taking the derivative at $t = 0$ we obtain the desired result $p_\ast(\tilde{v}) = v$.

In summary, we have a right inverse for $p_\ast :\; T_A C \to T_{p(A)} \widetilde{R}_{s,1}$ by the assignment
\begin{equation*}
    \frac{d}{dt}\bigg\vert_{t = 0} \tau_\mathrm{car}( \sigma(A) )^{-1} \theta( \sigma(A) )^{-1} \mathrm{e}^{t Y} \cdot E_{+\mathrm{i}}(K) \mapsto \frac{d}{dt}\bigg\vert_{t = 0} \sigma\big( \mathrm{e}^{tY / 2} \cdot E_{+\mathrm{i}}(K) \big) \cdot A
\end{equation*}
with $Y$ in the $\theta$-odd part of $\mathrm{Lie}(L)$. By the general reasoning sketched above, it follows that the mapping $p : \; C \to \widetilde{R}_{s,1}$ has the homotopy lifting property.
\end{proof}
\begin{remark}
The section $\sigma$, whose existence is a necessary condition for the statement of Lemma \ref{lem:p-props} to hold, exists if and only if $s \in 4 \mathbb{N} - 2$.
\end{remark}
\begin{proposition}\label{prop:6.2mz}
The map $\beta_\ast^\prime$ of (\ref{eq:rel-hom1}) is bijective for $s \in \{ 2 \,, 6 \}$ and $1 \leq d \ll \dim R_s(n)_0\,$.
\end{proposition}
\begin{proof}
For definiteness, let $s = 2$. Then $U = \mathrm{U}_n \times \mathrm{U}_n \,$. The Cartan involution $\theta$ has the effect of exchanging the two factors of $U = \mathrm{U}_n \times \mathrm{U}_n\,$, so the subgroup $\mathrm{Fix}(\theta) = U_K$ is the diagonal subgroup $\mathrm{U}_n\,$. The involution $\tau_\mathrm{car}$ acts by $\tau_\mathrm{Sp}$ in each factor, with $\tau_\mathrm{Sp} : \; \mathrm{U}_n \to \mathrm{U}_n$ such that $\mathrm{Fix}(\tau_\mathrm{Sp}) = \mathrm{Sp}_n\,$. Hence $G = \mathrm{Fix} (\tau_\mathrm{car}) = \mathrm{Sp}_n \times \mathrm{Sp}_n$ and $L = \mathrm{Fix} (\widetilde{\tau}_\mathrm{car}) = \mathrm{U}_n\,$, with intersection $H = G \cap L = \mathrm{Sp}_n\,$. The orbit of $L$ on $E_{+\mathrm{i}}(K)$ is $\widetilde{R}_{2,1} = L / H = \mathrm{U}_n / \mathrm{Sp}_n\,$.

The principal bundle $U \to U/U_K = C$ is trivial, and we may take $\sigma$ to be of the form $\sigma(A) = (u,\mathrm{Id})$, with the second factor being the neutral element. The involution $\tau_\mathrm{car}$ does not mix the two factors; therefore, the second factor of $\tau_\mathrm{car} (\sigma(A))$ is still neutral. Because the Cartan involution $\theta$ swaps factors and thus moves the neutral element to the first factor, $\theta(\sigma(A))$ and $\widetilde{\tau}_\mathrm{car} (\sigma(A))$ commute with $\sigma(A^\prime)$ and $\tau_\mathrm{car} (\sigma(A^\prime))$ for any pair $(A,A^\prime) \in C \times C$, as is required in order for the first condition of Lemma \ref{lem:p-props} to be met. Moreover, an element $A \in \widetilde{R}_{2,1}$ lifts to $\sigma(A) = (u \tau_\mathrm{Sp} (u)^{-1}, \mathrm{Id})$ for some $u \in \mathrm{U}_n\,$. In this case one has $\tau_\mathrm{car}(\sigma(A)) = (\tau_\mathrm{Sp} (u)\, u^{-1}, \mathrm{Id}) = \sigma(A)^{-1}$, which means that also the second condition of Lemma \ref{lem:p-props} is satisfied. The case of $s = 6$ is the same but for the substitutions $n \to n/4$ and $\mathrm{Sp} \to \mathrm{O}$.

Thus Lemma \ref{lem:p-props} applies, and from the properties stated there it follows that for $s \in \{ 2 \,, 6 \}$ we have a short exact sequence
\begin{equation}\label{eq:fibration}
    R_{s+1} \hookrightarrow C \stackrel{p}{\longrightarrow} \widetilde{R}_{s,1} ,
\end{equation}
where the first map is simply the inclusion of $R_{s+1} = p^{-1} (E_{-\mathrm{i}}(K))$ into $C$ and the second map, $p$, has the homotopy lifting property. By the latter property, (\ref{eq:fibration}) is a fibration.

Now it is a standard result of homotopy theory (see Thm.\ 4.41 of \cite{Hatcher-AT}) that the mapping $p_\ast$ induced by the projection $p$ of a fibration -- in the concrete setting at hand, that's the map
\begin{equation*}
    p_\ast : \; \pi_{d+1}(C,R_{s+1},A_\ast) \to \pi_{d+1} (\widetilde{R}_{s,1}, E_{-\mathrm{i}}(K)) ,
\end{equation*}
is an isomorphism of homotopy groups for all $d \geq 1$. By composing $p_\ast$ with the mapping $\beta_\ast^\prime$ of Eq.\ (\ref{eq:rel-hom1}), we arrive at the map
\begin{equation*}
    p_\ast \circ \beta_\ast^\prime : \; \pi_d(R_s(n),A_\ast) \to \pi_{d+1}(\widetilde{R}_{s,1}(2n), E_{-\mathrm{i}}(K)) .
\end{equation*}
By the second property of $p$ stated in Lemma \ref{lem:p-props}, this composition is identical to the standard Bott map of Prop.\ \ref{prop:real-Bott}. Since the latter is an isomorphism for $1 \leq d \ll \dim R_s(n)_0$ and $p_\ast$ is an isomorphism for all $d \geq 1$, it follows that so is $\beta_\ast^\prime$ for $1 \leq d \ll \dim R_s(n)_0 \,$.
\end{proof}
\begin{remark}
To draw the same conclusion for all classes $s$, one would need eight fibrations of the following type:
\begin{align*}
    \mathrm{U}/\mathrm{Sp} \hookrightarrow
    &\; (\mathrm{U} \times \mathrm{U}) / \mathrm{U} \longrightarrow (\mathrm{O} \times \mathrm{O}) / \mathrm{O} \,, \cr
    \mathrm{Sp}/(\mathrm{Sp} \times \mathrm{Sp}) \hookrightarrow
    &\; \mathrm{U} / (\mathrm{U} \times \mathrm{U})
    \longrightarrow \mathrm{O} / \mathrm{U} \,,\cr
    (\mathrm{Sp} \times \mathrm{Sp})/\mathrm{Sp} \hookrightarrow
    &\; (\mathrm{U} \times \mathrm{U})/\mathrm{U} \longrightarrow
    \mathrm{U} / \mathrm{Sp} \,, \cr
    \mathrm{Sp} / \mathrm{U} \hookrightarrow
    &\; \mathrm{U} / (\mathrm{U} \times \mathrm{U})
    \longrightarrow \mathrm{Sp} / (\mathrm{Sp} \times \mathrm{Sp}) \,,\cr
    \mathrm{U} / \mathrm{O} \hookrightarrow
    &\; (\mathrm{U} \times \mathrm{U})/\mathrm{U} \longrightarrow
    (\mathrm{Sp} \times \mathrm{Sp}) / \mathrm{Sp} \,,\cr
    \mathrm{O}/(\mathrm{O} \times \mathrm{O}) \hookrightarrow
    &\; \mathrm{U} / (\mathrm{U} \times \mathrm{U})
    \longrightarrow \mathrm{Sp} / \mathrm{U} \,, \cr
    (\mathrm{O} \times \mathrm{O})/\mathrm{O} \hookrightarrow
    &\; (\mathrm{U} \times \mathrm{U})/\mathrm{U} \longrightarrow
    \mathrm{U} / \mathrm{O} \,, \cr
    \mathrm{O} / \mathrm{U} \hookrightarrow
    &\; \mathrm{U} / (\mathrm{U} \times \mathrm{U})
    \longrightarrow \mathrm{O} / (\mathrm{O} \times \mathrm{O}) \,.
\end{align*}
The third ($s = 2$) and seventh ($s = 6$) of these are the fibrations
discussed in the proof of Proposition \ref{prop:6.2mz}. While the others are available \cite{Giffen} in the $K$-theory limit of infinitely many bands ($n \to \infty$), they do not seem to exist at finite $n$.
\end{remark}

Anticipating the further developments of the next section, the fruit of all our labors in this paper will be Theorem \ref{thm:7.2}, which applies to all symmetry classes $s$. Here we state and prove that result in a preliminary version restricted to $s \in \{ 2, 6 \}$.
\begin{proposition}\label{thm:6.1}
Let $M$ be a path-connected $\mathbb{Z}_2$-CW complex, and let $s = 2 $ or $s = 6$. Then the mapping (\ref{eq:prop5.1}) between homotopy classes of $\mathbb{Z}_2$-equivariant maps,
\begin{equation*}
    \beta_\ast^{\mathbb{Z}_2} :\; [M , C_s (n)]_\ast^{\mathbb{Z}_2} \to [\tilde{S} M , C_{s+1}(2n)]_\ast^{\mathbb{Z}_2} \,,
\end{equation*}
which increases the symmetry index and the momentum-space dimension of a symmetry-protected topological phase by one, is bijective for $1 \le \dim M \ll \dim C_s(n)_0\,$.
\end{proposition}
\begin{proof}
After the identification $[\tilde{S} M, C_{s+1}(2n)]_\ast^{\mathbb{Z}_2} = [M, \Omega_K C_{s+1}(2n)]_\ast^{\mathbb{Z}_2}$, our statement is an immediate consequence of the $G$-Whitehead Theorem as explained in Section \ref{sect:5}. Recall that in order for that theorem to apply in the case of a $\mathbb{Z}_2$-equivariant mapping $\beta :\; Y \to Z$, one has to show that $\beta^H : \; Y^H \to Z^H$ is highly connected for all subgroups $H$ of $\mathbb{Z}_2\,$. We have done so (with the identifications $Y = C_s(n)$ and $Z = \Omega_K C_{s+1}(2n)$, and for $s \in \{ 2, 6 \}$) for $H = \{ e \}$ (by Proposition \ref{prop:beta}, which holds if $1 \leq d \ll \dim C_s(n)_0$) and for $H = \mathbb{Z}_2$ (by Proposition \ref{prop:6.2mz}, valid for $1 \leq d \ll R_s(n)_0$). By inspection of Table \ref{table:CsRs} one sees that the condition $d \ll R_s(n)_0$ is already implied by the condition $d \ll C_s(n)_0\,$.

In both cases, $H = \{ e \}$ and $H = \mathbb{Z}_2\,$, the fact that (for $s = 2, 6)$ there is no bijection between $\pi_0 (C_s(n))$ and $\pi_0 (\Omega_K C_{s+1}(2n))$ (resp.\ between $\pi_0 (R_s(n))$ and $\pi_0 \big( (\Omega_K C_{s+1}(2n))^{\mathbb{Z}_2} \big)$) is remedied by the assumption that $M$ is path-connected. Indeed, under that condition the image of the base-point preserving map $\beta$ (resp.\ $\beta^{\mathbb{Z}_2}$) lies entirely in the connected component of $\Omega_K C_{s+1}(2n)$ (resp.\ $(\Omega_K C_{s+1}(2n))^{\mathbb{Z}_2}$) containing the base point and we may simply restrict to that single connected component. With this detail in mind, the $G$-Whitehead Theorem indeed applies to give the stated result.
\end{proof}

\section{Bijection in homotopy for all $s$}\label{sect:7}
\setcounter{equation}{0}

In this section we extend the statement of Proposition \ref{thm:6.1} to all symmetry classes $s$. In order to do so, we prove an analog of Prop.\ \ref{thm:6.1} for the position-type suspension. This analog holds true for all $s$, and we use this fact, in conjunction with $S^r \tilde{S} M = \tilde{S} S^r M$, to extend Prop.\ \ref{thm:6.1} to all symmetry classes. For a corollary, we specialize to the case of $M = \mathrm{S}^{d_x , \,d_k}$ (a sphere with $d_x$ position-like and $d_k$ momentum-like coordinates) to reproduce the generalized Periodic Table for topological insulators and superconductors in the presence of a defect with codimension $d_x+1$ \cite{TeoKane}. Such a defect may be enclosed by a large sphere $\mathrm{S}^{d_x}$ and at every point of this sphere, the setting of \cite{kitaev} applies, giving a configuration space $M = \mathrm{S}^{d_x} \times \mathrm{T}^{d_k}$ if the system without defect has discrete translational symmetry, or $M = \mathrm{S}^{d_x} \times \mathrm{S}^{d_k}$ for continuous translation invariance. In \cite{KG14} it is proved that one may replace $\mathrm{S}^{d_x} \times \mathrm{T}^{d_k}$ (resp.\ $\mathrm{S}^{d_x} \times \mathrm{S}^{d_k}$) by $\mathrm{S}^{d_x ,\, d_k}$ at the expense of losing ``weak'' invariants. The resulting sets $[\mathrm{S}^{d_x,\, d_k}, C_s(n)]_\ast^{\mathbb{Z}_2}$ of equivariant homotopy classes are listed in Table \ref{table:periodictable}.

\begin{table}
\begin{center}
\begin{tabular}{|c|c|c|c|c|c|}
\hline
index 	& symmetry	& \multicolumn{4}{c|}{$d_k-d_x$}\\
\cline{3-6}
$s$	& label 	& $0$ & $1$& $2$ & $3$ \\
\hline
0 & $A$ 		&  ${\mathbb Z}$		& $0$			& ${\mathbb Z}$	& 0\\
1 & $\AIII$ 	& $0$				& ${\mathbb Z}$ 	& $0$			& ${\mathbb Z}$\\
\hline
0 & $D$ 		& ${\mathbb Z}_2$		& ${\mathbb Z}_2$	& ${\mathbb Z} $	 & $0$\\
1 & $\DIII$	&  $0$				& ${\mathbb Z}_2$ & ${\mathbb Z}_2$	& ${\mathbb Z}$ \\
2 & $\AII$ 	& ${\mathbb Z}$		& $0$			& ${\mathbb Z}_2$	& ${\mathbb Z}_2$\\
3 & $\CII$ 	& $0$				& ${\mathbb Z}$ 	& $0$ 			& ${\mathbb Z}_2$\\
4 & $C$ 		& $0$				& $0$ 			& ${\mathbb Z}$ 	& $0$\\
5 & $\CI$ 	& $0$				& $0$ 			& $0$			& ${\mathbb Z}$ \\
6 & $\AI$ 	& ${\mathbb Z}$		& $0$			& $0$ 			& $0$  \\
7 & $\BDI$ 	& ${\mathbb Z}_2$		&${\mathbb Z}$	& $0$			& $0$\\
\hline
\end{tabular}
\end{center}
\vspace{10pt}
\caption{The sets $[\mathrm{S}^{d_x,\,d_k},C_s(n)]_\ast^{\mathbb{Z}_2}$ for $1\le d_x + d_k\ll \dim C_s(n)_0\,$, also known as the Periodic Table of topological insulators and superconductors. The classes $A$ (Sect.\ \ref{sect:2.3.1}) and $\AIII$ (Sect.\ \ref{sect:AIII}) are included with trivial $\mathbb{Z}_2$-actions. The entries $0$, $\mathbb{Z}_2$ and $\mathbb{Z}$ mean sets with one, two and (countably) infinitely many elements, respectively. They are groups only when $d_x \geq 1$. For $d_x = d_k = 0$, the three entries of $\mathbb{Z}$ change to the sets $\mathbb{Z}_{2n+1}$ (class $A$), $\mathbb{Z}_{n/2+1}$ (class $\AII$) and $\mathbb{Z}_{n/4+1}$ (class $\AI$), which are in bijective correspondence with the connected components of $C_0(n)$ (class $A$), $R_2(n)$ (class $\AII$) and $R_6(n)$ (class $A$); see Table \ref{table:CsRs}.} \label{table:periodictable}
\end{table}

\subsection{Increasing the position-like dimension}

{}From Definition (\ref{eq:Z2-Bottmap}) recall the map $\beta$ given by
\begin{equation*}
    \beta_t(A) = \mathrm{e}^{(t\, \pi /2) K J(A)} \cdot A .
\end{equation*}
In the following we use the same definition, albeit with $A \in C_{s}(n)$ (rather than the previous $A \in C_{s+2}(2n)$) and with $\tau_\mathrm{car} (K) = K$ (rather than $\tau_\mathrm{car}(K) = -K$). The latter change, i.e., replacing the imaginary generator $K$ by a real one, has an important consequence: the second property listed in Lemma \ref{lem:4.1} changes from $\beta_t(A)^\perp = \beta_{1-t}(A^\perp)$ to $\beta_t (A)^\perp = \beta_{t}(A^\perp)$. Hence, the additional coordinate $t$ is now position-like rather than momentum-like. This means that the modified curve $t \mapsto \beta_t(A)$ agrees with the original Bott map \cite{milnor}: all $\mathbb{Z}_2$-fixed points $A \in R_s(n) \subset C_s(n)$ are now mapped to $\mathbb{Z}_2$-fixed points $\beta_t(A) \in R_{s-1}(n) \subset C_{s-1}(n)$ for all $t$.
\begin{theorem}\label{thm:7.1}
For a path-connected $\mathbb{Z}_2$-CW complex $M$ with $1\le\dim M \ll \dim C_s(n)_0\,$, the original Bott map $\beta$ induces a bijection
\begin{equation*}
    [M , C_s (n)]_\ast^{\mathbb{Z}_2} \stackrel{\sim}{\to} [S M , C_{s-1} (n)]_\ast^{\mathbb{Z}_2} \,,
\end{equation*}
where $SM$ is the position-type suspension (i.e., the $\mathbb{Z}_2$-action on the suspension coordinate is trivial).
\end{theorem}
\begin{proof}
After the identification  $[S M , C_{s-1}(n)]_\ast^{\mathbb{Z}_2} = [M , \Omega_K C_{s-1}(n)]_\ast^{\mathbb{Z}_2}$, we can apply the $\mathbb{Z}_2$-Whitehead Theorem \cite{Greenlees-May}. For the trivial subgroup $\{e\} \subset \mathbb{Z}_2\,$, the map $\beta :\; C_s(n) \to \Omega_K C_{s-1}(n)$ is the complex Bott map and therefore highly connected. Similarly, for the full group $\mathbb{Z}_2\,$, the map $\beta$ restricts to the real Bott map $R_s(n) \to \Omega_K R_{s-1}(n)$, which is also highly connected. The obstruction that there may be a mismatch between $\pi_0$ for $C_s(n)$ resp.\ $R_s(n)$ and $\Omega_K C_{s-1}(n)$
resp.\ $\Omega_K R_{s-1}(n)$, is avoided by the reasoning described in the proof of Proposition \ref{thm:6.1}.
\end{proof}

By specializing the result above to the case of $M = \mathrm{S}^{d_x ,\, d_k}$ (which is path-connected unless $d_x = d_k = 0$) and using $S M = S( \mathrm{S}^{d_x ,\, d_k}) = \mathrm{S}^{d_x + 1 ,\, d_k}$ we immediately get the following:
\begin{corollary}\label{cor:7.1}
There exists a bijection
\begin{equation*}
    [\mathrm{S}^{d_x ,\, d_k} , C_s (n) ]_\ast^{\mathbb{Z}_2} \stackrel{\sim}{\to} [\mathrm{S}^{d_x + 1 ,\, d_k} , C_{s-1} (n) ]_\ast^{\mathbb{Z}_2}
\end{equation*}
for $1 \le d_x + d_k \ll \dim C_s(n)_0\,$.
\end{corollary}

\subsection{Increasing the momentum-like dimension}

We now state and prove for all symmetry classes $s$ an analog of Theorem \ref{thm:7.1} for the step of increasing the momentum-like dimension:
\begin{theorem}\label{thm:7.2}
For a path-connected $\mathbb{Z}_2$-CW complex $M$ with $1\le\dim M \ll \dim C_s(n)_0$ there is, for any symmetry class $s$, a bijection
\begin{equation*}
    [M , C_s (n)]_\ast^{\mathbb{Z}_2} \simeq [\tilde{S} M , C_{s+1}(2n)]_\ast^{\mathbb{Z}_2} \,,
\end{equation*}
where $\tilde{S} M$ is the momentum-type suspension (i.e., the $\mathbb{Z}_2$-action on $\tilde{S}M$ reverses the suspension coordinate).
\end{theorem}
\begin{proof}
The idea of the proof is to first apply Theorem \ref{thm:7.1} repeatedly in order to adjust the symmetry index $s$ to be either 2 or 6 (for concreteness, we settle on the arbitrary choice of 2 here), then use the statement of Proposition \ref{thm:6.1} to increase the momentum-like dimension by one unit, and finally go to the symmetry index $s+1$ by retracing the initial steps.

To spell out the details, let $s = 2 + r$. Then Theorem \ref{thm:7.1} implies that there is a bijection
\begin{equation*}
    [M , C_{s} (n)]_\ast^{\mathbb{Z}_2} \simeq [S^r M , C_{2}(n) ]_\ast^{\mathbb{Z}_2}\,,
\end{equation*}
where $S^r M$ is the $r$-fold suspension of $M$. Here we made use of the fact that if $M$ is path-connected, then so is its suspension. We next apply Proposition \ref{thm:6.1} to obtain a bijection
\begin{equation*}
    [S^r M, C_{2} (n)]_\ast^{\mathbb{Z}_2} \simeq [\tilde{S} S^r M , C_{3}(2n)]_\ast^{\mathbb{Z}_2} \,.
\end{equation*}
Finally, we observe that $\tilde{S} S^r M = S^r \tilde{S} M$ and carry out $r$ applications of Thm.\ \ref{thm:7.1} in reverse:
\begin{equation*}
    [S^r \tilde{S} M , C_{3}(2n)]_\ast^{\mathbb{Z}_2} \simeq [\tilde{S} M , C_{s+1} (2n)]_\ast^{\mathbb{Z}_2} \,,
\end{equation*}
which completes the proof. Note that the cases $s = 0$ and $s = 1$ are included in this treatment via $s = 8$ and $s = 9$ by the 8-fold periodicity $C_{s+8}(16n) = C_s(n)$ and $R_{s+8}(16n) = R_s(n)$.
\end{proof}

Specializing once more to $M = \mathrm{S}^{d_x ,\, d_k}$ we have
\begin{corollary}\label{cor:7.2}
For $1 \le d_x + d_k\ll \dim C_s(n)_0\,$, there is a bijection
\begin{equation*}
    [\mathrm{S}^{d_x ,\, d_k} , C_s (n)]_\ast^{\mathbb{Z}_2} \simeq [\mathrm{S}^{d_x ,\, d_k+1} , C_{s+1}(2n)]_\ast^{\mathbb{Z}_2} \,.
\end{equation*}
\end{corollary}
\begin{proof}
Although this result follows directly from the more general one in Theorem \ref{thm:7.2}, it may be instructive to repeat the proof:
\begin{align*}
    [\mathrm{S}^{d_x ,\, d_k} , C_s (n)]_\ast^{\mathbb{Z}_2} &\simeq [\mathrm{S}^{d_x+s-2 ,\, d_k} , C_2 (n)]_\ast^{\mathbb{Z}_2} \cr
    &\simeq [\mathrm{S}^{d_x+s-2 ,\, d_k+1} , C_3 (2n) ]_\ast^{ \mathbb{Z}_2} \simeq [\mathrm{S}^{d_x ,\, d_k+1} , C_{s+1} (2n)]_\ast^{\mathbb{Z}_2} \,,
\end{align*}
as it clearly shows our chain of reasoning for a special case of importance in physics.
\end{proof}
\begin{remark}
For $d_x = d_k = 0$ the maps of Cor.\ \ref{cor:7.1} and Cor.\ \ref{cor:7.2} are injective but not surjective in general. In the latter case, some indication of this can be found in Section \ref{sect:4.3} and Remark \ref{rem:5.1} of the present paper and also in Sections 4.3 and 4.5 of \cite{RK-MZ-Nobel}.
\end{remark}

{}From the combination of the Corollaries \ref{cor:7.1} and \ref{cor:7.2}, the entries of Table \ref{table:periodictable} are determined by just specifying one column of entries for variable symmetry index $s$ but fixed values of the dimensions $d_x$ and $d_k$, subject to $d_x + d_k \ge 1$. For example, one may take $(d_x \,, d_k) = (1,0)$, in which case $[\mathrm{S}^{1,0} , C_s(n) ]_\ast^{\mathbb{Z}_2}$ is none other than the well-known fundamental group $\pi_1(R_s(n))$ (or $\pi_1(C_s(n))$ for the classes $A$ and $\AIII$).

\section{Stability bounds}\label{sect:8}
\setcounter{equation}{0}

In stating our theorems, \ref{thm:7.1} and \ref{thm:7.2}, we simply posed the qualitative condition $d = \dim M \ll \dim C_s(n)_0\,$, leaving their range of validity unspecified. To fill this quantitative void, we are now going to formulate precise conditions on $d$  (as a function of $\dim C_s(n)_0$) in order for the theorems to apply.

\subsection{Connectivity of inclusions}

In the definition of the space $C_s(n)$ with involution $\tau_s$ fixing the subspace $R_s(n)$, the number $n$ takes values in $m_s \mathbb{N}$ for a minimal integer $m_s \ge 1$ which depends on the symmetry class $s$. This restriction $n \in m_s \mathbb{N}$ stems from the requirement that $W = \mathbb{C}^{2n}$ must carry a representation of the Clifford algebra generated by the pseudo-symmetries $J_1, \ldots, J_s$. The numbers $m_s$ can be obtained by choosing the minimal parameters in Table \ref{table:CsRs} and are related to the ones found in \cite[Table~2]{ABS} and \cite[Table~V]{StoneEtAl}. The result is shown in the following list, which can be continued beyond $s=8$ by the relation $m_{s+8} = m_s/16$:
\begin{center}
\begin{tabular}{c|ccccccccc}
$s$	& 0 & 1 & 2 & 3 & 4 & 5 & 6 & 7 & 8 \\
\hline
$m_s$	& 1 & 2 & 2 & 4 & 4 & 4 & 4 & 8 & 16
\end{tabular}
\end{center}

Let the Clifford generators in the definition of $C_s(n)$ be denoted by $J_l$ and those of $C_s(m_s)$ by $J_l^\prime$ ($l = 1, \ldots, s$). For any symmetry class $s$, let a fixed element $A_0 \in R_s(m_s) \subset C_s(m_s)$ be given. We then have a natural inclusion
\begin{equation}
    i_s : \; C_s(n) \hookrightarrow C_s(n+m_s), \quad A \mapsto A \oplus A_0\,,
\end{equation}
where $C_s (n + m_s)$ is defined with Clifford generators $J_l \oplus J_l^\prime$ (for $l = 1, \ldots, s$). The map $i_s$ has the property of being equivariant with respect to the $\mathbb{Z}_2$-action on its image and domain:
\begin{equation}
    i_s(A)^\perp = A^\perp\oplus A_0^\perp = A^\perp\oplus A_0 = i_s(A^\perp) .
\end{equation}
In particular, its restriction $i_s^{\mathbb{Z}_2}$ to the subspace $C_s(n)^{\mathbb{Z}_2} = R_s(n)$ has image in $R_s(n + m_s)$.

The goal of this section is to prove the following theorem:
\begin{theorem}\label{theorem:8.1}
Given a path-connected $\mathbb{Z}_2$-CW complex $M$ and a number (of bands) $n = m_s\, r$ for some integer $r \in \mathbb{N}$, the induced map
\begin{equation*}
    (i_s)_\ast : \; [M , C_s(n)]_\ast^{\mathbb{Z}_2} \to [M, C_s (n + m_s) ]_\ast^{\mathbb{Z}_2}
\end{equation*}
is bijective if $\dim M < \nm$ and $\dim M^{\mathbb{Z}_2} < \nmZ\,$, and remains surjective under the weakened conditions $\dim M \le \nm$ and $\dim M^{\mathbb{Z}_2} \le \nmZ\,$. The values of $\nm$ and $\nmZ$ are given in the following table (where the complex classes are included by replacing the $\mathbb{Z}_2$-actions on $M$, $C_s(n)$ and $C_s(n+m_s)$ by the trivial one and omitting the conditions on $M^{\mathbb{Z}_2}$):
\begin{center}
\begin{tabular}{c|c|c|c|c}
s 	    &$C_s(m_s\, r)_0$	 &--     &$\nm$     &Case  \\
\hline
even	&$\mathrm{U}_{p+q} / \mathrm{U}_p \times \mathrm{U}_q$	
        &--    &$\min(2p+1,2q+1)$ 			 &(iv) \\
odd	    &$\mathrm{U}_r$  &--     &$2r$		 &(i) \\
\multicolumn{5}{c}{} \\
s       &$C_s(m_s\, r)_0$	&$C_s(m_s\, r)^{\mathbb{Z}_2}_0$
        &$\nmZ$	&Case \\
\hline
0       &$\mathrm{U}_{2r}/\mathrm{U}_r\times \mathrm{U}_r$	
        &$\mathrm{O}_{2r}/\mathrm{U}_{r}$	&$2r-1$	  & (ii)\\
1       &$\mathrm{U}_{2r}$	&$\mathrm{U}_{2r}/\mathrm{Sp}_{2r}$
        &$4r$				& (ii) \\
2       &$\mathrm{U}_{2p+2q}/\mathrm{U}_{2p}\times \mathrm{U}_{2q}$	
        &$\mathrm{Sp}_{2p+2q}/\mathrm{Sp}_{2p}\times \mathrm{Sp}_{2q}$	
        &$\min(4p+3,4q+3)$	&(iv) \\
3       &$\mathrm{U}_{2r}$	&$\mathrm{Sp}_{2r}$	 &$4r+2$     &(i) \\
4       &$\mathrm{U}_{2r}/\mathrm{U}_r\times \mathrm{U}_r$	
        &$\mathrm{Sp}_{2r}/\mathrm{U}_{r}$	     &$2r+1$     &(iii) \\
5       &$\mathrm{U}_r$		&$\mathrm{U}_{r}/\mathrm{O}_{r}$ &$r$
        &(iii)\\
6       &$\mathrm{U}_{p+q}/\mathrm{U}_p\times \mathrm{U}_q$	
        &$\mathrm{O}_{p+q}/\mathrm{O}_{p}\times \mathrm{O}_{q}$	
        &$\min(p,q)$		&(iv)\\
7       &$\mathrm{U}_r$		&$\mathrm{O}_{r}$	 &$r-1$		&(i)
\end{tabular}
\end{center}

\noindent For the complex symmetry classes with even $s$ (class $A$) as well as the real classes $s = 2, 6$ (classes $\AII$ and $\AI$), the single parameter $r$ is refined to $r = p + q$ in order to accommodate the possibility of the base point lying in different connected components of $C_s(m_s r)$.
\end{theorem}
\begin{remark}
The choice of $p$ and $q$ in the refinement $r=p+q$ amounts to choosing a Fermi energy and thus declaring the number of valence bands to be $p$ and the number of conduction bands to be $q$ (or vice versa).
\end{remark}
\begin{proof}
Since $M$ is path-connected and all maps are base-point preserving, we may replace $C_s(n) = C_s(m_s\, r)$ by its connected component (denoted by $C_s(m_s\, r)_0$ in the table) containing the base point $A_\ast\,$. Then, by applying the $\mathbb{Z}_2$-Whitehead Theorem, we obtain the desired statements provided that $i_s$ is $\nm$-connected and $i_s^{\mathbb{Z}_2}$ is $\nmZ$-connected (after restriction to $C_s(m_s\, r)_0$ resp.\ $C_s(m_s\, r)_0^{\mathbb{Z}_2}$), with numbers $\nm$ and $\nmZ$ that are yet to be determined. The latter is done in the remainder of the proof, where we distinguish between four cases.

\subsection*{Case (i)}

We start with the three rows attributed to case (i) in the tables. These enjoy the property of having Lie groups for their target spaces and we can make use of the following three fiber bundles:
\begin{align*}
    \mathrm{O}_r \hookrightarrow \mathrm{O}_{r+1} &\to \mathrm{O}_{r+1} / \mathrm{O}_r = \mathrm{S}^r , \\
    \mathrm{U}_r \hookrightarrow \mathrm{U}_{r+1} &\to \mathrm{U}_{r+1} / \mathrm{U}_r = \mathrm{S}^{2r+1} , \\
    \mathrm{Sp}_{2r} \hookrightarrow \mathrm{Sp}_{2r+2} &\to \mathrm{Sp}_{2r+2} / \mathrm{Sp}_{2r} = \mathrm{S}^{4r+3} ,
\end{align*}
each of which gives rise to a long exact sequence in homotopy.
By using $\pi_j (\mathrm{S}^d) = 0$ for $j < d$, we infer from these sequences the following values of $d_1$ and $d_2:$
\begin{alignat*}{2}
    \nmZ &= r-1 \quad &&\text{for} \quad \mathrm{O}_r \hookrightarrow \mathrm{O}_{r+1} \,, \\
    \nm  &= 2r \quad &&\text{for} \quad \mathrm{U}_r \hookrightarrow \mathrm{U}_{r+1} \,, \\
    \nmZ &= 4r+2 \quad &&\text{for} \quad \mathrm{Sp}_{2r} \hookrightarrow \mathrm{Sp}_{2r+2} \,.
\end{alignat*}

For the next two cases, (ii) and (iii), the target spaces are quotients $G_r / H_r$ with $G_r$ and $H_r$ being either an orthogonal, a unitary, or a symplectic group. The strategy in the following will be to apply the result of case (i) to the exact sequence associated to the fiber bundle
\begin{align*}
    H_r \hookrightarrow G_r \to G_r / H_r \,.
\end{align*}
We distinguish between case (ii) where the inclusion $G_r \hookrightarrow G_{r+1}$ is at most as connected as the inclusion $H_r \hookrightarrow H_{r+1}$, and case (iii) where it is more connected.

\subsection*{Case (ii)}

Let $G_r \hookrightarrow G_{r+1}$ be $m$-connected, where $m$ is less than or equal to the connectivity of $H_r \hookrightarrow H_{r+1}$. Then for all $j\in\mathbb{N}$ with $1\le j\le m-1$ there is the following commutative diagram:

\begin{center}
\begin{tikzcd}[column sep=0.85em]
\pi_{j}(H_r)  \arrow{r} \arrow{d}{\simeq}
& \pi_{j}(G_r)  \arrow{r} \arrow{d}{\simeq}
& \pi_{j}(G_r/H_r)  \arrow{r} \arrow{d}{(i_s^{\mathbb{Z}_2})_\ast}
& \pi_{j-1}(H_r)  \arrow{r} \arrow{d}{\simeq}
& \pi_{j-1}(G_r)  \arrow{d}{\simeq}\\
\pi_{j}(H_{r+1})  \arrow{r}
& \pi_{j}(G_{r+1})  \arrow{r}
& \pi_{j}(G_{r+1}/H_{r+1})  \arrow{r}
& \pi_{j-1}(H_{r+1})  \arrow{r}
& \pi_{j-1}(G_{r+1})
\end{tikzcd}
\end{center}

\noindent By the five-lemma $(i_s^{\mathbb{Z}_2} )_\ast$ is an isomorphism for all $j$ with $1\le j\le m-1$. The map $(i_s^{\mathbb{Z}_2})_\ast: \pi_0 (G_r/H_r) \to \pi_0 (G_{r+1}/H_{r+1})$ needs to be investigated separately. This task is facilitated by the fact that domain or codomain contain more than one element only for $G_r/H_r = \mathrm{O}_{2r} / \mathrm{U}_r$ ($s=0$). In that case, $\pi_0(\mathrm{O}_{2r}/\mathrm{U}_r) = \mathbb{Z}_2 = \pi_0(\mathrm{O}_{2r+2}/\mathrm{U}_{r+1})$. In the realization of $\mathrm{O}_{2r}/\mathrm{U_r}$ as an orbit, all elements $A \in \mathrm{O}_{2r}/\mathrm{U}_r$ may be written as $A = gA_\ast$ for a fixed $A_\ast \in \mathrm{O}_{2r}/\mathrm{U}_r$ and $g \in \mathrm{O}_{2r} \,$. The two connected components are distinguished by $\det(g) = \pm 1$ and we compute
\begin{align}
    i_0^{\mathbb{Z}_2}(A) = i_0^{\mathbb{Z}_2}(gA_\ast) = gA_\ast\oplus A_0 = (g\oplus\mathrm{Id})(A_\ast\oplus A_0) .
\end{align}
Since $\det(g\oplus\mathrm{Id}) = \det(g)$, it follows that the map $(i_0^{\mathbb{Z}_2})_\ast$ is a bijection on $\pi_0\,$.

By considering the part further left in the long exact sequences, we obtain the commutative diagram

\begin{center}
\begin{tikzcd}
\pi_{m}(G_r)  \arrow{r} \arrow{d}{\text{surjective}}
& \pi_{m}(G_r/H_r)  \arrow{r} \arrow{d}{(i_s^{\mathbb{Z}_2})_\ast}
& \pi_{m-1}(H_r)  \arrow{r} \arrow{d}{\simeq}
& \pi_{m-1}(G_r)  \arrow{d}{\simeq}\\
\pi_{m}(G_{r+1})  \arrow{r}
& \pi_{m}(G_{r+1}/H_{r+1})  \arrow{r}
& \pi_{m-1}(H_{r+1})  \arrow{r}
& \pi_{m-1}(G_{r+1})
\end{tikzcd}
\end{center}

\noindent Here, the second four-lemma leads to the conclusion that $(i_s^{\mathbb{Z}_2})_\ast$ is surjective. Combining all results, it follows that the inclusion $i_s^{\mathbb{Z}_2}$ is $m$-connected. Hence, $\nmZ = m$.

\subsection*{Case (iii)}

Consider now the complementary case, where $H_r \hookrightarrow H_{r+1}$ is $m$-connected with $m$ less than the connectivity of $G_r \hookrightarrow G_{r+1}$. We again use parts of the long exact sequence associated to the bundle $H_r \hookrightarrow G_r \to G_r / H_r$ in order to determine the connectivity of the inclusion $i_s^{\mathbb{Z}_2}$. Similar to the previous case, consider the following commutative diagram for $1\le j\le m$:

\begin{center}
\begin{tikzcd}[column sep=small]
\pi_{j}(H_r)  \arrow{r} \arrow{d}{\text{surjective}}
& \pi_{j}(G_r)  \arrow{r} \arrow{d}{\simeq}
& \pi_{j}(G_r/H_r)  \arrow{r} \arrow{d}{(i_s^{\mathbb{Z}_2})_\ast}
& \pi_{j-1}(H_r)  \arrow{r} \arrow{d}{\simeq}
& \pi_{j-1}(G_r)  \arrow{d}{\simeq}\\
\pi_{j}(H_{r+1})  \arrow{r}
& \pi_{j}(G_{r+1})  \arrow{r}
& \pi_{j}(G_{r+1}/H_{r+1})  \arrow{r}
& \pi_{j-1}(H_{r+1})  \arrow{r}
& \pi_{j-1}(G_{r+1})
\end{tikzcd}
\end{center}

\noindent Again, by the five-lemma $(i_s^{\mathbb{Z}_2})_\ast$ is an isomorphism for all $j$ with $1\le j\le m$. Notice that a difference to the previous case is the fact that the leftmost vertical map is only surjective. The extension to $j=0$, where the diagram above is not defined, is trivial here since all spaces involved are path-connected. Further to the left in the exact sequence we find the commutative diagram

\begin{center}
\begin{tikzcd}
\pi_{m+1}(G_r)  \arrow{r} \arrow{d}{\simeq}
& \pi_{m+1}(G_r/H_r)  \arrow{r} \arrow{d}{(i_s^{\mathbb{Z}_2})_\ast}
& \pi_{m}(H_r)  \arrow{r} \arrow{d}{\text{surjective}}
& \pi_{m}(G_r)  \arrow{d}{\simeq}\\
\pi_{m+1}(G_{r+1})  \arrow{r}
& \pi_{m+1}(G_{r+1}/H_{r+1})  \arrow{r}
& \pi_{m}(H_{r+1})  \arrow{r}
& \pi_{m}(G_{r+1})
\end{tikzcd}
\end{center}

\noindent The second four-lemma again implies that $(i_s^{\mathbb{Z}_2})_\ast$ is surjective. Therefore, in this case $i_s^{\mathbb{Z}_2}$ is $(m+1)$-connected, leading to the result that $\nmZ = m + 1$.

\subsection*{Case (iv)}

In the remaining three rows of the table, the target space has the form of a quotient $G_{p+q} / G_p \times G_q\,$. For the product of any two spaces $Y$ and $Z$, one has a natural isomorphism \cite{Hatcher-AT}
\begin{equation}
    \pi_j(Y\times Z) \simeq \pi_j(Y) \times \pi_j(Z)
\end{equation}
for all $j \ge 0$. Setting $Y = G_p$ and $Z = G_q$ and using this isomorphism, it follows that the inclusions $G_p\hookrightarrow G_{p+1}$ and $G_q \hookrightarrow G_{q+1}$ give rise to a commutative diagram

\begin{center}
\begin{tikzcd}
\pi_{j}(G_p\times G_q)  \arrow{r} \arrow{d}{\simeq}
& \pi_{j}(G_{p+1}\times G_{q+1}) \arrow{d}{\simeq}\\
\pi_{j}(G_p)\times\pi_j(G_q)  \arrow{r}
& \pi_{j}(G_{p+1})\times\pi_j(G_{q+1})
\end{tikzcd}
\end{center}
\noindent Hence, if $G_p\hookrightarrow G_{p+1}$ is $m$-connected and $G_q \hookrightarrow G_{q+1}$ $m^\prime$-connected, then $G_p \times G_q \hookrightarrow G_{p+1}\times G_{q+1}$ is $\min(m,m^\prime)$-connected. In particular, excluding the trivial case where $p=0$ or $q=0$, the inclusion $G_p \times G_q \hookrightarrow G_{p+1}\times G_{q+1}$ is always less connected than $G_{p+q} \hookrightarrow G_{p+q+2}$ and we can follow the steps of case (iii) with $H_r$ replaced by $G_p \times G_q\,$. As a result, $\nm = \min(m,m^\prime) + 1 = \min(m+1,m^\prime+1)$ (and the same for $\nmZ$). This completes the determination of $d_1$ and $d_2$ and, hence, the proof of the theorem.
\end{proof}

Specializing to the physically most relevant case of $M = \mathrm{S}^{d_x ,\, d_k}$, we obtain
\begin{corollary}\label{cor:8.1}
The induced map
\begin{equation*}
    (i_s)_\ast:\; [\mathrm{S}^{d_x ,\, d_k} , C_s(n)]_\ast^{\mathbb{Z}_2} \to [\mathrm{S}^{d_x ,\, d_k}, C_s (n+m_s)]_\ast^{\mathbb{Z}_2}
\end{equation*}
is bijective if $1 \le d_x + d_k < \nm$ and $d_x < \nmZ$ and surjective if $1 \le d_x + d_k \le \nm$ and $d_x \le \nmZ$.
\end{corollary}

Once the conditions for $(i_s)_\ast$ to be bijective are met, we are in the \textit{stable regime} mentioned in Section \ref{sect:equiv-class}. In that case, Theorem \ref{theorem:8.1} can be applied repeatedly to give a bijection
\begin{equation*}
    [M,C_s(n)]_\ast^{\mathbb{Z}_2} \to [M,C_s(\infty)]_\ast^{\mathbb{Z}_2} \,,
\end{equation*}
where $C_s(\infty)$ is the direct limit under $i_s\,$. This is the limit where $K$-theory applies for arbitrary path-connected $\mathbb{Z}_2$-CW complexes $M$ of finite dimension. For example, taking the complex class $A$ (even $s$ and trivial $\mathbb{Z}_2$-actions), the right-hand side is often written as $[M,\mathrm{BU}]_\ast$ and is in bijection with $\tilde K_\mathbb{C}(M)$.

Given a configuration space $M$, Theorem \ref{theorem:8.1} spells out the exact boundary to the stable regime of $K$-theory. However, as discussed in Section \ref{sect:equiv-class}, on the unstable side there is a further distinction in some symmetry classes between homotopy classes and isomorphism classes of vector bundles. This is the case for the real symmetry classes $s = 2$ (class $\AII$) and $s = 6$ (class $\AI$) as well as the complex symmetry class with even $s$ (class $A$), all three of which have been handled in case (iv) in the proof of Theorem \ref{theorem:8.1}. In these symmetry classes, there is a $\mathrm{U}_1$-symmetry leading to a decomposition of the fibers $A_k \in C_s(n)$ ($k\in M$) as $A_k = A^\mathrm{p}_k \oplus A^\mathrm{h}_k$, where $\mathrm{p}$ stands for particles or conduction bands and $\mathrm{h}$ for holes or valence bands. Recall from Section \ref{sect:2.3.1} that $A_k$ is already determined by $A^\mathrm{h}_k$. The bundle with fiber $A^\mathrm{h}_k$ over $k \in M$ is a Quaternionic vector bundle in the sense of \cite{dupont} (class $\AII$), a Real vector bundle in the sense of \cite{atiyah} (class $\AI$) or an ordinary complex vector bundle (class $A$) over $M$. In \cite{deNittis-AI} and \cite{deNittis-AII}, these vector bundles have been classified up to isomorphism for $M = \mathrm{S}^{d_x ,\, d_k}$ with $d_k \le 4$ and $d_x \le 1$. However, as was emphasized in Section \ref{sect:equiv-class}, in the situation at hand, where we have \textit{sub}vector bundles, isomorphism classes agree with homotopy classes only when $\dim A^\mathrm{p}_k$ is large compared to $\dim M$ and $\dim M^{\mathbb{Z}_2}$. It is the goal of the following to specify precisely what is meant by ``large'' in each of the three distinguished symmetry classes ($A$, $\AI$, $\AII$).

The inclusion $i_s$ adds dimensions to both $A^{\rm h}$ and $A^{\rm p}$, corresponding to the addition of valence bands \textit{and} conduction bands. This increases $p$ to $p+1$ and $q$ to $q+1$, as was considered in case (iv) of Theorem \ref{theorem:8.1} above. This inclusion can be refined by two separate inclusions: given a fixed $A_0 = A_0^{\rm p} \oplus A_0^{\rm h} \in C_s(m_s)$, one may add valence bands,
\begin{equation}
    i_s^{\rm h} :\; C_s(n) \hookrightarrow C_s (n+m_s/2) , \quad A \mapsto A \oplus A_0^{\rm h} \,,
\end{equation}
or conduction bands,
\begin{equation}
    i_s^{\rm p} :\; C_s(n) \hookrightarrow C_s (n + m_s/2) , \quad
    A \mapsto A \oplus A_0^{\rm p} \,.
\end{equation}
Since the situation is entirely symmetric, we will focus on $i_s^{\rm p}$ for the remainder of this section. In the realization of $C_s(n)$ and $R_s(n)$ as (unions of) homogeneous spaces, we have (restricting to one connected component as in Theorem \ref{theorem:8.1})
\begin{equation}\label{eq:conduction-band-inclusion}
    \begin{split}
    i_2^\mathrm{p} : \mathrm{U}_{2p+2q}/\mathrm{U}_{2p}\times \mathrm{U}_{2q}&\hookrightarrow \mathrm{U}_{2p+2q+2}/\mathrm{U}_{2p}\times \mathrm{U}_{2q+2}, \\
    (i_2^\mathrm{p})^{\mathbb{Z}_2} : \mathrm{Sp}_{2p+2q}/\mathrm{Sp}_{2p}\times \mathrm{Sp}_{2q}&\hookrightarrow \mathrm{Sp}_{2p+2q+2}/\mathrm{Sp}_{2p}\times \mathrm{Sp}_{2q+2},\\
    i_6^\mathrm{p} : \mathrm{U}_{p+q}/\mathrm{U}_{p}\times \mathrm{U}_{q}&\hookrightarrow \mathrm{U}_{p+q+1}/\mathrm{U}_{p}\times \mathrm{U}_{q+1}, \\
    (i_6^\mathrm{p})^{\mathbb{Z}_2} : \mathrm{O}_{p+q}/\mathrm{O}_{p}\times \mathrm{O}_{q}&\hookrightarrow \mathrm{O}_{p+q+1}/\mathrm{O}_{p}\times \mathrm{O}_{q+1}.
    \end{split}
\end{equation}
Note that the complex class $A$ may be included in this treatment by taking the inclusion $i_6^{\rm p}$ with $\mathbb{Z}_2$-action ignored.

All of these maps have the form
\begin{equation}
    G_{p+q} / G_p \times G_q \hookrightarrow G_{p+q+1} / G_p \times G_{q+1} \,.
\end{equation}
Since the inclusion $G_{p+q} \hookrightarrow G_{p+q+1}$ (for $p > 0$) is always more connected than the inclusion $G_q \hookrightarrow G_{q+1}\,$, we find ourselves in the setting of case (iii) in the proof of Theorem \ref{theorem:8.1}. Thus, if $G_{q}\hookrightarrow G_{q+1}$ is $m$-connected, then the inclusion $i_s^\mathrm{p}$ is $(m+1)$-connected, independent of the parameter $p$. Using the $\mathbb{Z}_2$-Whitehead Theorem once more, we can now prove the following:
\begin{corollary}\label{cor:8.2}
For a path-connected $\mathbb{Z}_2$-CW complex $M$, the induced map adding a conduction band,
\begin{equation*}
    (i_s^{\rm p})_\ast :\; [M , C_s(n)]_\ast^{\mathbb{Z}_2} \to [M , C_s(n + m_s/2)]_\ast^{\mathbb{Z}_2}\,,
\end{equation*}
is bijective or surjective according to the following table:
\begin{center}
\begin{tabular}{l|c|c}
            &bijective	     &surjective \\
\hline
class $A$   &$\dim M < 2q+1$    &$\dim M \le 2q+1$ \\
class $\AI$  &$\dim M < 2q+1$ and $\dim M^{\mathbb{Z}_2} < q$	
            &$\dim M \le 2q+1$ and $\dim M^{\mathbb{Z}_2} \le q$ \\
class $\AII$	&$\dim M < 4q+3$	 &$\dim M \le 4q+3$ \\
\end{tabular}
\end{center}
\end{corollary}
\begin{proof}
The proof is analogous to that of Theorem \ref{theorem:8.1}. For class $A$, the fact that $i_6^{\rm p}$ is $(2q+1)$-connected leads to the result. Proceeding to class $\AI$, we have a non-trivial $\mathbb{Z}_2$-action and therefore the additional requirement on $\dim M^{\mathbb{Z}_2}$ due to the fact that $(i_6^{\rm p})^{\mathbb{Z}_2}$ is $q$-connected. For class $\AII$, there is a slight change in the requirement for $\dim M$ because of the factor of two in the indices ($q\to 2q$, see Eq.\ \eqref{eq:conduction-band-inclusion}). Furthermore, since $(i_2^{\rm p})^{ \mathbb{Z}_2}$ is $(4q+3)$-connected while $i_2^{\rm p}$ is only $(4q+1)$-connected, the additional requirement on $\dim M^{\mathbb{Z}_2}$ is always fulfilled owing to $\dim M^{\mathbb{Z}_2} \le \dim M$.
\end{proof}

For $M = \mathrm{S}^{d_x ,\, d_k}$, the table in the Corollary simplifies to the following:
\begin{center}
\begin{tabular}{l|c|c}
            &bijective					 &surjective\\\hline
class $A$	&$d_x + d_k < 2q+1$	         &$d_x+d_k<2q+1$\\
class $\AI$	&$d_x + d_k < 2q+1$ and $d_x < q$				
            &$d_x + d_k \le 2q+1$ and $d_x \le q$ \\
class $\AII$	&$d_x+d_k<4q+3$	             &$d_x + d_k \le 4q+3$ \\
\end{tabular}
\end{center}

Notice the difference to the result in Theorem \ref{theorem:8.1}: rather than requiring both $p$ and $q$ to be large, only one of the two indices is required to be large. In fact, if the configuration space $M$ meets the conditions for bijectivity as listed above, the set of (equivariant) homotopy classes is in bijection with the set of isomorphism classes of rank-$p$ complex vector bundles (class $A$), rank-$p$ Real vector bundles (class $\AI$) and rank-$2p$ Quaternionic vector bundles (class $\AII$) with fixed fibers over the base point $k_\ast \in M$. Thus, we have derived the exact boundary, within the unstable regime, below which isomorphism classes of vector bundles may differ from homotopy classes.
\begin{remark}
The restrictive condition of holding the fiber over $k_\ast \in M$ fixed can be removed by applying the free version \cite{TomDieck} (instead of the one with fixed base points) of the $\mathbb{Z}_2$-Whitehead Theorem for a connected component of $C_s(n)$. The alternative offered by this free version is of relevance for the case of $M = \mathrm{T}^d$ \cite{KG14}, where fixing a base point may be unnatural from the physics perspective.
\end{remark}

We now list all potentially unstable cases violating the conditions of bijectivity in Corollary \ref{cor:8.1} and Corollary \ref{cor:8.2}. There are infinitely many possibilities in general if $d_x$ and $d_k$ are unrestricted. However, the physically most relevant cases are those with $d_k \le 3$ and $d_x < d_k\,$. The latter inequality is needed on physical grounds since the dimension of the defect is $d_k-d_x-1 \ge 0$. Table \ref{table:unstable} lists all cases which are not in the stable regime and may therefore differ from the stable classification.
\begin{table}
\begin{center}
\begin{tabular}{|c|c|c|c|c||c|c||c|}
\hline
complex 	& symmetry	 &\multicolumn{3}{c||}{$d_x=0$}&\multicolumn{2}{c||}{$d_x=1$}&\multicolumn{1}{c|}{$d_x=2$}\\\cline{3-8}
class $s$	& label	& $d_k=1$ & $d_k=2$ & $d_k=3$& $d_k=2$ & $d_k=3$&	 $d_k=3$\\\hline
even	& $A$		&&&$q=1$&$q=1$&$q=1$&$q\le2$\\
odd	& $\AIII$	&&$r=1$&$r=1$&$r=1$&$r\le2$&$r\le2$\\\hline
\multicolumn{7}{c}{}\\\hline
real		& symmetry	 &\multicolumn{3}{c||}{$d_x=0$}&\multicolumn{2}{c||}{$d_x=1$}&\multicolumn{1}{c|}{$d_x=2$}\\\cline{3-8}
class $s$	& label	& $d_k=1$ & $d_k=2$ & $d_k=3$& $d_k=2$ & $d_k=3$&	 $d_k=3$\\\hline
0 & $D$	&&&$r=1$&$r=1$&$r=1$&$r\le2$\\
1 & $\DIII$	&&&&&$r=1$&$r=1$\\
2 & $\AII$	&&&&&&\\
3 & $\CII$	&&&&&$r=1$&$r=1$\\
4 & $C$	&&&$r=1$&$r=1$&$r=1$&$r\le2$\\
5 & $\CI$	&&$r=1$&$r=1$&$r=1$&$r\le2$&$r\le2$\\
6 & $\AI$	&&&$q=1$&$q=1$&$q=1$&$q\le2$\\
7 & $\BDI$	&$r=1$&$r=1$&$r=1$&$r\le2$&$r\le2$&$r\le3$\\\hline
\end{tabular}
\vspace{10pt}
\caption{All potentially unstable cases for $d_k\le 3$ and $d_x < d_k\,$.} \label{table:unstable}
\end{center}
\end{table}

In Table \ref{table:unstable}, the cases in which isomorphism classes of vector bundles give the same classification as homotopy classes are included. In order to leave this intermediate regime (i.e.\ in order for the classification by homotopy to differ from that by isomorphism), the conditions for $q$ need to be met additionally by $p$. For instance, neither the stable classification nor the classification of complex vector bundles give any non-trivial topological phases for $d_k + d_x = 3$ in class $A$, but the Hopf insulator \cite{MRW} with $q = p = 1$ has a homotopy classification by $\mathbb{Z}$. It may also happen that non-trivial phases disappear in the unstable regime: in class $\AIII$ with $d_k + d_x = 3$, the stable $\mathbb{Z}$ classification is lost for $r = 1$ since $[\mathrm{S}^{d_x ,\, d_k},\mathrm{U}_1]_\ast = \pi_3 (\mathrm{U}_1) = 0$.

For $d_x = 0$, there is at most one exception for all entries which is neither in the stable regime nor in the ``intermediate'' regime of vector bundle isomorphism classes (since for the latter $p = q = 1$). The resulting change of the classification is shown in Table \ref{table:unstablechanges}. The changes in the first two rows for $d_k = 3$ are the ones described before. There are only two additional changes in the remainder of the table: for $s = 5$ (class $\CI$) all non-trivial topological phases vanish in dimension $d_k = 3$ for similar reasons as in class $\AIII$. However, there is an important change for $s = 4$ (class $C$) from trivial ($0$) to non-trivial ($\mathbb{Z}_2$) by a superconducting analog of the class-$A$ Hopf insulator. We are planning to discuss this class-$C$ Hopf superconductor in more detail in a future publication.

\begin{table}
\begin{center}
\begin{tabular}{c|c|c|c}
&\multicolumn{3}{|c}{$d_x=0$}\\\hline
Class	& $d_k=1$ & $d_k=2$ & $d_k=3$\\\hline
$A$	&&&$0\to\mathbb{Z}$\\
$\AIII$	&&$0\to0$&$\mathbb{Z}\to0$\\
\multicolumn{4}{c}{}\\
&\multicolumn{3}{|c}{$d_x=0$}\\\hline
Class $s$	& $d_k=1$ & $d_k=2$ & $d_k=3$\\\hline
0	&&&$0\to0$\\
1	&&&\\
2	&&&\\
3	&&&\\
4	&&&$0\to\mathbb{Z}_2$\\
5	&&$0\to0$&$\mathbb{Z}\to0$\\
6	&&&$0\to0$\\
7	&$\mathbb{Z}\to\mathbb{Z}$&$0\to0$&$0\to0$
\end{tabular}
\end{center}
\vspace{10pt}
\caption{Instances of potential changes to the stable classification in Table \ref{table:periodictable} which are captured neither by $K$-theory nor by isomorphism classes of vector bundles. Entries here are for the case of $r = q = 1$ in Table \ref{table:unstable}.} \label{table:unstablechanges}
\end{table}

\section{Appendix: proof of Proposition \ref{prop:2.1-new}}
\setcounter{equation}{0}

Recall the mathematical setting of $s \geq 4$ pseudo-symmetries $J_1, \ldots, J_s$ constraining the vector spaces $A_k$ by Eqs.\ (\ref{eq-mz:2.32}). We must show that the solutions $A_k$ of (\ref{eq-mz:2.32}) are in bijection with the solutions $a_k$ of Eqs.\ (\ref{eq:red-psym}) for the reduced system of generators $j_1, j_2, j_5, \ldots, j_s\,$.

Thus, let there be on $\widetilde{W} = \mathbb{C}^2 \otimes W$ a set of $s \geq 4$ orthogonal unitary operators $J_1, \ldots, J_s$ subject to the relations (\ref{Eq:Cliff-s}). Forming the two operators
\begin{equation*}
    K = \mathrm{i} J_1 J_2 J_3\, , \quad I = J_4 \,,
\end{equation*}
where $K$ is seen to be imaginary, let the shortened system $J_5, \ldots, J_s, I , K$ define complex and real classifying spaces $C_{s-2}(2n)$ and $R_{s-3,1}(2n)$ by the exact analog of Eqs.\ (\ref{eq:def-Cq2}) and (\ref{eq:def-Rq2}) with $s$ replaced by $s-4$. We then know from Proposition \ref{prop:2.1X} that there exist bijections
\begin{equation*}
    C_{s-2}(2n) \to C_{s-4}(n), \quad R_{s-3,1}(2n) \to R_{s-4}(n), \quad A_k \mapsto a_k \,,
\end{equation*}
which are given by intersecting $A_k$ with $E_{+1}(L)$ for $L \equiv J_1 J_2 J_3 J_4$ and applying the projector $\Pi = \frac{1}{2}(\mathrm{Id} - \mathrm{i} K)$ to obtain $a_k\,$. The spaces on the right-hand side are determined by Eqs.\ (\ref{eq:def-Cq}) and (\ref{eq:def-Rq}) via the system $j_l = L\, J_l \big\vert_W$ ($l = 4, \ldots, s$) defined as in (\ref{eq:4.15-mz}). Note that the restricted generators $j_l$ ($5 \leq l \leq s$) satisfy the third set of relations in (\ref{eq:red-CliffAlg}).

It remains to take into account the presence of the additional generators $J_1$, $J_2$, and $J_3\,$. These commute with $K = \mathrm{i} J_1 J_2 J_3$ and thus preserve the decomposition $\widetilde{W} = W_+ \oplus W_- = E_{+\mathrm{i}} (K) \oplus E_{-\mathrm{i}}(K)$. Simply restricting them to the subspace $W = W_+$ as
\begin{equation*}
    j_l = J_l \big\vert_W \quad (1 \leq l \leq 3) ,
\end{equation*}
we obtain the relations stated in the first and second line of Eqs.\ (\ref{eq:red-CliffAlg}). We also observe that the process of reduction to $W$ makes $j_3$ and $j_4$ redundant as $j_3 = j_2 j_1$ and $j_4 = -\mathrm{Id}_W$.

To prove Proposition \ref{prop:2.1-new}, we have to show that the conditions on $A_k$ due to the pseudo-symmetries $J_1, J_2$ are equivalent to the conditions on $a_k$ due to the symmetries $j_1, j_2\,$. The key observation here is that the pseudo-symmetry relations $J_l \, A_k = A_k^\mathrm{c}$ for $l = 1, 2, 3$ have the following refinement:
\begin{equation*}
    J_l (A_k \cap E_{+1}(L)) = A_k^\mathrm{c} \cap E_{-1}(L) \quad (1 \leq l \leq 3) ,
\end{equation*}
because $J_1$, $J_2$, $J_3$ anti-commute with $L = J_1 J_2 J_3 J_4$ and hence exchange the two eigenspaces $E_{+1}(L)$ and $E_{-1}(L)$. By applying the projector $\Pi = {\textstyle{\frac{1}{2}}} (\mathrm{Id} - \mathrm{i} K)$ to this equation in order to distill $a_k = \Pi (A_k \cap E_{+1}(L))$, it follows from Corollary \ref{cor:2.1} that
\begin{equation*}
    j_l \, a_k = (\Pi \circ J_l) (A_k \cap E_{+1}(L)) = \Pi (A_k^\mathrm{c} \cap E_{-1}(L)) = a_k \quad (1 \leq l \leq 3) ,
\end{equation*}
 owing to the fact that the operators $J_l$ ($l = 1, 2, 3$) preserve the decomposition $W = W_+ \oplus W_-\,$.

Conversely, the conditions $j_l \, a_k = a_k$ transform into the conditions $J_l \, A_k = A_k^\mathrm{c}$ ($l = 1, 2, 3$) by the inverse map $a_k \mapsto A_k$ given in (\ref{eq-mz:2.31}). This proves the said proposition.

\bigskip\noindent\textbf{Acknowledgment.} ---
Financial support by the Deutsche Forschungsgemeinschaft via the Sonderforschungsbereich/Transregio 12 is acknowledged. The senior author is supported by DFG grant ZI 513/2-1, the junior author by a scholarship of the Deutsche Telekom Stiftung and a stipend of the Bonn-Cologne Graduate School of Physics \& Astronomy. Both authors are grateful for the warm hospitality of the Erwin-Schr\"odinger International Institute for Mathematical Physics (Vienna) where this article reached its final form.

\end{document}